\newif\ifreport\reporttrue

\documentclass[journal]{IEEEtran}
\usepackage{setspace}
\usepackage{cite}
\usepackage{graphicx}
\usepackage{caption}
\usepackage{subcaption}
\usepackage[cmex10]{amsmath}
\usepackage{mathtools}
\usepackage{amsthm}
\usepackage{amsfonts}
\usepackage{amssymb}
\usepackage{amsmath}
\usepackage{algorithm}
\usepackage{algorithmic}
\usepackage{bm,xstring}
\usepackage{fixltx2e}
\usepackage{tikz,hyperref}
\usetikzlibrary{decorations.pathreplacing}

\newcommand{\age}{\Delta}

\DeclareMathOperator*{\argmax}{argmax}

\usepackage{color}

\def\blue{\color{blue}}

\newcommand{\ignore}[1]{}
\usepackage[singlelinecheck=false]{caption}

\newtheorem{lemma}{Lemma}

\newtheorem{theorem}{Theorem}
\newtheorem{corollary}{Corollary}
\theoremstyle{definition}

\newcommand{\EE}{\mathbb{E}}

\ifCLASSINFOpdf
\else
  \fi
\begin{document}
%
\title{Sampling and Remote Estimation for the Ornstein-Uhlenbeck Process through Queues: Age of Information and Beyond}
%
%
%
\author{Tasmeen Zaman Ornee, \emph{Student Member, IEEE} and Yin Sun, \emph{Senior Member, IEEE}\\
\thanks{Tasmeen Zaman Ornee and Yin Sun are with Dept. of ECE,  Auburn University, Auburn, AL (emails: tzo0017@auburn.edu, yzs0078@auburn.edu).}
\thanks{This work was supported in part by NSF grant CCF-1813050 and ONR grant N00014-17-1-2417.}
}

\maketitle
\begin{abstract}
Recently, a connection between the age of information and remote estimation error was found in a sampling problem of Wiener processes: If the sampler has no knowledge of the signal being sampled,  the optimal sampling strategy is to minimize the age of information; however, by exploiting causal knowledge of the signal values, it is possible to achieve a smaller estimation error. 
In this paper, we generalize the previous study by investigating a problem of sampling a stationary Gauss-Markov process named the Ornstein-Uhlenbeck (OU) process, where we aim to find useful insights for solving the problems of sampling more general signals. The optimal sampling problem is formulated as a constrained continuous-time Markov decision process (MDP) with an uncountable state space. We provide an exact solution to this MDP: The optimal sampling policy is a threshold policy on \emph{instantaneous estimation error} and the threshold is found. Further, if the sampler has no knowledge of the OU process, the optimal sampling problem reduces to an MDP for minimizing a \emph{nonlinear} age of information metric. The age-optimal sampling policy is a threshold policy on \emph{expected estimation error} and the threshold is found. In both problems, the optimal sampling policies can be computed by low-complexity algorithms (e.g., bisection search and Newton's method), and the curse of dimensionality is circumvented. These results hold for (i) general service time distributions of the queueing server and (ii) sampling problems both with and without a sampling rate constraint. 
Numerical results are provided to compare different sampling policies.
\end{abstract}


%
\IEEEpeerreviewmaketitle

\section{Introduction}
%
%
%
%
\IEEEPARstart{T}{imely} updates of the system state are of significant importance for state estimation and decision making in networked control and cyber-physical systems, such as UAV navigation, robotics control, mobility tracking, and environment monitoring systems. To evaluate the freshness of state updates, the concept of \emph{Age of Information}, or simply \emph{age}, was introduced to measure the timeliness of state samples received from a remote transmitter \cite{KaulYatesGruteser-Infocom2012, 139341, 3326507}. Let $U_t$ be the generation time of the freshest  received state sample at time $t$. The age of information, as a function
of $t$, is defined as $\Delta_t = t -U_t$, which is the time difference between the freshest samples available at the transmitter and receiver.

Recently, the age of information concept has received significant attention, because of the extensive applications of state updates among systems connected over communication networks. The states of many systems, such as UAV mobility trajectory and sensor measurements, are in the form of a  signal $X_t$, that may change slowly at some time and vary more dynamically later. Hence, the time difference described by the age $\Delta_t = t - U_t$ only partially characterize the variation $X_t -X_{U_t}$ of the system state, and the state update policy that minimizes the age of information does not minimize the state estimation error. This result was first shown in \cite{2020Sun}, where a  sampling problem of Wiener processes was solved and the optimal sampling policy was shown to have an intuitive structure. As the results therein hold only for signals that can be modeled as a Wiener process, one would wonder how to, and whether it is possible to, extend \cite{2020Sun} for handling more general signal models.  
 
In this paper, we generalize \cite{2020Sun} by exploring a problem of sampling an Ornstein-Uhlenbeck (OU) process $X_t$.
From the obtained results, we hope to find useful structural properties of the optimal sampler design that can be potentially applied to more general signal models.
The OU process $X_t$ is the continuous-time analogue of the well-known first-order autoregressive process, i.e., AR(1) process. The OU process is defined as the solution to the stochastic differential equation (SDE) \cite{PhysRev.36.823,Doob1942}
\begin{align}\label{eq_SDE}
dX_t =\theta (\mu-X_t) dt + \sigma dW_t, 
\end{align}
where $\mu$,  $\theta >0$,  and $\sigma >0$ are parameters and $W_{t}$ represents a Wiener process. It is  the only nontrivial continuous-time process that is stationary, Gaussian, and Markovian \cite{Doob1942}. Examples of first-order systems that can be described as the OU process include interest rates, currency exchange rates, and commodity prices (with modifications) \cite{Evans1994}, control systems such as node mobility in mobile ad-hoc networks, robotic swarms, and UAV systems \cite{41c4ac91,Kim2018}, and physical processes such as the transfer of liquids or gases in and out of a tank \cite{Nuno2011}.

As shown in Fig. \ref{fig_model}, samples of an OU process are forwarded to a remote estimator through a channel in a first-come, first-served (FCFS) fashion. The samples experience \emph{i.i.d.} random transmission times over the channel, which is caused by random sample size, channel fading, interference, congestions, and etc. For example, UAVs flying close to WiFi access points may suffer from long communication delay and instability issues, because they receive strong interference from the WiFi access points \cite{Vinogradov_2018}.
We assume that at any time only one sample can be served by the channel. The samples that are waiting to be sent are stored in a queue at the transmitter. Hence, the channel is modeled as an FCFS queue with \emph{i.i.d.} service times. The service time distributions considered in this paper are quite general: they are only required to have a finite mean. This queueing model is helpful to analyze the robustness of remote estimation systems with occasionally long transmission times.

\begin{figure}
\centering
\includegraphics[width=0.4\textwidth]{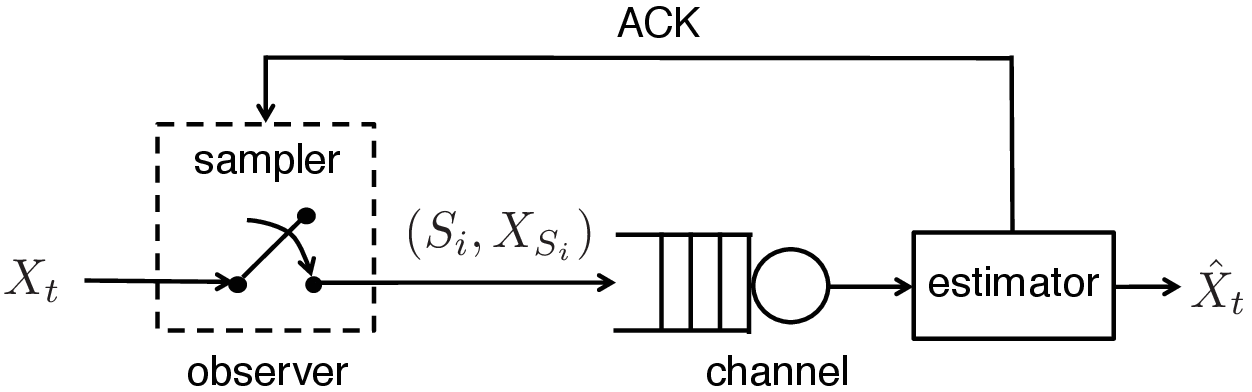}   
\caption{System model.}\vspace{-0.0cm}
\label{fig_model}
\vspace{-1em}
\end{figure}  
The estimator utilizes causally received samples to construct an estimate $\hat X_t$ of the real-time signal value $X_t$. The quality
of remote estimation is measured by  the time-average mean-squared estimation error, i.e.,
\begin{align}
{\mathsf{mse}} = \limsup_{T\rightarrow \infty}\frac{1}{T}\mathbb{E}\left[\int_0^{T} (X_t - \hat X_t)^2dt\right].
\end{align}
Our goal is to find the optimal sampling policy that minimizes $\mathsf{mse}$ by causally choosing the sampling times subject to a maximum sampling rate constraint. In practice, the cost (e.g., energy, CPU cycle, storage) for state updates increases with the average sampling rate. Hence, we are striking to find the optimum tradeoff between estimation error and update cost. In addition, the unconstrained problem is also solved. The contributions of this paper are summarized as follows:

\begin{itemize}

\item The optimal sampling problem for minimize the ${\mathsf{mse}}$ under a sampling rate constraint is formulated as a constrained continuous-time Markov decision process (MDP) with an  uncountable state space. Because of the curse of dimensionality, such problems are often lack of low-complexity solutions that are arbitrarily accurate. However, we were able to solve this MDP exactly: The optimal sampling policy is proven to be a threshold policy on \emph{instantaneous} estimation error, where the threshold is a non-linear function  $v(\beta)$ of a parameter $\beta$. The value of $\beta$ is equal to  the summation of the optimal objective value of the MDP and the optimal Lagrangian dual variable associated to the sampling rate constraint. If there is no sampling rate constraint, the Lagrangian dual variable is zero and hence $\beta$ is exactly the optimal objective value. 
Among the technical tools developed to prove this result is a free boundary method \cite{Peskir2006}, \cite{Bernt2000} for finding the optimal stopping time of the OU process. 

\item The optimal sampler design of Wiener process in \cite{2020Sun} is a limiting case of the above result. By comparing the optimal sampling policies of OU process and Wiener process, we find that the threshold function $v(\beta)$ changes according to the signal model, where the parameter $\beta$ is determined in the same way for both signal models.

\item Further, we consider a class of signal-agnostic sampling policies, where the sampling times are determined without using knowledge of the signal value of the observed OU process; the parameters of the OU process are known.
The optimal signal-agnostic sampling problem is equivalent to an MDP for minimizing the time-average of a nonlinear age function $p(\Delta_t) $, which has been solved recently in \cite{SunNonlinear2019}. The age-optimal sampling policy is a threshold policy on \emph{expected} estimation error, where the threshold function is simply $v(\beta) =\beta$ and the parameter $\beta$ is determined in the same way as above.  

\item  The above results hold for (i) general service time distributions with a finite mean and (ii) sampling problems both with and without a sampling rate constraint. Numerical results suggest that the optimal sampling policy is better than zero-wait sampling and the classic uniform sampling.

\end{itemize}

One interesting observation from these results is that the threshold function $v(\beta)$ varies with respect to the signal model and sampling problem, but the parameter $\beta$ is determined in the same way.

\subsection{Related Work}

The results in this paper are tightly related to recent studies on the age of information $\Delta_t$, e.g., \cite{KaulYatesGruteser-Infocom2012, SunSPAWC2018, SunNonlinear2019,Xiao2018, Ahmed2019, KamTIT2016,AgeOfInfo2016,KamTIT2018,yates2019age,
HeTIT2018,JooTON2018,BedewyJournal2017,BedewyJournal2017_2,multiflow18,KadotaINFOCOM2018,Talak:2018,Lu:2018,maatouk2020status,ZhouTCOM2019,ZhouTWC2020,NET-060, 8940930_1,Ahmed2020ArXiv}, which does not have a signal model. The average age and average peak age have been analyzed for various queueing systems in, e.g., \cite{KaulYatesGruteser-Infocom2012,KamTIT2016, KamTIT2018,yates2019age}. The optimality of the Last-Come, First-Served (LCFS) policy, or more generally the Last-Generated, First-Served (LGFS) policy, was established for various queueing system models in 
\cite{BedewyJournal2017,BedewyJournal2017_2,multiflow18, maatouk2020status}. Optimal sampling policies for minimizing non-linear age functions were developed in, e.g., \cite{AgeOfInfo2016, SunSPAWC2018, SunNonlinear2019,Ahmed2020ArXiv}. 
Age-optimal transmission scheduling of wireless networks were investigated in, e.g., \cite{HeTIT2018,JooTON2018,KadotaINFOCOM2018,Talak:2018,Lu:2018,ZhouTCOM2019,ZhouTWC2020}. 

On the other hand, this paper is also a contribution to the area of remote estimation, e.g., \cite{Hajek2008,Nuno2011,Rabi2012,nayyar2013,Basar2014,GAO201857,ChakravortyTAC2020}. In \cite{Rabi2012,Basar2014},  optimal sampling of Wiener processes was studied, where the transmission time from the sampler to the estimator is zero. Optimal sampling of OU processes was also considered in \cite{Rabi2012}, which is solved by discretizing time and using dynamic programming to solve the discrete-time optimal stopping problems. However, our optimal sampler of OU processes is obtained analytically. Remote estimation over several different channel models was recently studied in, e.g., \cite{GAO201857,ChakravortyTAC2020}. In \cite{Hajek2008,Nuno2011,Rabi2012,nayyar2013,Basar2014,GAO201857,ChakravortyTAC2020}, the optimal sampling policies were proven to be threshold policies.
Because of the queueing model, our optimal sampling policy has a different structure from those in \cite{Hajek2008,Nuno2011,Rabi2012,nayyar2013,Basar2014,GAO201857,ChakravortyTAC2020}. Specifically, in our optimal sampling policy,   sampling is suspended when the server is busy and is reactivated once the server becomes idle. In addition, we are able to characterize the threshold precisely.
The optimal sampling policy for the Wiener process in \cite{2020Sun} is a limiting case of ours. Remote estimation of the Wiener process with random two-way delay was recently considered in \cite{TsaiINFOCOM2020}. 
In \cite{GuoISIT2020}, a jointly optimal sampler, quantizer, and estimator design was found for a class of continuous-time Markov processes under a bit-rate constraint. A recent survey on remote estimation systems was presented in \cite{jog2019channels}.


\section{Model and  Formulation}

{}

\subsection{System Model}
We consider the remote estimation  system illustrated in Fig. \ref{fig_model}, where an observer takes samples from an OU process $X_t$ and forwards the samples to an estimator through a communication channel. The channel is modeled as a single-server FCFS queue with \emph{i.i.d.} service times. 
The system starts to operate at time $t=0$. The $i$-th sample is generated at time $S_i$ and is delivered to the estimator at time $D_i$ with a service time $Y_i$, which satisfy $S_i \leq S_{i+1}$, $S_i +Y_i \leq D_i$,  $D_i +Y_{i+1}\leq D_{i+1}$, and $0<\mathbb{E}[Y_i] < \infty$ for all $i$. Each sample packet $(S_i, X_{S_i})$ contains  the sampling time $S_i$ and the sample value $X_{S_i}$.
Let $U_t= \max\{S_{i}: D_{i} \leq t\}$ be the sampling time of the latest received sample at time  $t$. The \emph{age of information}, or simply  \emph{age}, at time  $t$ is defined as  \cite{KaulYatesGruteser-Infocom2012, 139341}
\begin{align}\label{eq_age}
\Delta_{t} = t-U_t = t- \max\{S_{i}: D_{i} \leq t\},
\end{align}
which is shown in Fig. \ref{fig:age1}. 
Because $D_i \leq D_{i+1}$, $\Delta_t$ can be also expressed as
\begin{align}\label{eq_age2}
\Delta_t = t- S_i, ~\text{if}~t\in[D_i,D_{i+1}),~i=0,1,2,\ldots
\end{align}
The initial state of the system is assumed to satisfy $S_0 = 0$, $D_0 = Y_0$, $X_0$ and $\age_0$ are finite constants. The parameters $\mu$, $\theta$, and $\sigma$ in \eqref{eq_SDE} are known at both the sampler and estimator.

Let $I_t \in\{0,1\}$ represent the idle/busy state of the server at time $t$.  We assume that whenever a sample is delivered, an acknowledgement is sent back to the sampler with zero delay. By this, the idle/busy state $I_t$ of the server is known at the sampler. Therefore, the information
that is available at the sampler at time $t$ can be expressed as $\{X_s, I_s: 0\leq s\leq t\}$.

\begin{figure}
\centering
\begin{tikzpicture}[scale=0.23]
\draw [<-|] (0,11)  -- (0,0) -- (14.5,0);
\draw [|->] (15,0) -- (30.5,0) node [below] {\small$t$};
\draw (-2,12) node [right] {\small$\age_t$};
\draw
(0,0) node [below] {\small$S_0$}
(8,0) node [below] {\small$S_1$}
(17,0) node [below] {\small$S_{j-1}$}
(24,0) node [below] {\small$S_{j}$};
\fill
(8,0)  circle[radius=4pt]
(17,0)  circle[radius=4pt]
(24,0)  circle[radius=4pt]
(0,0)  circle[radius=4pt]
(4,0)  circle[radius=4pt]
(11,0)  circle[radius=4pt]
(20,0)  circle[radius=4pt]
(27,0)  circle[radius=4pt];
\draw
(4,0) node [below] {\small$D_0$}
(11,0) node [below] {\small$D_1$}
(21,0) node [below] {\small$D_{j-1}$}
(27,0) node [below] {\small$D_{j}$};
\draw[ thick, domain=0:4] plot (\x, {\x+3})  -- (4, {4});
 \draw [ thick, domain=4:11] plot (\x, {\x})  -- (11, {3});
\draw[ thick, domain=11:14] plot (\x, {\x-8});
\draw[ thick, domain=16:20] plot (\x, {\x-14}) -- (20, {3});
\draw[ thick, domain=20:27] plot (\x, {\x-17}) -- (27, {3});
\draw[ thick, domain=27:29] plot (\x, {\x-24});
\draw[  thin,dashed,  domain=0:4] plot (\x, {\x})-- (4, 0);
\draw[  thin,dashed,  domain=8:11] plot (\x, {\x-8})-- (11, 0);
\draw[  thin,dashed,  domain=17:20] plot (\x, {\x-17})-- (20, 0);
\draw[  thin,dashed,  domain=24:27] plot (\x, {\x-24})-- (27, 0);
\end{tikzpicture}
\caption{Evolution of the age $\age_t$ over time.}
\label{fig:age1}
\vspace{-0.5cm}
\end{figure}

\subsection{Sampling Policies}
In causal sampling policies, each sampling time $S_i$ is determined based on the up-to-date information that is available at the sampler, without using any future information. In probability theory, such sampling times are represented by \emph{stopping times}. 

To define \emph{stopping time} precisely, the concepts of $\sigma$-field and filtration are needed. Let us define the $\sigma$-field
\begin{align}
\mathcal{N}_t &= \sigma(X_s, I_s: 0\leq s\leq t), \nonumber
\end{align}
which is the set of events whose occurrence are determined by the realization of the process $\{X_s, I_s, 0\leq s\leq t\}$ up to time $t$. A filtration is a non-decreasing sequence of $\sigma$-fields. Our analysis requires a strong Markov property, which is satisfied when the filtration is right-continuous. Define
\begin{align}
\mathcal{N}_t^+ = \cap_{s>t}\mathcal{N}_s, 
\end{align}
then  $\{\mathcal N_t^+, t\geq 0\}$ is a right-continuous filtration of the information process $\{X_s, I_s, t \geq 0\}$ \cite{Durrettbook10}. In a causal sampling policy, each sampling time is a stopping time with respect to $\{\mathcal N_t^+, t\geq 0\}$, i.e.,
\begin{align}\label{eq_sampling_policy}
\{S_{i}\leq t\} \in \mathcal{N}_t^+,~\forall t\geq0.
\end{align}
In other words, whether sample $i$ has been generated by time $t$ (i.e., whether \{$S_i \leq t$\} or \{$S_i > t$\}) is determined by the realization of the process $\{X_s, I_s, 0\leq s\leq t\}$ up to time $t$.


Let $\pi = (S_1, S_2, . . . )$ represent a sampling policy. We use $\Pi$ to represent the set of \emph{causal} sampling policies that satisfy two conditions: (i) Each sampling policy $\pi\in\Pi$ satisfies \eqref{eq_sampling_policy} for all $i$. (ii) The sequence of inter-sampling times $\{T_i = S_{i+1}-S_i, i=0,1,\ldots\}$ forms a \emph{regenerative process} \cite[Section 6.1]{Haas2002}: There exists an increasing sequence  $0\leq {k_1}<k_2< \ldots$ of almost surely finite random integers such that the post-${k_j}$ process $\{T_{k_j+i}, i=0,1,\ldots\}$ has the same distribution as the post-${k_0}$ process $\{T_{k_0+i}, i=0,1,\ldots\}$ and is independent of the pre-$k_j$ process $\{T_{i}, i=0,1,\ldots, k_j-1\}$; further, we assume that $\mathbb{E}[{k_{j+1}}-{k_j}]<\infty$, $\mathbb{E}[S_{k_{1}}]<\infty$, and $0<\mathbb{E}[S_{k_{j+1}}-S_{k_j}]<\infty, ~j=1,2,\ldots$\footnote{We will optimize $\limsup_{T\rightarrow \infty}\mathbb{E}[\int_0^{T} (X_t - \hat X_t)^2dt]/T$, but operationally a nicer criterion is $\limsup_{i\rightarrow \infty}\mathbb{E}[\int_{0}^{D_i} (X_t - \hat X_t)^2dt]/{\mathbb{E}[D_i]}$. These criteria are corresponding to two definitions of ``average cost per unit time'' that are widely used in the literature of semi-Markov decision processes. These two criteria are equivalent, if $\{T_1,T_2,\ldots\}$ is a regenerative process, or more generally, if $\{T_1,T_2,\ldots\}$ has only one ergodic class. If not condition is imposed, however, they are different. The interested readers are referred to \cite{Ross1970,Mine1970,Hayman1984,Feinberg1994,Bertsekas2005bookDPVol1} for more discussions.} 

From this, 
we can obtain that $S_i$ is finite almost surely for all $i$.
We assume that the OU process $\{X_t, t \geq 0\}$ and the service times $\{Y_i, i= 1, 2, \dots\}$ are mutually independent, and do not change according to the  sampling policy. 

A sampling policy $\pi\in\Pi$ is said to be \emph{signal-agnostic} (\emph{signal-aware}), if $\pi$ is (not necessarily) independent of $\{X_t,t\geq 0\}$. 
Let $\Pi_{\text{signal-agnostic}}\subset \Pi$ denote the set of signal-agnostic sampling policies, defined as   
\begin{align}
\!\!\Pi_{\text{signal-agnostic}} \!=\! \{\pi\!\in\Pi: \pi \text{ is independent of }\{X_t, t\geq 0\}\}.\!\!\!
\end{align}

\subsection{MMSE Estimator}

According to \eqref{eq_sampling_policy}, $S_i$ is a finite stopping time. By using \cite[Eq. (3)]{Maller2009} and the strong Markov property of the OU process \cite[Eq. (4.3.27)]{Peskir2006},  $X_t$ is expressed as 
\begin{align}\label{eq_process}
X_t 
= &X_{S_i} e^{-\theta (t-S_i)}+ \mu\big[1-e^{-\theta (t-S_i)} \big] \nonumber\\
&+ \frac{\sigma}{\sqrt{2\theta}}e^{-\theta (t-S_i)} W_{e^{2 \theta (t-S_i)}-1}, \text{ if }t \in [S_i,\infty).
\end{align}

At any time $t\geq0$, the estimator uses causally received samples to construct an estimate $\hat X_t$ of the real-time signal value $X_t$. The information available to the estimator consists of two parts: (i) $M_t = \{(S_i, X_{S_i}, D_i): D_i \leq t\}$, which contains the sampling time $S_i$, sample value $X_{S_i}$, and delivery time $D_i$ of the samples that have been delivered by time $t$ and (ii) the fact that no sample has been received after the last  delivery time $\max\{D_i: D_i\leq t\}$. Similar to \cite{Rabi2012,SOLEYMANI20161, 2020Sun}, we assume that the estimator neglects the second part of information.\footnote{We note that this assumption can be removed by considering a joint sampler and estimator design problem. Specifically, it was shown in \cite{Hajek2008,Nuno2011,nayyar2013,GAO201857,ChakravortyTAC2020} that when the sampler and estimator are jointly optimized in discrete-time systems, the optimal estimator has the same expression no matter with or without the second part of information. As pointed out in \cite[p. 619]{Hajek2008}, such a structure property of the MMSE estimator can be also established for continuous-time systems. The goal of this paper is to find the closed-form expression of the optimal sampler under this assumption. The remaining task of finding the jointly optimal sampler and estimator design can be done by further using the majorization techniques developed in \cite{Hajek2008,Nuno2011,nayyar2013,GAO201857,ChakravortyTAC2020}; see \cite{GuoISIT2020} for a recent treatment on this task.}
\ifreport
Then, as shown in Appendix \ref{app_MMSE}, the minimum mean square error (MMSE)  estimator is determined by
\fi
\begin{align}\label{eq_esti}
\hat{X}_{t}  =  \mathbb{E}[{X}_t | M_t ] =& X_{S_i} e^{-\theta (t-S_i)}+ \mu\big[1-e^{-\theta (t-S_i)} \big], \nonumber\\
                &\text{if}~t\in[D_i,D_{i+1}),~i=0,1,2,\ldots
\end{align}
Hence, the estimation error of the MMSE estimator is
\begin{align} \label{eq_esti_2}
X_t - \hat X_t =& \frac{\sigma}{\sqrt{2\theta}}e^{-\theta t} W_{e^{2 \theta t}-1}, \nonumber\\
&\text{if}~ t\in[D_i, D_{i+1}), i = 0,1,2,\ldots\!\!\!\!\!\!\!\!\!\!\!\!
\end{align}
\subsection{Problem Formulation}


The goal of this paper is to find the optimal sampling policy that minimizes the mean-squared estimation error subject to an average sampling-rate constraint, which is formulated as the following problem: 
\begin{align}\label{eq_DPExpected}
{\mathsf{mse}}_{\text{opt}}=
\inf_{\pi\in\Pi}~& \limsup_{T\rightarrow \infty}\frac{1}{T}\mathbb{E}\left[\int_0^{T} (X_t - \hat X_t)^2dt\right] \\
~\text{s.t.}~&\liminf_{n\rightarrow \infty} \frac{1}{n} 
\mathbb{E}\left[\sum_{i=1}^n (S_{i+1}-S_i)\right]\geq \frac{1}{f_{\max}},\label{eq_constraint}
\end{align}
where ${\mathsf{mse}}_{\text{opt}}$ is the optimum value of \eqref{eq_DPExpected} and  $f_{\max}$ is the maximum allowed sampling rate. When $f_{\max} = \infty$, this problem becomes an unconstrained problem.   

%
%

\section{Main Results}\label{sec_main_result}
\subsection{Signal-aware Sampling without Rate Constraint}
Problem \eqref{eq_DPExpected} is a constrained continuous-time MDP with a continuous state space. However, we found an exact solution to this problem. 
To present this solution, let us consider an OU process $O_t$ with the initial state $O_t = 0$ and parameter $\mu = 0$. According to \eqref{eq_process}, $O_t$ can be expressed as
\begin{align}\label{eq_OU}
O_t = \frac{\sigma}{\sqrt{2\theta}}e^{-\theta t} W_{e^{2 \theta t}-1}.
\end{align} 
Define
\begin{align}\label{eq_mse_Yi}
& \mathsf{mse}_{Y_i} = \mathbb{E}[O_{Y_i}^2] = \frac{\sigma^2}{2\theta} \mathbb{E} [1 - e^{-{2 \theta {Y_i}}}], \\
& \mathsf{mse}_{\infty} = \mathbb{E}[O_{\infty}^2] = \frac{\sigma^2}{2\theta}.\label{eq_mse_infty}
\end{align}
In the sequel, we will see that $\mathsf{mse}_{Y_i}$ and $\mathsf{mse}_{\infty}$ are the lower and upper bounds of $\mathsf{mse}_\text{opt}$, respectively. According to \eqref{eq_esti_2} and \eqref{eq_OU}-\eqref{eq_mse_infty}, $\mathsf{mse}_{Y_i}$ represents the estimation error when the estimation is made based on a sample that was generated $Y_i$ seconds ago, and $\mathsf{mse}_{\infty}$ represents the estimation error for the case that no sample has been delivered to the estimator before.
We will also need to use the function\footnote{If $x=0$, $G(x)$ is defined as its right limit $G(0)=\lim_{x\rightarrow 0^+}G(x)= 1$.}
\begin{align}\label{eq_g1}
& G(x) =\!\! \frac{e^{x^2}}{x} {\int_{0}^{x} e^{-{t^2}} dt}\!\! \thinspace = \thinspace \frac{e^{x^2}}{x} \frac{\sqrt{\pi}}{2} \thinspace {\text{erf}}(x), ~x\in[0,\infty),
\end{align}
where $\text{erf}(\cdot)$ is the  error function \cite{2007247}, defined as
\begin{align}\label{eq_erf}
{\text{erf}}(x) = \frac{2}{\sqrt \pi} \int_0^x e^{-t^2} dt.
\end{align}





We first consider the unconstrained optimal sampling problem, i.e., $f_{\max} = {\infty}$, such that the rate constraint \eqref{eq_constraint} can be removed. In this scenario, the optimal sampler is provided in the following theorem. 

\begin{theorem}\label{thm1} (Sampling without Rate Constraint).
If $f_{\max} = {\infty}$ and the $Y_i$'s are i.i.d. with $0<\mathbb{E}[Y_i] < \infty$, then $(S_1(\beta),S_2(\beta),\ldots)$  with a parameter $\beta$ is an optimal solution to \eqref{eq_DPExpected}, 
%
where 
\begin{align}\label{eq_opt_solution}
S_{i+1} (\beta)= \inf \left\{ t \geq D_i(\beta):\! \big|X_t - \hat X_t\big| \!\geq\! {v}(\beta)\right\},
\end{align}
$D_i (\beta)= S_i (\beta)+ Y_i$, ${v}(\beta)$ is defined by 
\begin{align}\label{eq_threshold}
{v}(\beta) = \frac{\sigma}{\sqrt \theta}G^{-1}\left(\frac{\mathsf{mse}_{\infty} - \mathsf{mse}_{Y_i}}{\mathsf{mse}_{\infty} - {\beta}}\right),
\end{align}
$G^{-1}(\cdot)$ is the inverse function of $G(\cdot)$ in \eqref{eq_g1} and $\beta$ is the unique root of
\begin{align}\label{thm1_eq22}
\!\!\! \mathbb{E}\left[\int_{D_i(\beta)}^{D_{i+1}(\beta)}(X_t-\hat X_t)^2dt\right]\! - \!{\beta} {\mathbb{E}[D_{i+1}(\beta)\!-\!D_i(\beta)]} \!=\! 0.\!\!\! 
\end{align}
The optimal objective value to \eqref{eq_DPExpected} is given by
\emph{\begin{align}\label{thm1_eq23}
\mathsf{mse}_{\text{opt}} = \frac{\mathbb{E}\left[\int_{D_i(\beta)}^{D_{i+1}(\beta)}\! (X_t-\hat X_t)^2dt\right]}{\mathbb{E}[D_{i+1}(\beta)\!-\!D_i(\beta)]}. {\noindent}
\end{align}}
{\!\!Furthermore}, $\beta$ is exactly the optimal value to \eqref{eq_DPExpected}, i.e., \emph{${\beta} = {\mathsf{mse}}_{\text{opt}}$}.
\end{theorem}

The proof of Theorem \ref{thm1}  is explained in Section \ref{sec_proof}. The optimal sampling policy in Theorem \ref{thm1} has a nice structure. Specifically, the $(i+1)$-th sample is taken at the earliest time $t$ satisfying two conditions: {(i)} The $i$-th sample has already been delivered by time $t$, i.e., $t\geq D_i(\beta)$, and {(ii)} the estimation error $|X_t - \hat X_t|$ is no smaller than a pre-determined threshold ${v}(\beta)$, where $v(\cdot)$ is a non-linear function defined in \eqref{eq_threshold}. In Section \ref{sec_proof}, it is shown that \emph{$\mathsf{mse}_{Y_i}  \leq \beta < \mathsf{mse}_{\infty}$.}
Further, it is not hard to show that $G(x)$ is strictly increasing on $[0, \infty)$ and $G(0) = 1$. Hence, its inverse function $G^{-1}(\cdot)$ and the threshold $v(\beta)$ are properly defined and $v(\beta)\geq 0$.

\subsubsection{Three Algorithms for Solving \eqref{thm1_eq22}} \label{algo_sec}
We now present three algorithms for computing the root of \eqref{thm1_eq22}. Because the $S_i(\beta)$'s are stopping times, numerically calculating the expectations in \eqref{thm1_eq22} appears to be a difficult task. Nonetheless, this challenge can be solved by resorting to the following lemma, which is obtained by using Dynkin's formula \cite[Theorem 7.4.1]{Bernt2000} and the optional stopping theorem.

\begin{figure}[!t]
\vspace{-3mm}
	\centering
	\includegraphics[width=0.5\textwidth]{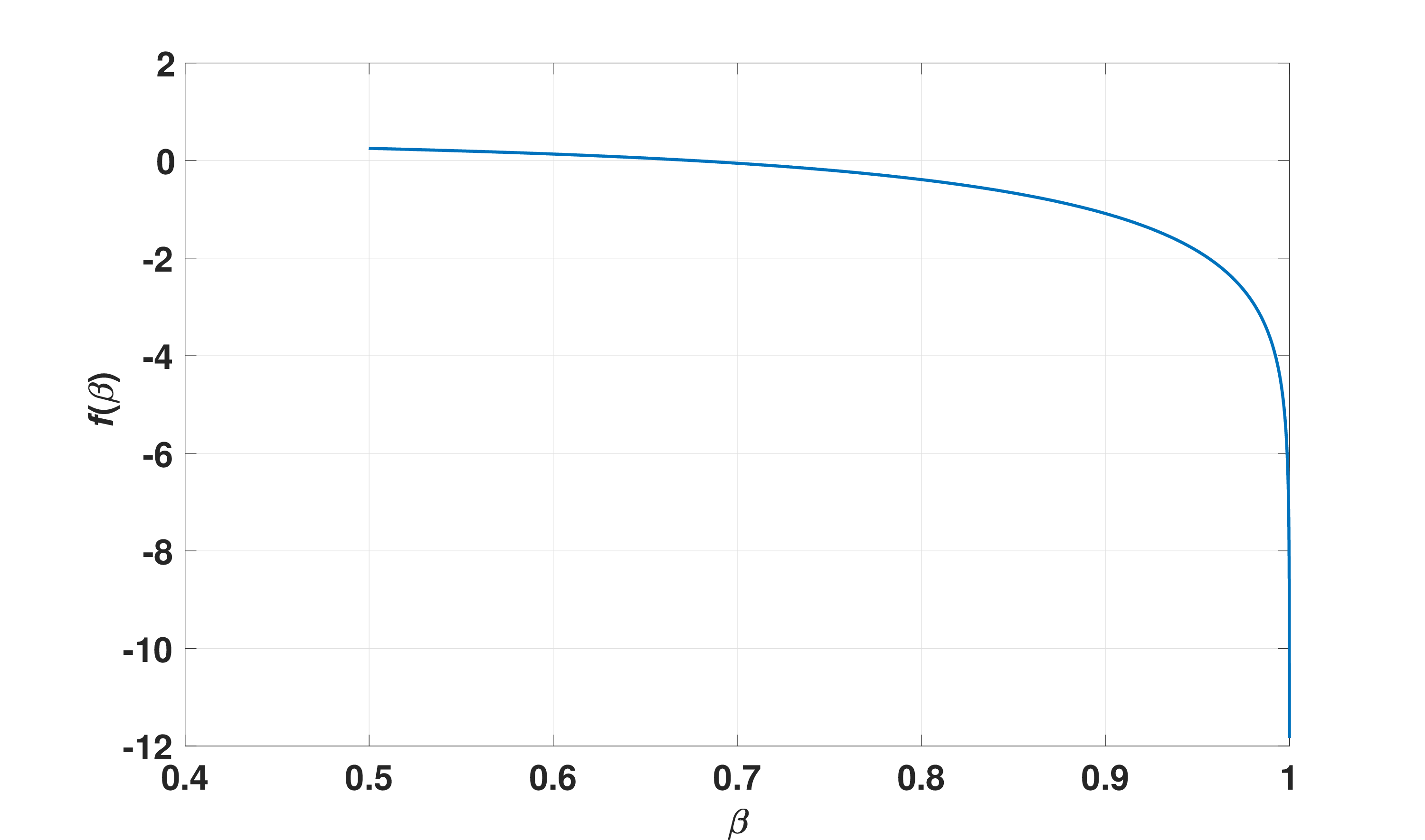}   
	\caption{$f(\beta)$ in \eqref{two_root_1} for \emph{i.i.d.} exponential service time with $\mathbb{E}[Y_i] = 1$, where the parameters of the OU process are $\sigma=1$ and $\theta=0.5$. For these parameters, $\mathsf{mse}_{Y_i} = 0.5$ and $\mathsf{mse}_{\infty}=1$.}
	\label{f_one_root}
	\vspace{-3mm}
\end{figure} 
\begin{lemma}\label{eq_expectation}
In Theorem \ref{thm1}, it holds that
\begin{align}\label{eq_expectation1}
\!\!\!\!\!\! & \mathbb{E}[D_{i+1}(\beta) - D_i(\beta)] \nonumber\\
=& \mathbb{E}[\max\{R_1 (v(\beta)) - R_1 (O_{Y_i}), 0\}] + \mathbb{E}[Y_i],\!\!\!\!\!\! 
\end{align}
\begin{align}
\!\!\!\!\!\! & {\mathbb{E}\left[\int_{D_i(\beta)}^{D_{i+1}(\beta)}\! (X_t-\hat X_t)^2dt\right]} \nonumber\\
= & \mathbb{E}[\max\{R_2 (v(\beta)) - R_2 (O_{Y_i}), 0\}] \nonumber\\
& + \mathsf{mse}_{\infty} [\mathbb{E} (Y_i)-\gamma] + \mathbb{E} \left[\max\{v^2(\beta), O_{Y_i}^2\}\right] \gamma, \!\!\!\!\!\!\label{eq_expectation2}
\end{align}
where
\begin{align}
& \gamma = \frac{1}{2 \theta} \mathbb{E} [1 - e^{-2 \theta Y_i}] \label{eq_gamma1}, \\
& R_1 (v) = {\frac{v^2}{\sigma^2}} \thinspace {}_2F_2\left(1,1;\frac{3}{2},2;\frac{\theta}{\sigma^2}v^2\right), \label{eq_R_1}\\
& R_2 (v) = -\frac{v^2}{2\theta} + {\frac{v^2}{2\theta}} \thinspace {}_2F_2\left(1,1;\frac{3}{2},2;\frac{\theta}{\sigma^2}v^2\right).\label{eq_R_2}
\end{align}
\end{lemma}

\ifreport
\begin{proof}
See Appendix \ref{eq_expectation_1}.
\end{proof}
\fi
In \eqref{eq_R_1} and \eqref{eq_R_2}, we have used the generalized hypergeometric function,  which is defined by \cite[Eq. 16.2.1]{200720017}
\begin{align}
& {}{_pF_q} (a_1, a_2, \cdots, a_p; b_1, b_2, \cdots b_q; z)\nonumber\\
=& {\sum_{n=0}^{\infty}} {\frac{{(a_1)_n} {(a_2)_n} \cdots{(a_p)_n}}{{(b_1)_n}{(b_2)_n} \cdots {(b_p)_n}}}{\frac{z^n}{n!}},
\end{align}
where
\begin{align}
& (a)_{0}=1,\\
& (a)_{n}=a(a+1)(a+2) {\cdots} (a+n-1),~ n \geq 1.
\end{align} 
Using Lemma 1, the expectations in \eqref{thm1_eq22} can be evaluated by Monte Carlo simulations of scalar random variables $O_{Y_i}$ and $Y_i$, which is much simpler than directly simulating the entire random process $\{O_t, t\geq 0\}$. 

For notational simplicity, we rewrite \eqref{thm1_eq22} as 
\begin{align} \label{two_root_1}
f(\beta) = {f_1(\beta) - {\beta} f_2(\beta)} = 0,
\end{align} 
where $f_1(\beta)={\mathbb{E}\left[\int_{D_i(\beta)}^{D_{i+1}(\beta)}\! (X_t-\hat X_t)^2dt\right]}$ and $f_2(\beta)={\mathbb{E}[D_{i+1}(\beta)\!-\!D_i(\beta)]}$.
The function $f(\beta)$ has several nice properties, which are asserted in the following lemma and illustrated in Fig. \ref{f_one_root}.

\begin{lemma}\label{beta_unique}
The function $f(\beta)$ has the following properties:

(i) $f(\beta)$ is concave, continuous, and strictly decreasing in $\beta$,

(ii) $f(\mathsf{mse}_{Y_i}) > 0$ and $\displaystyle \lim_{{{\beta} {\to} {\mathsf{mse}^{-}_{\infty}}}} f(\beta) = -{\infty}$.

\end{lemma}

\begin{proof}
See Appendix \ref{beta_unique_proof}.
\end{proof}

The uniqueness of the root of $f(\beta)$ follows immediately from Lemma \ref{beta_unique}.

%

Because $f(\beta)$ is decreasing and has a unique root, one can use a bisection search method to solve \eqref{thm1_eq22}, which is illustrated in Algorithm \ref{alg1}. The bisection search method has a globally linear convergence speed.

\begin{algorithm}[t!]
\caption{Bisection search method for solving \eqref{thm1_eq22}} \label{alg1}
\begin{algorithmic}[]
\STATE \textbf{given} $l=\mathsf{mse}_{Y_i}$, $u=\mathsf{mse}_{\infty}$, tolerance $\epsilon>0$.
\REPEAT
\STATE $\beta:= (l+u)/2$.
\STATE $o :={f_1(\beta) - {\beta} f_2(\beta)}.$
\STATE \textbf{if} $o\geq 0$, $l:=\beta$; \textbf{else}, $u:=\beta$.
\UNTIL $u-l\leq \epsilon$.
\STATE \textbf{return} $\beta$.
\end{algorithmic}
\end{algorithm}

\begin{algorithm}[t!]
\caption{Newton's method for solving \eqref{thm1_eq22}} \label{alg3}
\begin{algorithmic}[]
\STATE \textbf{given} tolerance $\epsilon>0$.
\STATE Pick initial value ${\beta}_0 {\in} [\mathsf{mse}_{\text{opt}}, \mathsf{mse}_{\infty})$.
\REPEAT
\STATE $\beta_{k+1}:= {\beta}_{k} - \frac{f({\beta}_k)}{f'({\beta}_k)}$.
\UNTIL $|\frac{f({\beta_k})}{f'({\beta_k})}| \leq \epsilon$.
\STATE \textbf{return} $\beta_{k+1}$.
\end{algorithmic}
\end{algorithm}

\begin{algorithm}[t!]
\caption{Fixed-point iterations for solving \eqref{thm1_eq22}} \label{alg5}
\begin{algorithmic}[]
\STATE \textbf{given} tolerance $\epsilon>0$.
\STATE Pick initial value $\beta_0 {\in} [{\mathsf{mse}}_{\text{opt}}, \mathsf{mse}_{\infty})$.
\REPEAT
\STATE ${\beta}_{k+1} := {\frac{f_1 (\beta_k)}{f_2 (\beta_k)}}$.
\UNTIL $| {\beta}_{k+1} - {\frac{f_1 (\beta_k)}{f_2 (\beta_k)}} | \leq \epsilon$.
\STATE \textbf{return} $\beta_{k+1}$.
\end{algorithmic}
\end{algorithm}

To achieve an even faster convergence speed, we can use Newton's method \cite{Mathews2004Numerical}
\begin{align} \label{newton_1}
{\beta}_{k+1} = {\beta}_{k} - \frac{f({\beta}_{k})}{f'({\beta}_{k})}
\end{align}
to solve \eqref{thm1_eq22}, as shown in Algorithm \ref{alg3}. We suggest choosing the initial value $\beta_{0}$ of Newton's method from the set $[\mathsf{{mse}_{\text{opt}}}, \mathsf{mse}_{\infty})$, i.e., $\beta_0$ is larger than the root $\mathsf{mse}_{\text{opt}}$. Such an initial value $\beta_{0}$ can be found by taking a few bisection search iterations, or by using the ${\mathsf{mse}}$ of a sub-optimal sampling policy \cite{CCWang2020}. Because $f(\beta)$ is a concave function, the choice of initial value $\beta_{0}\in[\mathsf{{mse}_{\text{opt}}}, \mathsf{mse}_{\infty})$ ensures that $\beta_{k}$ is a decreasing sequence converging to $\mathsf{mse}_{\text{opt}}$ \cite{spivak08}. 
Moreover, because $R_1(\cdot)$ and $R_2(\cdot)$ are twice continuously differentiable, the function $f(\beta)$ is twice continuously differentiable. Therefore, Newton's method is known to have a locally quadratic convergence speed in the neighborhood of the root $\mathsf{mse}_{\text{opt}}$ \cite[Chapter 2]{Mathews2004Numerical}.

\begin{figure}
	\centering
	\includegraphics[width=0.5\textwidth]{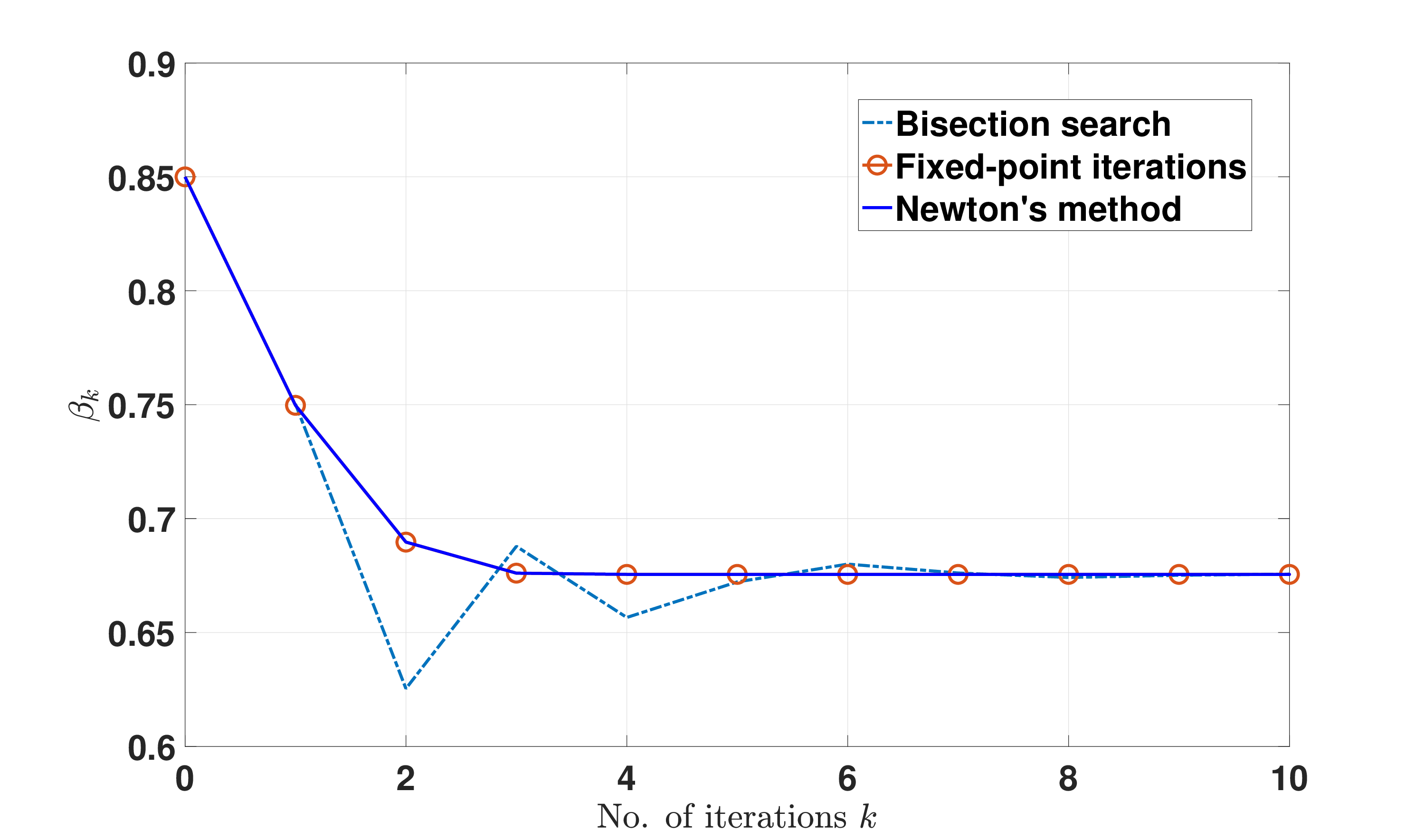}   
	\caption{Convergence of three algorithms for solving \eqref{thm1_eq22}, where the service times are \emph{i.i.d.} exponential with mean $\mathbb{E}[Y_i] = 1$, 
	the parameters of the OU process are $\sigma=1$ and $\theta=0.5$.}
	\label{figure_4}
	\vspace{-3mm}
\end{figure} 

\begin{figure}
	\centering
	\includegraphics[width=0.5\textwidth]{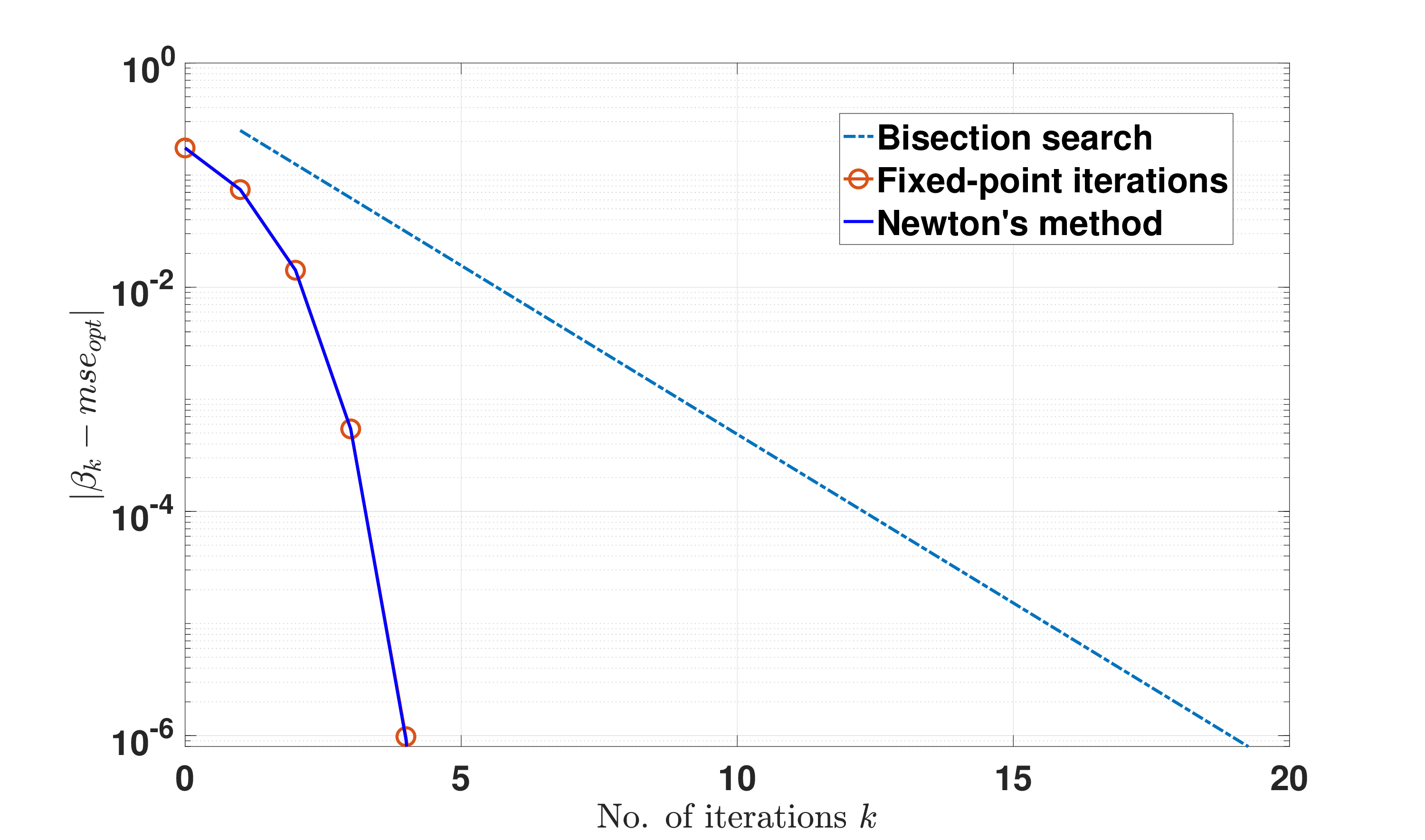}   
	\caption{Convergence of three algorithms for solving \eqref{thm1_eq22}, where the service times are \emph{i.i.d.} exponential with mean $\mathbb{E}[Y_i] = 1$, 
	the parameters of the OU process are $\sigma=1$ and $\theta=0.5$. For bisection search, we plot the difference $|u-l|$ between the upper bound $u$ and lower bound $l$, which is an upper bound of $| \beta_k - \mathsf{mse}_\text{opt} |$.}
	\label{figure5}
	\vspace{-3mm}
\end{figure} 

Newton's method requires to compute the gradient $f'(\beta_k)$, which can be solved by a finite-difference approximation, as in the secant method \cite{Mathews2004Numerical}. In the sequel, we introduce another approximation approach of Newton's method, which is of independent interest. In Theorem \ref{thm1}, we have shown that
\begin{align} \label{newton_11}
{\mathsf{mse}_{\text{opt}}} = \displaystyle \argmax_{{\beta} \in [\mathsf{{mse}_{\it{Y_i}}}, \mathsf{{mse}_{\infty}})} \frac{f_1(\beta)}{f_2(\beta)}.
\end{align}
Hence, the gradient of ${f_1(\beta)}/{f_2(\beta)}$ is equal to zero at the optimal solution $\beta = \mathsf{mse}_{\text{opt}}$, which leads to 
\begin{align} \label{newton_3}
f_1'(\mathsf{mse}_{\text{opt}}) f_2(\mathsf{mse}_{\text{opt}}) - f_1(\mathsf{mse}_{\text{opt}}) f_2'(\mathsf{mse}_{\text{opt}}) =0.
\end{align}
Therefore,
\begin{align} \label{opt_diff}
\mathsf{mse}_{\text{opt}} = \frac{f_1(\mathsf{mse}_{\text{opt}})}{f_2(\mathsf{mse}_{\text{opt}})} = \frac{f_1'(\mathsf{mse}_{\text{opt}})}{f_2'(\mathsf{mse}_{\text{opt}})}.
\end{align}
Because $f_1(\beta)$ and $f_2(\beta)$ are smooth functions, when $\beta_k$ is in the neighborhood of $\mathsf{mse}_{\text{opt}}$, \eqref{opt_diff} implies that $f_1'(\beta_k) - {\beta_k} f_2'(\beta_k) \approx {f_1'(\mathsf{mse}_{\text{opt}})} - {\mathsf{mse}_{\text{opt}}} {f_2'(\mathsf{mse}_{\text{opt}})} = 0$. Substituting this into \eqref{newton_1}, yields
\begin{align} \label{newton_2}
{\beta}_{k+1} = & {\beta}_{k} - \frac{f_1({\beta}_{k}) - {{\beta}_k} f_2({\beta}_k)}{f'_1({\beta}_{k}) - f_2(\beta_k) - {\beta_k} f'_2(\beta_k)} \nonumber\\
{\approx} & {\beta}_{k} - \frac{f_1({\beta}_{k}) - {{\beta}_k} f_2({\beta}_k)}{ - f_2(\beta_k)} \nonumber\\
= & \frac{f_1({\beta}_k)}{f_2({\beta}_k)},
\end{align}
which is a fixed-point iterative algorithm (see Algorithm \ref{alg5}) that was recently proposed in \cite{CCWang2020}. Similar to Newton's method, the fixed-point updates in \eqref{newton_2} converge to $\mathsf{mse}_{\text{opt}}$ if the initial value ${\beta_0} {\in} [\mathsf{mse}_{\text{opt}}, \mathsf{mse}_{\infty})$. Moreover, \eqref{newton_2} has a locally quadratic convergence speed, see 
\cite{CCWang2020} for a proof of this result. A numerical comparison of these three algorithms is shown in Fig.  \ref{figure_4} and Fig. \ref{figure5}. One can observe that the fixed-point updates and Newton's method converge faster than bisection search.

\begin{figure}
	\centering
	\includegraphics[width=0.5\textwidth]{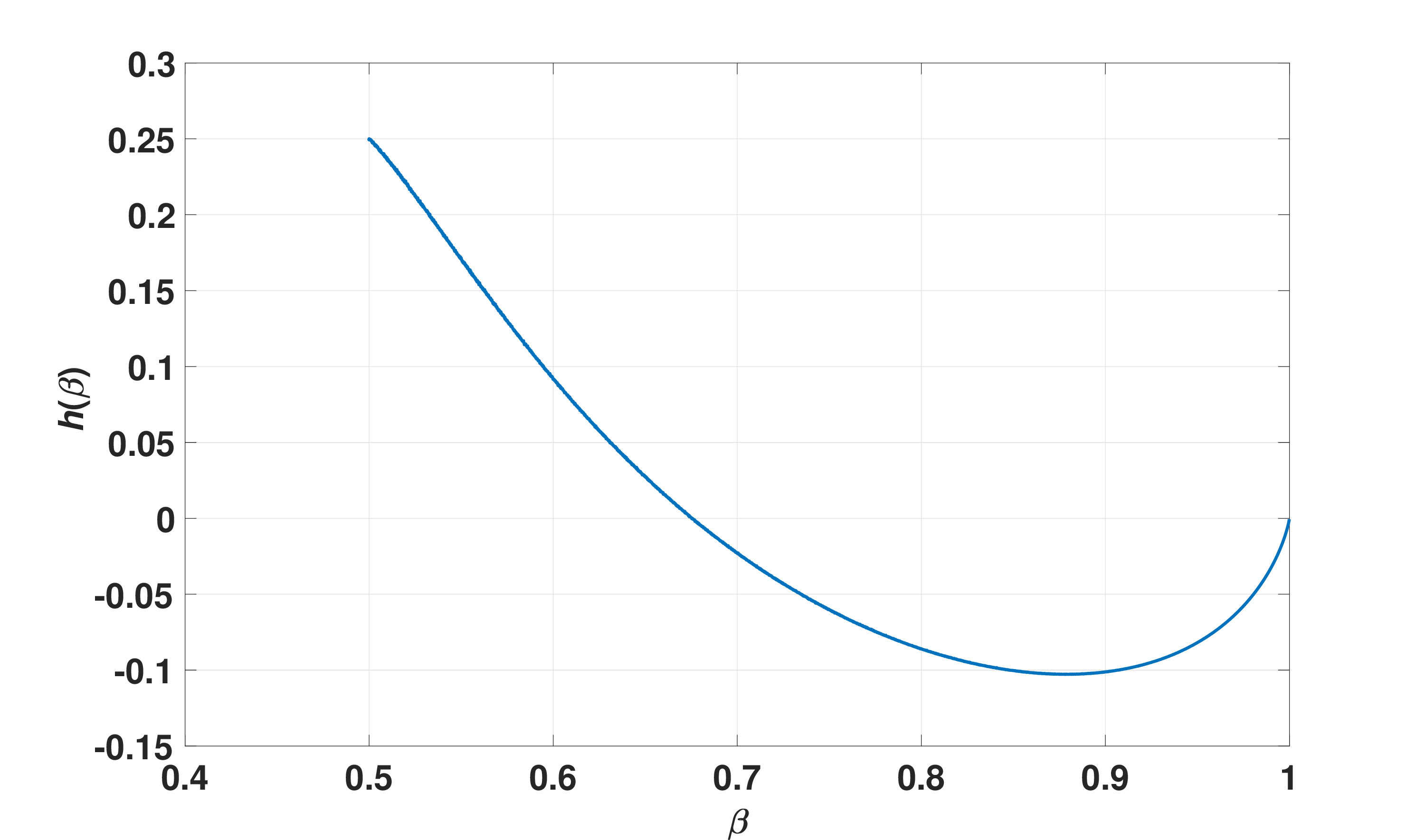}   
	\caption{The function $h(\beta)$ in \eqref{newton_4} for \emph{i.i.d.} exponential service time with $\mathbb{E}[Y_i] = 1$, where the parameters of the OU process are $\sigma=1$ and $\theta=0.5$. For these parameters, $\mathsf{mse}_{Y_i} = 0.5$ and $\mathsf{mse}_{\infty}=1$.}
	\label{f_two_root}
	\vspace{-3mm}
\end{figure} 

\begin{figure}
	\centering
	\includegraphics[width=0.5\textwidth]{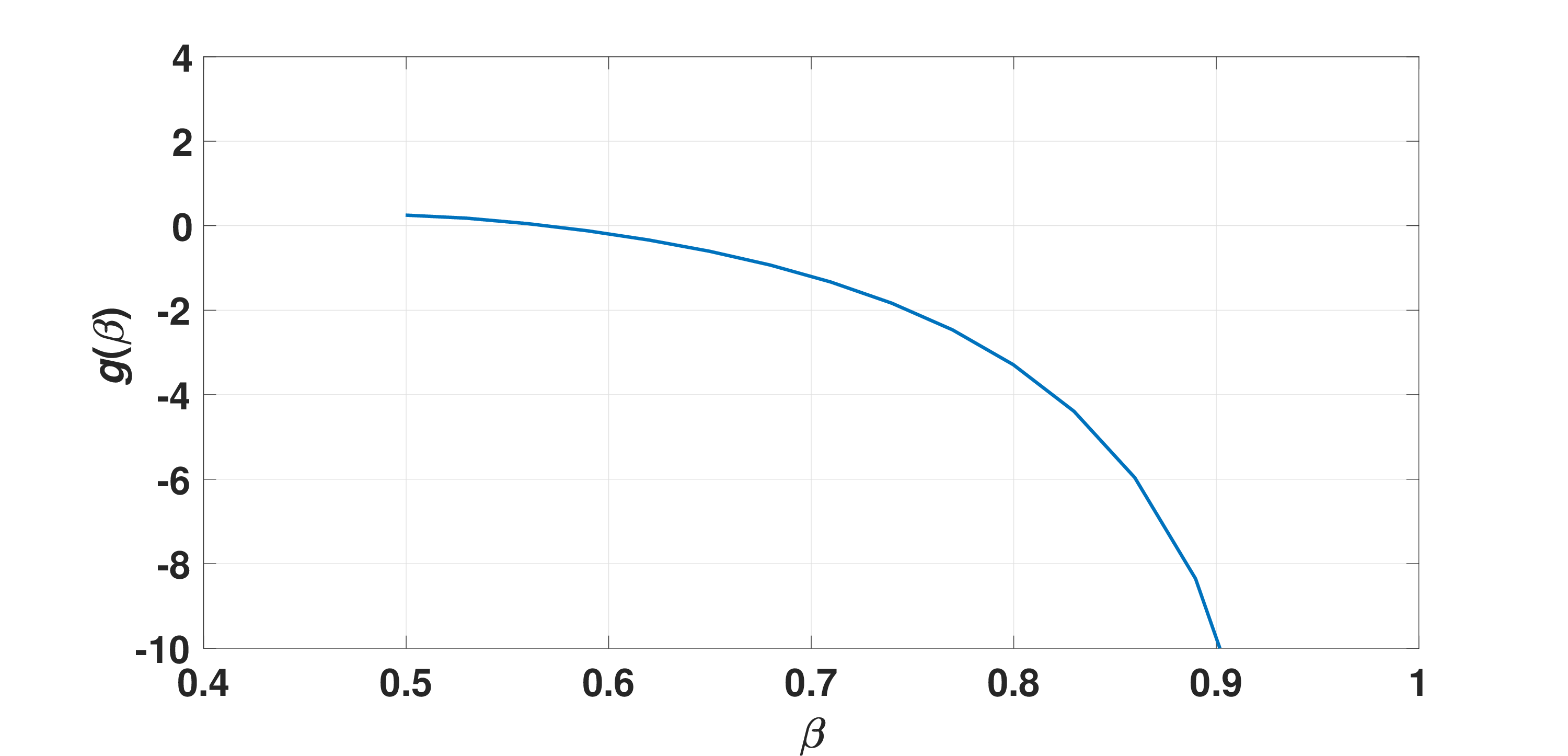}   
	\caption{The function $g(\beta)$ in \eqref{g} for \emph{i.i.d.} exponential service time with $\mathbb{E}[Y_i] = 1$ and $f_{\max} = 0.8$, where the parameters of the OU process are $\sigma=1$ and $\theta=0.5$. For these parameters, $\mathsf{mse}_{Y_i} = 0.5$ and $\mathsf{mse}_{\infty}=1$.}
	\label{g_one_root}
	\vspace{-3mm}
\end{figure}

We note that although \eqref{thm1_eq22}, and equivalently \eqref{two_root_1}, has a unique root $\mathsf{mse}_{\text{opt}}$, the fixed-point equation 
\begin{align} \label{newton_4}
h(\beta) = \frac{f_1(\beta)}{f_2(\beta)} - {\beta} = \frac{f_1(\beta)-\beta f_2(\beta)}{f_2(\beta)}  = 0
\end{align}
has two roots $\mathsf{mse}_{\text{opt}}$ and $\mathsf{mse}_{\infty}$. See  Fig. \ref{f_two_root} for an illustration of the two roots of $h(\beta)$. As shown in Appendix \ref{app_thm_strong_duality}, the correct root for computing the optimal threshold is $\mathsf{mse}_\text{opt}$. Interestingly, Algorithms \ref{alg1}-\ref{alg5} converge to the desired root $\mathsf{mse}_{\text{opt}}$, instead of $\mathsf{mse}_{\infty}$. Finally, we remark that these three algorithms can be used to find the optimal threshold in the age-optimal sampling problem studied in, e.g., \cite{SunSPAWC2018,SunNonlinear2019}.

\subsection{Signal-aware Sampling with Rate Constraint}

When the sampling rate constraint \eqref{eq_constraint} is taken into consideration, a solution to \eqref{eq_DPExpected} is expressed in the following theorem:

\begin{theorem} (Sampling with Rate Constraint).\label{thm_new}
If the $Y_i$'s are i.i.d. with $0<\mathbb{E}[Y_i] < \infty$, then \eqref{eq_opt_solution}-\eqref{thm1_eq22} is an optimal solution to \eqref{eq_DPExpected}.
The value of $\beta\geq0$ is determined in two cases: $\beta$ is the unique root  of \eqref{thm1_eq22}
if the root of \eqref{thm1_eq22} satisfies
\begin{align}\label{thm1_eq24}
\mathbb{E}[D_{i+1}(\beta)-D_i(\beta)] > {1}/{f_{\max}};
\end{align}  
otherwise, $\beta$ is the unique root  of 
\begin{align}\label{thm1_eq25}
\mathbb{E}[D_{i+1}(\beta)-D_i(\beta)] = {1}/{f_{\max}}.
\end{align}  
The optimal objective value to \eqref{eq_DPExpected} is  given by \emph{
\begin{align}\label{coro2_eq23}
{\mathsf{mse}}_{\text{opt}} = \frac{\mathbb{E}\left[\int_{D_i(\beta)}^{D_{i+1}(\beta)}\! (X_t-\hat X_t)^2dt\right]}{\mathbb{E}[D_{i+1}(\beta)\!-\!D_i(\beta)]}.
\end{align}} 
\end{theorem}


The proof of Theorem \ref{thm_new} is explained in Section \ref{sec_proof}. One can see that Theorem \ref{thm1} is a special case of Theorem \ref{thm_new} when $f_{\max} = \infty$. 

In Theorem \ref{thm_new}, the calculation of $\beta$ falls into two cases:
In one case, $\beta$ can be computed by solving \eqref{thm1_eq22} via Algorithms \ref{alg1}-\ref{alg5}.
For this case to occur, the sampling rate constraint \eqref{eq_constraint} needs to be inactive at the root of \eqref{thm1_eq22}. Because $D_i(\beta) = S_i(\beta) + Y_i$, we can obtain 
$\mathbb{E}[D_{i+1}(\beta) - D_i(\beta)] = \mathbb{E}[S_{i+1}(\beta) - S_i(\beta)]$ and hence \eqref{thm1_eq24} holds when the sampling rate constraint \eqref{eq_constraint} is inactive.

In the other case, $\beta$ is selected to satisfy the sampling rate constraint \eqref{eq_constraint} with equality, as required in \eqref{thm1_eq25}.
Before we solve \eqref{thm1_eq25}, let us first use $f_2(\beta)$ to express \eqref{thm1_eq25} as
\begin{align} \label{g}
g(\beta) = \frac{1}{f_{\max}} - f_2(\beta) = 0.
\end{align}
\begin{lemma}\label{beta_unique_2}
The function $g(\beta)$ has the following properties:

(i) $g(\beta)$ is continuous and strictly decreasing in $\beta$,

(ii)  $g(\mathsf{mse}_{Y_i}) \geq 0$ and $\displaystyle \lim_{{\beta} {\to} {\mathsf{mse}_{\infty}^{-}}} g(\beta) = -{\infty}$ if the root of \eqref{thm1_eq22} does not satisfy \eqref{thm1_eq24}.
\end{lemma}

\begin{proof}
See Appendix \ref{beta_unique_proof_2}.
\end{proof}
According to Lemma \ref{beta_unique_2}, \eqref{thm1_eq25} has a unique root in $[\mathsf{mse}_{Y_i}, \mathsf{mse}_{\infty})$, which is denoted as $\beta^*$. In addition, the numerical results in Fig. \ref{g_one_root} suggest that $g(\beta)$ should be concave, for which we do not have a proof. 

The root $\beta^*$ can be solved by using bisection search and Newton's method, which are explained in Algorithms \ref{alg2}-\ref{alg4}, respectively. 
Similar to the discussions in Section \ref{algo_sec}, the convergence of Algorithm \ref{alg2} is ensured by Lemma \ref{beta_unique_2}. Moreover, if 
$g(\beta)$ is concave and $\beta_0\in [\beta^*,\mathsf{mse}_{\infty})$, $\beta_{k}$ in Algorithm \ref{alg4} is a decreasing sequence converging to the root $\beta^*$ of \eqref{thm1_eq25} \cite{spivak08}.




\begin{algorithm}[t!]
\caption{Bisection search method for solving \eqref{thm1_eq25}} \label{alg2}
\begin{algorithmic}[]
\STATE \textbf{given} $l=\mathsf{mse}_{Y_i}$, $u=\mathsf{mse}_{\infty}$, tolerance $\epsilon>0$.
\REPEAT
\STATE $\beta:= (l+u)/2$.
\STATE $o :=\mathbb{E} [D_{i+1}(\beta)-D_i(\beta)]$.
\STATE \textbf{if} $o\geq {1}/{f_{\max}}$, $u:=\beta$; \textbf{else}, $l:=\beta$.
\UNTIL $u-l\leq \epsilon$.
\STATE \textbf{return} $\beta$.
\end{algorithmic}
\end{algorithm}

\begin{algorithm}[t!]
\caption{Newton's method for solving \eqref{thm1_eq25}} \label{alg4}
\begin{algorithmic}[]
\STATE \textbf{given} tolerance $\epsilon>0$.
\STATE Pick initial value ${\beta}_0 {\in} [\beta^*, \mathsf{mse}_{\infty})$.
\REPEAT
\STATE $\beta_{k+1}:= {\beta}_{k} - \frac{g({\beta}_k)}{g'({\beta}_k)}$.
\UNTIL $|\frac{g({\beta_k})}{g'({\beta_k})}| \leq \epsilon$.
\STATE \textbf{return} $\beta_{k+1}$.
\end{algorithmic}
\end{algorithm}\!\!
ter

\subsection{Special Case: 
Sampling of the Wiener Process}
In the limiting case that $\sigma = 1$ and ${\theta}\rightarrow0$, the OU process $X_t$ in \eqref{eq_SDE} becomes a Wiener process $X_t=W_t$. In this case, the MMSE estimator in \eqref{eq_esti} is given by 
\begin{align}\label{eq_Wiener}
\hat{X}_{t} = W_{S_i}, ~\text{if}~t\in[D_i,D_{i+1}).
\end{align}
\ifreport
As shown in Appendix \ref{app_Wiener}, $v(\cdot)$ defined by \eqref{eq_threshold} tends to 
\fi
\begin{align}\label{eq_Wiener}
v(\beta) = \sqrt{3{(\beta-\EE[Y_i])}}.
\end{align}
\begin{theorem} \label{coro2}
If $\sigma = 1$, ${\theta}\rightarrow0$, and the $Y_i$'s are i.i.d. with $0<\mathbb{E}[Y_i] < \infty$, then $(S_1(\beta),S_2(\beta),\ldots)$  with a parameter $\beta$ is an optimal solution to \eqref{eq_DPExpected}, where
%
\begin{align}\label{eq_coro2_opt_solution}
S_{i+1} (\beta)= \inf \left\{ t \geq D_i(\beta):\! \big|X_t - \hat X_t\big| \!\geq\! \sqrt{3(\beta-\EE[Y_i])}\right\},
\end{align}
$D_i (\beta)= S_i (\beta)+ Y_i$. The value of ${\beta} {\geq} 0$ is determined in two cases: $\beta$ is the unique root  of 
\begin{align}\label{coro2_eq22}
\!\!\! \mathbb{E}\left[\int_{D_i(\beta)}^{D_{i+1}(\beta)}(X_t-\hat X_t)^2dt\right]\! - \!{\beta} {\mathbb{E}[D_{i+1}(\beta)\!-\!D_i(\beta)]} \!=\! 0,\!\!\!
\end{align}
if the root of \eqref{coro2_eq22} satisfies $\mathbb{E}[D_{i+1}(\beta)-D_i(\beta)] > {1}/{f_{\max}}$; otherwise, $\beta$ is the unique root  of $\mathbb{E}[D_{i+1}(\beta)-D_i(\beta)] = {1}/{f_{\max}}$. The optimal objective value to \eqref{eq_DPExpected} is  given by \emph{
\begin{align}\label{coro2_eq23}
{\mathsf{mse}}_{\text{opt}} = \frac{\mathbb{E}\left[\int_{D_i(\beta)}^{D_{i+1}(\beta)}\! (X_t-\hat X_t)^2dt\right]}{\mathbb{E}[D_{i+1}(\beta)\!-\!D_i(\beta)]}.
\end{align}}
\end{theorem}

Theorem \ref{coro2} is an alternative form of Theorem 1 in \cite{2020Sun} and hence its proof is omitted. 
The benefit of the new expression in Theorem \ref{coro2} is that it allows to character $\beta$ based on the optimal objective value ${\mathsf{mse}}_{\text{opt}}$ and the sampling rate constraint \eqref{eq_constraint}, in the same way as in Theorems \ref{thm1}-\ref{thm_new}. This appears to be more fundamental than the expression in \cite{2020Sun}. The new form of optimal sampling policy of Wiener processes was also discovered in \cite{TsaiINFOCOM2020} without considering the constraint on \eqref{eq_constraint}.

\subsection{Signal-agnostic Sampling}

In signal-agnostic sampling policies, the sampling times $S_i$ are determined based only on the service times $Y_i$, but not on the observed OU process $\{X_t, t\geq 0\}$. 

\begin{lemma}\label{lem_estimation_error}
If $\pi\in\Pi_{\text{signal-agnostic}}$, then the mean-squared estimation error of the OU process $X_t$ at  time $t$ is 
\begin{align}\label{eq_lem_estimation_error1}
p(\Delta_t) = \!\!\mathbb{E}\left[(X_t-\hat X_t)^2\big|\pi, Y_1, Y_2, \dots \right] =  \frac{\sigma^2}{{2\theta}}\left(1-e^{-2\theta \Delta_t}\right),\!\!
\end{align}
which is a strictly increasing function of the age $\Delta_t$. 
\end{lemma}
\ifreport
\begin{proof}
See Appendix \ref{app_lem_estimation_error}. 
\end{proof}
\fi
According to Lemma \ref{lem_estimation_error}, for every policy $\pi\in\Pi_{\text{signal-agnostic}}$,
\begin{align}
\mathbb{E}\left[\int_0^{T} (X_t - \hat X_t)^2dt\right] \!\!= \mathbb{E}\left[ \int_0^{T}p(\Delta_t) dt\right].
\end{align}
Hence, minimizing the mean-squared estimation error among signal-agnostic sampling policies can be  formulated as the following MDP for minimizing the expected time-average of the nonlinear age function $p(\Delta_t)$ in \eqref{eq_lem_estimation_error1}: 
\begin{align}\label{eq_age}
\mathsf{mse}_{\text{age-opt}}=\inf_{\pi \in\Pi_{\text{signal-agnostic}} }\!\!& \limsup_{T\rightarrow \infty}\frac{1}{T}\mathbb{E}\left[ \int_0^{T}p(\Delta_t) dt\right] \\
~\text{s.t.}\quad  \quad \!\!&\liminf_{n\rightarrow \infty} \frac{1}{n} 
\mathbb{E}\left[\sum_{i=1}^n (S_{i+1}-S_i)\right]\!\!\geq \frac{1}{f_{\max}},\label{eq_constraint_age}
\end{align}
where $\mathsf{mse}_{\text{age-opt}}$ is the optimal value of \eqref{eq_age}. 
By \eqref{eq_lem_estimation_error1}, $p(\Delta_t)$ and $\mathsf{mse}_{\text{age-opt}}$ are bounded. Because $\Pi_{\text{signal-agnostic}}\subset \Pi$, it follows immediately that {$\mathsf{mse}_{\text{opt}} \leq \mathsf{mse}_{\text{age-opt}}$.}

Problem \eqref{eq_age} is one instance of the problems recently solved in Corollary 3 of  \cite{SunNonlinear2019} for general strictly increasing functions $p(\cdot)$. 
From this,  a solution to \eqref{eq_age} for signal-agnostic sampling is given by

\begin{theorem} \label{thm2}
If  the $Y_i$'s are i.i.d. with $0<\mathbb{E}[Y_i] < \infty$, then 
$(S_1(\beta),S_2(\beta),\ldots)$  with a parameter $\beta$ is an optimal solution to \eqref{eq_age}, where
%
\begin{align}\label{eq_thm2_opt_solution}
S_{i+1} (\beta)= \inf \left\{ t \geq D_i(\beta):\! \mathbb{E}[ (X_{t+Y_{i+1}}\! -\! \hat X_{t+Y_{i+1}})^2] \!\geq\! \beta\right\}\!,\!\!
\end{align}
$D_i (\beta)= S_i (\beta)+ Y_i$ and $\beta$ is the unique root  of 
\begin{align}\label{thm2_eq22_1}
\!\!\! \mathbb{E}\left[\int_{D_i(\beta)}^{D_{i+1}(\beta)}(X_t-\hat X_t)^2dt\right]\! - \!{\beta} {\mathbb{E}[D_{i+1}(\beta)\!-\!D_i(\beta)]} \!=\! 0,\!\!\!\end{align}
if the root of \eqref{thm2_eq22_1} satisfies $\mathbb{E}[D_{i+1}(\beta)-D_i(\beta)] > {1}/{f_{\max}}$; otherwise, $\beta$ is the unique root  of 
\begin{align}\label{thm2_eq22_2}
\mathbb{E}[D_{i+1}(\beta)-D_i(\beta)] = {1}/{f_{\max}}. 
\end{align}
The optimal objective value to \eqref{eq_age} is  given by \emph{
\begin{align}\label{thm2_eq23}
{\mathsf{mse}}_{\text{age-opt}} = \frac{\mathbb{E}\left[\int_{D_i(\beta)}^{D_{i+1}(\beta)}\! (X_t - \hat X_t)^2dt\right]}{\mathbb{E}[D_{i+1}(\beta)\!-\!D_i(\beta)]}.
\end{align} }
\end{theorem}

Theorem \ref{thm2} follows from Corollary 3 of \cite{SunNonlinear2019} and Lemma \ref{lem_estimation_error}. 
Similar to the case of signal-aware sampling, the roots of \eqref{thm2_eq22_1} and \eqref{thm2_eq22_2} can be solved by using Algorithms \ref{alg1}-\ref{alg4}. In fact, Algorithms \ref{alg1}-\ref{alg4} can be used for minimizing general non-decreasing age penalty \cite{SunNonlinear2019}.

\subsection{Discussions of the Results} 

The  difference among Theorems \ref{thm1}-\ref{thm2} is only in the expressions \eqref{eq_opt_solution}, \eqref{eq_coro2_opt_solution}, \eqref{eq_thm2_opt_solution} of threshold policies. In  signal-aware sampling policies \eqref{eq_opt_solution} and \eqref{eq_coro2_opt_solution}, the sampling time is determined by the \emph{instantaneous} estimation error $\big|X_t - \hat X_t\big|$, and the threshold function $v(\cdot)$ is determined by the specific signal model. 
In the signal-agnostic sampling policy \eqref{eq_thm2_opt_solution}, the sampling time is determined by the \emph{expected} estimation error $\mathbb{E}[ (X_{t+Y_{i+1}}\! -\! \hat X_{t+Y_{i+1}})^2]$ at time $t+Y_{i+1}$. We note that  if $t =S_{i+1}(\beta)$, then $t+Y_{i+1}=S_{i+1}(\beta)+Y_{i+1} = D_{i+1}(\beta)$ is the  delivery time of the new sample. Hence,  \eqref{eq_thm2_opt_solution} requires that the expected estimation error upon the delivery of the new sample is no less than   $\beta$. The parameter $\beta$ in Theorems \ref{thm1}-\ref{thm2} is determined by the optimal objective value and the sampling rate constraint in the same manner. Later on in
\ifreport
\eqref{eq_beta_value1},
\fi we will further see that $\beta$ is exactly equal to the summation of the optimal objective value of the MDP and the optimal Lagrangian dual variable associated to the sampling rate constraint. 
Finally, it is worth noting  that Theorems \ref{thm1}-\ref{thm2} hold for  all distributions of the service times $Y_i$ satisfying $0<\mathbb{E}[Y_i]<\infty$, and for both constrained and unconstrained sampling problems.
\section{Proof of the Main Results}\label{sec_proof}

We first provide the proof of Theorem \ref{thm_new}. After that Theorem \ref{thm1} follows immediately because it is a special case of Theorem \ref{thm_new}. We prove Theorem \ref{thm_new} in four steps: (i) We first show that  sampling should be suspended when the server is busy, which can be used to simplify \eqref{eq_DPExpected}. (ii) We use an extended Dinkelbach's method  \cite{Dinkelbach67}  and Lagrangian duality method to decompose the simplified problem into a series of mutually independent per-sample  MDP. (iii) We utilize the free boundary method from optimal stopping theory \cite{Peskir2006} to solve the per-sample MDPs analytically. (iv) Finally, we use a geometric multiplier method \cite{Bertsekas2003} to show that the  duality gap is zero. 
The above proof framework is an extension to that used in \cite{SunNonlinear2019, 2020Sun}, and the most challenging part is Step (iii). 

\subsection{Preliminaries} \label{sec_pre}
%
The OU process $O_t$ in \eqref{eq_OU} with initial state $O_t = 0$ and parameter $\mu = 0$ is  the solution to the SDE
\begin{align}
dO_t = - \theta O_t dt + \sigma dW_t. 
\end{align}
In addition, the infinitesimal generator of $O_t$ is \cite[Eq. A1.22]{Borodin1996}
\begin{align}\label{eq_generator}
  \mathcal{G}= -\theta u \frac{\partial} {\partial u} + \frac{\sigma^2}{2}\frac{\partial^2} {\partial u^2}.
\end{align}
According to \eqref{eq_process} and \eqref{eq_esti}, the estimation error $(X_t-\hat X_t)$ is of the same distribution with $O_{t-S_i}$, if $t\in[D_i,D_{i+1})$. 
By using Dynkin's formula and the optional stopping theorem, we  obtain the following lemma.
\begin{lemma}\label{lem_stop}
Let $\tau \geq0$ be a stopping time of the OU process $O_t$ with $ \EE\left[ \tau\right]<\infty$, 
then
\begin{align}
\mathbb{E}\left[\int_0^\tau O_t^2 dt\right] & = \mathbb{E}\left[\frac{\sigma^2}{2\theta}\tau - \frac{1}{2\theta}O_\tau^2\right].\label{eq_stop}
\end{align}
If, in addition, $\tau$ is the first exit time of a bounded set, then
\begin{align}\label{eq_stop11}
&\mathbb{E}\left[\tau\right] = \mathbb{E} [R_1(O_{\tau})] ,\\
&\mathbb{E}\left[\int_0^\tau O_t^2 dt\right]  = \mathbb{E} [R_2(O_{\tau})],\label{eq_stop12}
\end{align}
where $R_1(\cdot)$ and $R_2(\cdot)$ are defined in \eqref{eq_R_1} and \eqref{eq_R_2}, respectively. 
\end{lemma}
\begin{proof}
See Appendix \ref{app_lem4}. 
\end{proof}
\subsection{Suspend Sampling When the Server is Busy} %

By using the strong Markov property of the OU process $X_t$ and the orthogonality principle of MMSE estimation, we obtain the following useful lemma: 

\begin{lemma}\label{lem_zeroqueue}
Suppose that a feasible sampling policy for problem \eqref{eq_DPExpected} is $\pi$, in which at least one sample is taken when the server is busy processing an earlier generated sample. Then, there exists another feasible policy $\pi'$ for problem \eqref{eq_DPExpected} which has a smaller estimation error than policy $\pi$. Therefore, in \eqref{eq_DPExpected}, it is suboptimal to take a new sample before the previous sample is delivered.
\end{lemma}

\begin{proof}
See Appendix \ref{app_lem1}. 
\end{proof}
A similar result was obtained in \cite{2020Sun} for the sampling of Wiener processes.
By Lemma \ref{lem_zeroqueue}, there is no loss to consider a sub-class of sampling policies $\Pi_1\subset\Pi$ such that each  sample is generated and sent out  after all previous samples are delivered, i.e., 
\begin{align}
\Pi_1 = \{\pi\in\Pi: S_{i} = G_{i} \geq D_{i-1} \text{ for all $i$}\}. \nonumber
\end{align}
For any policy $\pi\in\Pi_1$, the \emph{information} used for determining $S_i$ includes: (i) the history of signal values $(X_t: t\in[0, S_i])$ and (ii) the service times  $(Y_1,\ldots, Y_{i-1})$ of previous samples. Let us define the $\sigma$-fields $\mathcal{F}_t = \sigma(X_s: s\in[0, t])$ and $\mathcal{F}_t^+ = \cap_{r>t}\mathcal{F}_r$. 
Then, $\{\mathcal{F}_t^+,t\geq0\}$ is the {filtration} (i.e., a non-decreasing and right-continuous family of $\sigma$-fields) of the OU process $X_t$.
Given the service times   $(Y_1,\ldots, Y_{i-1})$ of previous samples, $S_{i} $ is a \emph{stopping time} with respect to the filtration $\{\mathcal{F}_t^+,t\geq0\}$ of the OU process $X_t$, that is
\begin{align}
[\{S_{i}\leq t\} | Y_1,\ldots, Y_{i-1}] \in \mathcal{F}_t^+.\label{eq_stopping}
\end{align}  
Hence, the policy space $\Pi_1$ can be expressed as
\begin{align}\label{eq_policyspace}
\!\!\!\!\Pi_1 = & \{S_i : [\{S_{i}\leq t\} | Y_1,\ldots, Y_{i-1}] \in \mathcal{F}_t^+, \nonumber\\
&~~~~~~~\text{$T_i$ is a regenerative process} \}.\!\!
\end{align}  
Let $Z_i = S_{i+1} - D_{i}\geq0$ represent the \emph{waiting time} between the delivery time $D_{i}$ of the $i$-th sample and the generation time $S_{i+1}$ of the $(i+1)$-th sample. Then, $S_i =  \sum_{j=0}^{i-1} (Y_{j} + Z_j)$ and $D_i = \sum_{j=0}^{i-1} (Y_{j} + Z_j) + Y_i$ for each $i=1,2,\ldots$ 
Given $(Y_0,Y_1,\ldots)$, $(S_1,S_2,\ldots)$ is uniquely determined by $(Z_0,Z_1,\ldots)$. Hence, one can also use $\pi = (Z_0,Z_1,\ldots)$ to represent a sampling policy. 

Because $\{X_t-\hat X_t, t\in[D_i,D_{i+1})\}$ and $\{O_{t-S_i}, t\in[D_i,D_{i+1})\}$ are of the same distribution,  for each $i = 1,2,\ldots$,
\begin{align}\label{eq_integral}
&\mathbb{E}\left[\int_{D_{i}}^{D_{i+1}} (X_t-\hat X_t)^2dt\right] \nonumber\\
= & \mathbb{E}\left[\int_{D_{i}}^{D_{i+1}} O_{t-S_i}^2dt\right] 
= \mathbb{E}\left[\int_{Y_{i}}^{Y_i+Z_i+Y_{i+1}} O_{s}^2ds\right].  
\end{align}
Because $T_i $ is a regenerative process, the renewal theory \cite{Ross1996} tells us that $\frac{1}{n} 
\mathbb{E}[S_n]$ is a convergent sequence and \begin{align}
&\limsup_{T\rightarrow \infty}\frac{1}{T}\mathbb{E}\left[\int_0^{T} (X_t-\hat X_t)^2dt\right] \nonumber
\\
=& \lim_{n\rightarrow \infty}\frac{\mathbb{E}\left[\int_0^{D_n} (X_t-\hat X_t)^2dt\right]}{\mathbb{E}[D_n]} \nonumber\\
=& \lim_{n\rightarrow \infty}\frac{\sum_{i=1}^n  \mathbb{E}\left[\int_{Y_{i}}^{Y_i+Z_i+Y_{i+1}} O_{s}^2ds\right]}{\sum_{i=1}^n \mathbb{E}\left[Y_i+Z_i\right]}.
\end{align}
Hence, \eqref{eq_DPExpected} can be rewritten as the following MDP:
\begin{align}\label{eq_Simple}
{\mathsf{mse}}_{\text{opt}}=&\inf_{\pi\in\Pi_1} \lim_{n\rightarrow \infty}\frac{\sum_{i=1}^n  \mathbb{E}\left[\int_{Y_{i}}^{Y_i+Z_i+Y_{i+1}} O_{s}^2ds\right]}{\sum_{i=1}^n \mathbb{E}\left[Y_i+Z_i\right]} \\
&~\text{s.t.}~\lim_{n\rightarrow \infty} \frac{1}{n} 
\sum_{i=1}^n \mathbb{E}\left[Y_i+Z_i\right]\geq \frac{1}{f_{\max}},\nonumber
\end{align}
where ${\mathsf{mse}}_{\text{opt}}$ is the optimal  value of \eqref{eq_Simple}.

\subsection{Reformulation of Problem \eqref{eq_Simple}}

In order to solve \eqref{eq_Simple}, let us consider the following MDP with a parameter $c\geq 0$:
\begin{align}\label{eq_SD}
\!\!\!\!\!\!h(c)\!=\!\inf_{\pi\in\Pi_1}\!&\lim_{n\rightarrow \infty}\frac{1}{n}\sum_{i=1}^n\!\mathbb{E}\!\left[\int_{Y_{i}}^{Y_i+Z_i+Y_{i+1}} \!\!O_{s}^2ds-c(Y_i+Z_i)\right]\!\!\!\!\\
\text{s.t.}~&\lim_{n\rightarrow \infty} \frac{1}{n} 
\sum_{i=1}^n \mathbb{E}\left[Y_i+Z_i\right]\geq \frac{1}{f_{\max}},\nonumber
\end{align}
where $h(c)$ is the optimum  value of \eqref{eq_SD}. 
Similar with Dinkelbach's method  \cite{Dinkelbach67} for nonlinear fractional programming, the following lemma holds for the MDP \eqref{eq_Simple}:
\begin{lemma} \cite{2020Sun}\label{lem_ratio_to_minus}
The following assertions are true: 
\begin{itemize}
\vspace{0.5em}
\item[(a).] \emph{${\mathsf{mse}}_{\text{opt}} \gtreqqless c $} if and only if $h(c)\gtreqqless 0$. 
\vspace{0.5em}
\item[(b).] If $h(c)=0$, the solutions to \eqref{eq_Simple}
and \eqref{eq_SD} are identical. 
\end{itemize}
\end{lemma}

Hence, the solution to \eqref{eq_Simple} can be obtained by solving \eqref{eq_SD} and seeking $c=\mathsf{mse}_{\text{opt}}\geq 0$ such that 
\begin{align}\label{eq_c}
h(\mathsf{mse}_{\text{opt}})=0.
\end{align}

\subsection{Lagrangian Dual Problem of  \eqref{eq_SD}  }

Next, we use the Lagrangian dual approach to solve \eqref{eq_SD} with $c=\mathsf{mse}_{\text{opt}}$. 
We define the  Lagrangian associated with \eqref{eq_SD} as
\begin{align} \label{lagrange}
& L(\pi;\lambda) \nonumber\\
=
& \lim_{n\rightarrow \infty}\frac{1}{n}\sum_{i=1}^n\mathbb{E}\bigg[\int_{Y_{i}}^{Y_i+Z_i+Y_{i+1}}\!\! O_{s}^2ds-(\mathsf{mse}_{\text{opt}}+\lambda)(Y_i\!+\!Z_i)\bigg] \nonumber\\
&+ \frac{\lambda}{f_{\max}},
\end{align}
where $\lambda\geq 0$ is the dual variable.  
Let
\begin{align}\label{eq_primal}
e(\lambda) = \inf_{\pi\in\Pi_1}  L(\pi;\lambda).
\end{align}
Then, the  dual problem of \eqref{eq_SD} is defined by
\begin{align}\label{eq_dual}
& d=\max_{\lambda\geq 0}e(\lambda),
\end{align}
where $d$ is the optimum value of \eqref{eq_dual}. Weak duality \cite{Bertsekas2003} implies $d \leq h(\mathsf{mse}_{\text{opt}})$. In Section \ref{sec:dual}, we will establish strong duality, i.e., $d = h(\mathsf{mse}_{\text{opt}})$.

In the sequel, we decompose \eqref{eq_primal}  into a sequence of  mutually independent per-sample MDPs. Let us define
\begin{align}\label{eq_beta_value1}
\beta &= \mathsf{mse}_{\text{opt}}+\lambda.
\end{align}
As shown in Appendix \ref{app_eq_simplification_1}, by using Lemma \ref{lem_stop}, we can obtain
\begin{align} \label{eq_simplification_1}
&\mathbb{E}\left[\int_{Y_{i}}^{Y_i+Z_i+Y_{i+1}} O_{s}^2ds\!-\! \beta(Y_i\!+\!Z_i)\right] \nonumber\\
=& \mathbb{E}\left[\int_{Y_{i}}^{Y_i+Z_i} (O_{s}^2-\beta) ds+ \gamma  O_{Y_i+Z_i}^2\right] \nonumber\\
&+ \frac{\sigma^2}{2\theta}[\EE (Y_{i+1}) -\gamma]- \beta \EE [Y_{i+1}],
\end{align}
where $\gamma$ is defined in \eqref{eq_gamma1}.
For any $s\geq 0$, define the $\sigma$-fields $\mathcal{F}^{s}_t = \sigma(O_{s+r}-O_{s}: r\in[0, t])$ and  the right-continuous filtration $\mathcal{F}_t^{s+} = \cap_{r>t}\mathcal{F}_r^{s}$. Then, $\{\mathcal{F}^{s+}_t,t\geq0\}$ is the {filtration}  of the {time-shifted OU process} $\{O_{s+t}-O_{s},t\in[0,\infty)\}$. Define ${\mathfrak{M}}_s$ as the set of integrable stopping times of $\{O_{s+t}-O_{s},t\in[0,\infty)\}$, i.e.,
\begin{align}
\mathfrak{M}_{s} = \{\tau \geq 0:  \{\tau\leq t\} \in \mathcal{F}^{s+}_t, \mathbb{E}\left[\tau\right]<\infty\}.
\end{align}
By 
using  a sufficient statistic of \eqref{eq_primal},  
we can  obtain
\begin{lemma}\label{thm_solution_form}
An optimal solution  $(Z_0,Z_1,\ldots)$ 
to 
\eqref{eq_primal} 
satisfies 
\begin{align}\label{eq_opt_stopping}
&\inf_{Z_i\in \mathfrak{M}_{Y_i}}\mathbb{E}\left[\int_{Y_i}^{Y_i+Z_i} \!\!\!(O_{s}^2-\beta)ds + \gamma  O_{Y_i+Z_i}^2\bigg|  O_{Y_i},Y_i \right],
\end{align}
where  $\beta\geq0$ and $\gamma\geq0$ are defined in \eqref{eq_beta_value1} and \eqref{eq_gamma1}, respectively.
\end{lemma}
\begin{proof}
See Appendix \ref{app_thm_solution_form}. 
\end{proof}
By this, \eqref{eq_primal} is decomposed as  a series of per-sample  MDP \eqref{eq_opt_stopping}.

%
%
%
%
%
%
%
\ifreport
\subsection{Analytical Solution to Per-Sample MDP \eqref{eq_opt_stopping}}
\fi
%
%
\ifreport
We solve \eqref{eq_opt_stopping} by using the free-boundary approach for optimal stopping problems \cite{Peskir2006}.
\fi

Let us consider an OU process $V_t$ with  initial state $V_0=v $ and parameter $\mu=0$. Define the $\sigma$-fields $\mathcal{F}^{V}_t = \sigma(V_{s}: s\in[0, t])$, $\mathcal{F}^{V+}_t = \cap_{r>t}\mathcal{F}_r^{V}$, and the {filtration} $\{\mathcal{F}^{V+}_t,t\geq0\}$ associated to $\{V_t,t\geq0\}$. Define ${\mathfrak{M}}_V$ as the set of integrable stopping times of $\{V_t,t\in[0,\infty)\}$, i.e.,
\begin{align}
\mathfrak{M}_{V} = \{\tau \geq 0:  \{\tau\leq t\} \in \mathcal{F}^{V+}_t, \mathbb{E}\left[\tau\right]<\infty\}.
\end{align}
\ifreport
Our
\fi
 goal is to solve the following optimal stopping problem for any given initial state $v\in \mathbb{R}$
\begin{align}\label{eq_stop_problem}
\!\sup_{\tau \in \mathfrak{M}_{V} } \mathbb{E}_v\left[-\gamma  V_{\tau}^2 - \int_{0}^{\tau} \!\!\!(V_{s}^2-\beta)ds\right],
\end{align}
\ifreport
where $\mathbb{E}_v[\cdot]$ is the conditional expectation for given initial state $V_0 =v$, $\gamma$ and $\beta$ are given by \eqref{eq_gamma1} and \eqref{eq_beta_value1}, respectively. Hence, \eqref{eq_opt_stopping} is one instance of  \eqref{eq_stop_problem} with $v= O_{Y_i}$, where the supremum is taken over all stopping times $\tau$ of $V_t$. In this subsection, we focus on the case that $\beta$ in \eqref{eq_stop_problem} satisfies $\mathsf{mse}_{Y_i} \leq {\beta} < \mathsf{mse}_{\infty}$. Later on in Section \ref{sec:dual}, we will show that this condition is indeed satisfied by the optimal solution to \eqref{eq_SD}.
\fi

In order to solve \eqref{eq_stop_problem}, we first find a candidate solution to \eqref{eq_stop_problem} by solving a free boundary problem; then we prove that the free boundary solution is indeed the value function of \eqref{eq_stop_problem}:
%

\subsubsection{A Candidate Solution to \eqref{eq_stop_problem}}
Now, we show how to solve \eqref{eq_stop_problem}. 
The general optimal stopping theory in Chapter I of \cite{Peskir2006} tells us that the following guess of the stopping time should be optimal for Problem \eqref{eq_stop_problem}: 
\begin{align}\label{eq_optimal_stopping123}
\tau_{*} = \inf\{t \geq 0: |V_{t}|\geq v_* \},
\end{align}
where $v_*\geq 0$ is the optimal stopping threshold to be found.
Observe that in this guess, the continuation region $(-v_*,v_*)$ is assumed symmetric around zero. This is because the  OU process is symmetric, i.e., the process $\{-V_t,t\geq 0\}$ is also an OU process started at $-{V_0} = -v$. Similarly, we can also argue that the value function of problem \eqref{eq_stop_problem} should be even.

According to \cite[Chapter 8]{Peskir2006}, and \cite[Chapter 10]{Bernt2000},  the value function 
and the optimal stopping threshold $v_*$ should satisfy the following free boundary  problem:
\begin{align}\label{eq_free1}
&
\frac{\sigma^2}{2} H''(v) -\theta v H'(v) 
=  v^2 - \beta,~ ~~ v\in(-v_*,v_*),\\
&H(\pm v_*) =  -\gamma  v_*^2, \label{eq_free2}\\
&H'(\pm v_*) = \mp 2 \gamma  v_*.\label{eq_free3}
\end{align}
\ifreport
In Appendix \ref{app_eq_W},
\fi
 we use the integrating factor method \cite[Sec. I.5]{Amann1990} to find 
the general solution to \eqref{eq_free1}, which is given by
\begin{align}\label{eq_W}
H(v) =& -\frac{v^2}{2\theta} + \left({\frac{1}{2{\theta}}} - {\frac{\beta}{{\sigma}^2}}\right) 
{}_2F_2\left(1,1;\frac{3}{2},2;\frac{\theta}{\sigma^2}v^2\right) {v^2}\nonumber\\
& + C_1\text{erfi}\left(\frac{\sqrt{\theta}}{\sigma}v\right) + C_2,~ \quad \quad v\in(-v_*,v_*),
\end{align}
where $C_1$ and $C_2$ are constants to be found for satisfying \eqref{eq_free2}-\eqref{eq_free3}, 
 and {erfi}$(x)$ is the imaginary error function, i.e.,
\begin{align}\label{eq_erfi}
\text{erfi}(x) = \frac{2}{\sqrt \pi} \int_0^x e^{t^2} dt.
\end{align}
Because $H(v)$ should be even but {erfi}$(x)$ is odd, we should choose $C_1=0$. Further, in order to satisfy the boundary condition \eqref{eq_free2}, $C_2$ is chosen as 
\begin{align} 
\!\!\!\!\!\!C_2 = \!\! \frac{1}{2\theta}\EE \left(e^{-2\theta Y_{i}} \right)  v_*^2\!  -\! \left({\frac{1}{2{\theta}}}\! - \!{\frac{\beta}{{\sigma^2}}}\right)\!\! {}_2F_2\big(1,1;\frac{3}{2},2;\frac{\theta}{\sigma^2}v_*^2\big){v_*^2}, \!\!\!\!\!\!
\end{align}
where we have used \eqref{eq_gamma1}.
With this, the expression of $H(v)$ is obtained in the continuation region ($-v_*$, $v_*$). In the stopping region $|v|\geq v_*$, the stopping time in \eqref{eq_optimal_stopping123} is simply $\tau_*=0$, because $|V_0| = |v|\geq v_*$. Hence, if $|v|\geq v_*$, the objective value  achieved by the sampling time \eqref{eq_optimal_stopping123} is
\begin{align} \label{val_new}
\!\! \mathbb{E}_v\left[-\gamma  v^2 -\! \int_{0}^{0} \!\!\!(V_{s}^2-\beta)ds\right] \!= \!-\gamma  v^2 . \!\!\!\! 
\end{align} 
Combining \eqref{eq_W}-\eqref{val_new}, we obtain a candidate of the value function for \eqref{eq_stop_problem}:
\begin{align}\label{eq_W1}
\!\! \!\! H(v)  = \left\{\!\! \begin{array}{l l}-\frac{v^2}{2\theta} + \left({\frac{1}{2{\theta}}} - {\frac{\beta}{{\sigma}^2}}\right) 
{}_2F_2\big(1,1;\frac{3}{2},2;\!\!&\!\!\!\!\frac{\theta}{\sigma^2}v^2\big) {v^2}+ C_2, \\
& \text{ if }~ |v|<v_*,\\
-\gamma  v^2, & \text{ if }~ |v|\geq v_*. \end{array}\right. \!\!\!\!\!\!
\end{align}

Next, we find a candidate value of the optimal stopping threshold $v_*$. 
By taking the gradient of $H(v)$, we get
\begin{align}\label{H_gradient}
\!\!\! H'(v) = -\frac{v}{\theta} + \left({\frac{\sigma}{{\theta}^{\frac{3}{2}}}} - {\frac{2\beta}{{\sigma}{\sqrt{\theta}}}}\right) F\left(\frac{\sqrt \theta}{\sigma}v\right),~ ~~ v\in(-v_*,v_*), \!\!\!
\end{align}
where 
\begin{align}
 F(x) = e^{x^2} \int_0^{x} e^{-t^2} dt.  
\end{align}
The boundary condition \eqref{eq_free3} implies that $v_*$ is the root of 
\begin{align}\label{eq_threshold1}
 -\frac{v}{\theta} + \left({\frac{\sigma}{{\theta}^{\frac{3}{2}}}} - {\frac{2\beta}{{\sigma}{\sqrt{\theta}}}}\right) F\left(\frac{\sqrt \theta}{\sigma}v\right) = - 2\gamma v.
 \end{align}
Substituting \eqref{eq_mse_Yi}, \eqref{eq_mse_infty}, and \eqref{eq_gamma1} into \eqref{eq_threshold1}, yields that $v_*$ is the root of 
\begin{align}\label{eq_threshold11}
& \left({\mathsf{mse}_{\infty} - {\beta}}\right) G{\bigg(\frac{\sqrt{\theta}}{\sigma} v\bigg)} = {\mathsf{mse}_{\infty} - \mathsf{mse}_{Y_i}},
 \end{align}
where $G(\cdot)$ is defined in \eqref{eq_g1}. Because $\mathsf{mse}_{Y_i} \leq {\beta} < \mathsf{mse}_{\infty}$, $G(x)$ is strictly increasing on $[0, \infty)$, and $G(0) =1$, we know that \eqref{eq_threshold11} has a unique non-negative root $v_*$. Further, the root $v_*$ can be expressed as a function $v(\beta)$ of $\beta$, where $v(\beta)$ is defined in \eqref{eq_threshold}. By this, we obtain a candidate solution to \eqref{eq_stop_problem}.

\subsubsection{Verification of the Optimality of the Candidate Solution}
Next, we use It\^{o}'s formula to verify the above candidate solution is indeed optimal, as stated in the following theorem: 

\begin{theorem}\label{thm_optimal_stopping}
If $\mathsf{mse}_{Y_i} \leq {\beta} < \mathsf{mse}_{\infty}$, then for all $v\in \mathbb{R}$, $H(v)$ in \eqref{eq_W1} is the value function of the optimal stopping problem \eqref{eq_stop_problem}. In addition,
the optimal stopping time for solving \eqref{eq_stop_problem} is $\tau_*$ in \eqref{eq_optimal_stopping123}, where $v_*=v(\beta)$ is given by \eqref{eq_threshold}. 
\end{theorem} 

In order to prove Theorem \ref{thm_optimal_stopping}, we need to establish the following properties of $H(v)$ in \eqref{eq_W1}, for the case that $\mathsf{mse}_{Y_i} \leq {\beta} < \mathsf{mse}_{\infty}$ is satisfied in \eqref{eq_stop_problem}:
\begin{lemma}\label{lem_stop1}
$H(v) = {\mathbb{E}_{v}} {\left[-\gamma  V_{\tau_*}^2 - \int_{0}^{\tau_*} (V_{s}^2-\beta)ds\right]}.$
\end{lemma}

\ifreport
\begin{proof}
See Appendix \ref{lem_stop_11}.
\end{proof}
\fi

\begin{lemma}\label{lem_stop2}
$H(v) \geq {-\gamma}v^2$ for all $v\in {\mathbb{R}}$.
\end{lemma}

\ifreport
\begin{proof}
See Appendix \ref{lem_stop_22}.
\end{proof}
\fi
A function $f(v)$ is said to be \emph{excessive} for the process $V_t$ if 
\begin{align}
{\mathbb{E}_{v}} f(V_t) \leq f(v), \forall t \geq 0, v\in{\mathbb{R}}.
\end{align}
By using It\^{o}'s formula in stochastic calculus, we can obtain
\begin{lemma}\label{lem_stop3}
The function $H(v)$ is excessive for the process $V_t$.
\end{lemma}

\ifreport
\begin{proof}
See Appendix \ref{lem_stop_33}.
\end{proof}
\fi
Now, we are ready to prove Theorem \ref{thm_optimal_stopping}.

\begin{proof}[Proof of Theorem \ref{thm_optimal_stopping}] In Lemmas \ref{lem_stop1}-\ref{lem_stop3}, we have shown that $H(v) = {\mathbb{E}_{v}} {\left[-\gamma  V_{\tau_*}^2 - \int_{0}^{\tau_*} (V_{s}^2-\beta)ds\right]}$, $H(v) \geq -\gamma v^2$, and $H(v)$ is an excessive function. 
Moreover, from  the proof of Lemma \ref{lem_stop1}, we know that $\mathbb{E}_{v} [\tau_* ] <\infty$ holds for all $v\in \mathbb{R}$. Hence, $\mathbb{P}_{v} (\tau_* < {\infty}) = 1$ for all $v\in \mathbb{R}$. 
These conditions and Theorem 1.11 in \cite[Section 1.2]{Peskir2006} imply that $\tau_*$ is an optimal stopping time of \eqref{eq_stop_problem}. This completes the proof.
\end{proof}
\ifreport
Because \eqref{eq_opt_stopping} is a special case of \eqref{eq_stop_problem}, we can get from Theorem \ref{thm_optimal_stopping} that 

\begin{corollary}\label{coro_stop}
If $\mathsf{mse}_{Y_i} \leq {\beta} < \mathsf{mse}_{\infty}$, then a solution to \eqref{eq_opt_stopping} is $(Z_1(\beta), $ $Z_2(\beta), \ldots)$, where
\begin{align}\label{eq_optimal_stopping1234}
Z_i(\beta) = \inf\{t \geq 0: |O_{Y_i+t}|\geq v(\beta) \},
\end{align}
and $v(\beta)$ is defined  in  \eqref{eq_threshold}.
\end{corollary} 

\fi

\subsection{Zero Duality Gap between \eqref{eq_SD} and 
\eqref{eq_dual}} \label{sec:dual}

Strong duality is established in the following thoerem:
\begin{theorem}\label{thm6_strong_duality}
If  the service times $Y_i$ are {i.i.d.} with $0<\mathbb{E}[Y_i]<\infty$, then the duality gap between \eqref{eq_SD} and \eqref{eq_dual} is zero. Further, $(Z_0(\beta), $ $Z_1(\beta), \ldots)$ is an optimal solution to both \eqref{eq_SD} and \eqref{eq_dual}, where $Z_i(\beta)$ is determined by
 \begin{align}\label{eq_Z_1}
 Z_i(\beta) = \inf\{t \geq 0: |O_{Y_i+t}|\geq v(\beta) \},
 \end{align} 
 $v(\beta)$ is defined  in  \eqref{eq_threshold}, $\beta\geq0$ is the root  of 
\begin{align}\label{thm2_eq22}
{\mathbb{E}\left[\int_{Y_i}^{Y_i+Z_i(\beta)+Y_{i+1}}\! O_t^2dt\right]}-{\beta} {\mathbb{E}[Y_i+Z_i(\beta)]} =0,
\end{align} 
if~$\mathbb{E}[Y_i+Z_i(\beta)] > {1}/{f_{\max}}$; otherwise, $\beta$ is the root  of $\mathbb{E}[Y_i+Z_i(\beta)] = {1}/{f_{\max}}$. In both cases, $\mathsf{mse}_{Y_i} \leq {\beta} < \mathsf{mse}_{\infty}$ is satisfied, and hence \eqref{eq_threshold} is well-defined. Further, the optimal objective value to \eqref{eq_Simple} is  given by \emph{
\begin{align}\label{thm2_eq23}
{\mathsf{mse}}_{\text{opt}} = \frac{\mathbb{E}\left[\int_{Y_i}^{Y_i+Z_i(\beta)+Y_{i+1}}\! O_t^2dt\right]}{\mathbb{E}[Y_i+Z_i(\beta)]}.
\end{align} }

\end{theorem}
\begin{proof}
We use \cite[Prop. 6.2.5]{Bertsekas2003} to find a \emph{geometric multiplier} for  \eqref{eq_SD}. This suggests  that the duality gap between \eqref{eq_SD} and \eqref{eq_dual} must be zero, because otherwise there exists no geometric multiplier \cite[Prop. 6.2.3(b)]{Bertsekas2003}. The details are provided in Appendix \ref{app_thm_strong_duality}.
\end{proof}
Hence, Theorem~\ref{thm_new} follows from Theorem \ref{thm6_strong_duality}. Because Theorem \ref{thm1} is a special case of Theorem \ref{thm_new}, Theorem \ref{thm1} is also proven. 


\ignore{\begin{figure}
	\centering
	\includegraphics[width=0.35\textwidth]{figure4_unconstrained_gaussian_discrete}   
	\caption{Average mutual information of the Gauss-Markov source versus the coefficient $a$ in \eqref{eq_markov}, where the service times $Y_i$ are equal to either $1$ or $21$ with probability 0.5.}
	\label{figure4_unconstrained_gaussian_discrete}
\end{figure} 
\begin{figure}
	\centering
	\includegraphics[width=0.5\textwidth]{figure10_lognormal_differentpenalties_exp}   
	\caption{MSE vs. the scale parameter  $\alpha$ of \emph{i.i.d.} log-normal service time for $f_{\max}=\infty$ and $\mathbb{E} [Y_i] = 1$. OU process parameters $\sigma=1$ and $\theta=0.2$ has been chosen.
	 }
	\label{figure10_lognormal_differentpenalties_exp}
\end{figure} 
}


\begin{figure}
	\centering
	\includegraphics[width=0.5\textwidth]{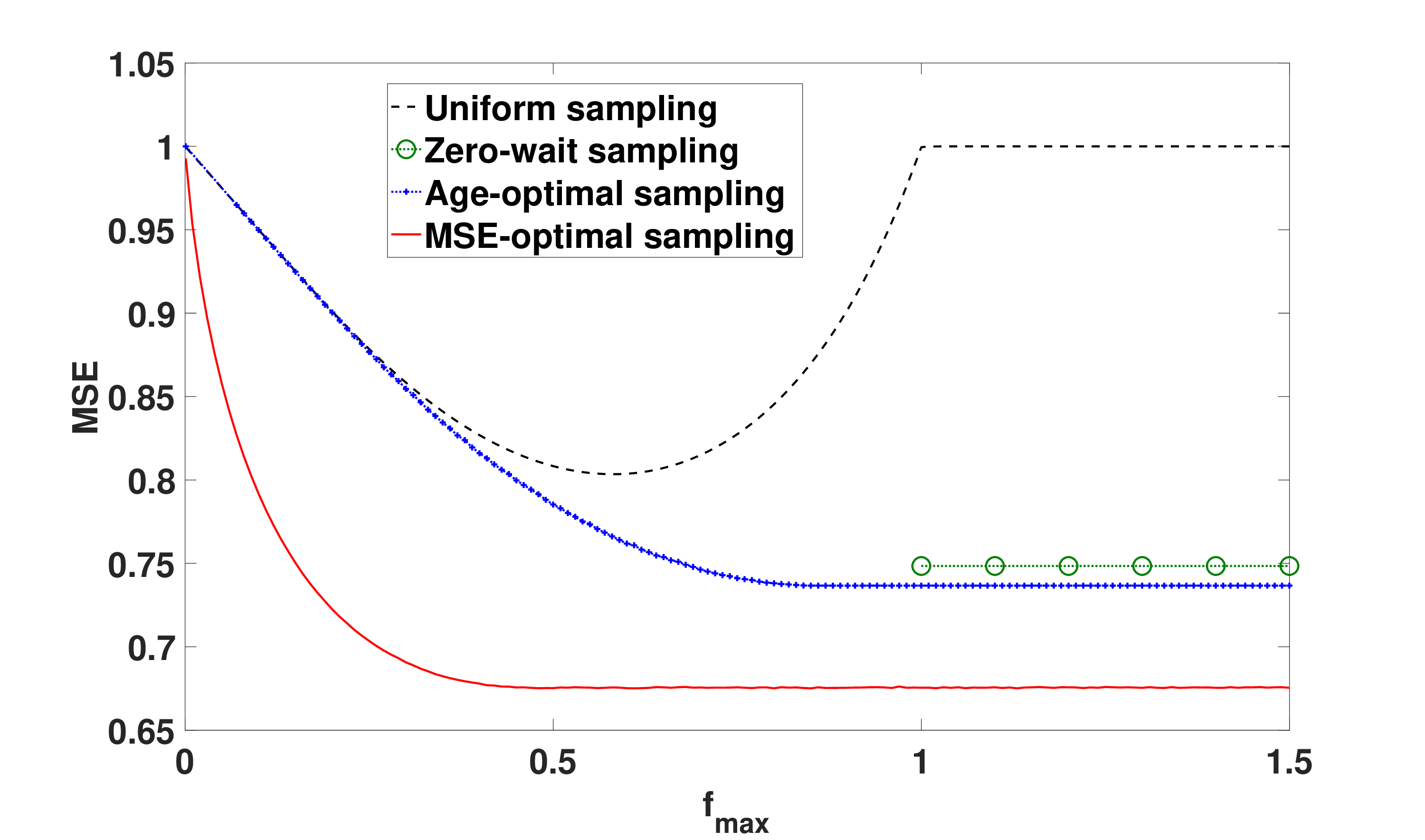}   
	\caption{MSE vs $f_{\max}$ tradeoff for \emph{i.i.d.} exponential service time with $\mathbb{E}[Y_i] = 1$, where the parameters of the OU process are $\sigma=1$ and $\theta=0.5$.}
	\label{figure1_exponential}
	\vspace{-3mm}
\end{figure}

\section{Numerical Comparisons}\label{sec:numerical}
In this section, we evaluate the estimation error achieved by the following four sampling policies:

\begin{itemize}
\item[1.] \textit{Uniform sampling}: Periodic sampling with a period given by $S_{i+1} - S_i = {1}/{f_{\max}}$.
\item[2.] \textit{Zero-wait sampling} \cite{AgeOfInfo2016,KaulYatesGruteser-Infocom2012}: The sampling policy given by
\begin{align} \label{eq_Zero_wait}
S_{i+1} = S_i+ Y_i,
\end{align} 
which is infeasible when $f_{\max} <1/\mathbb{E}[Y_i].$ 

\item[3.] \emph{Age-optimal sampling} \cite{SunNonlinear2019}: The sampling policy given by Theorem~\ref{thm2}.

\item[4.] \textit{MSE-optimal sampling}: The sampling policy given by Theorem~\ref{thm1}. 
\end{itemize}
Let ${\mathsf{mse}_{\text{uniform}}}$, ${\mathsf{mse}_{\text{zero-wait}}}$, ${\mathsf{mse}_{\text{age-opt}}}$, and ${\mathsf{mse}_{\text{opt}}}$, 
be the MSEs of uniform sampling, zero-wait sampling, age-optimal sampling,  MSE-optimal sampling, respectively. 
We can obtain  
\begin{align}
&{\mathsf{mse}_{Y_i}}\leq{\mathsf{mse}_{\text{opt}}}\leq{\mathsf{mse}_{\text{age-opt}}}\leq\mathsf{mse}_{\text{uniform}}\leq{\mathsf{mse}_{\infty}}, \nonumber\\
&
{\mathsf{mse}_{\text{age-opt}}}\leq{\mathsf{mse}_{\text{zero-wait}}}\leq{\mathsf{mse}_{\infty}},
\end{align}
whenever zero-wait sampling is feasible, which fit with our numerical results. The expectations in \eqref{eq_R_1} and \eqref{eq_R_2} are evaluated by taking the average over 1 million samples. The parameters of the OU process are given by $\sigma = 1$, $\theta = 0.5$, and $\mu$ can be chosen arbitrarily because it does not affect the estimation error. 

Figure \ref{figure1_exponential} illustrates the tradeoff between the MSE and $f_{\max}$ for \emph{i.i.d.} exponential service times with mean $\mathbb{E}[Y_i] = 1$. Because $ \mathbb{E}[Y_i] = 1$, the maximum throughput of the queue is  $1$. The lower bound $\mathsf{mse}_{Y_i}$ is 0.5 and the upper bound $\mathsf{mse}_{\infty}$ is $1$. In fact, as $Y_i$ is an exponential random variable with mean 1, $\frac{\sigma^2 }{2\theta}(1-e^{-2\theta Y_i})$ has a uniform distribution on $[0,1]$. Hence, $\mathsf{mse}_{Y_i} = 0.5$. For small values of $f_{\max}$, age-optimal sampling is similar to uniform sampling, and hence $\mathsf{mse}_\text{age-opt}$ and $\mathsf{mse}_\text{uniform}$ are close to each other in the regime. 
However, as $f_{\max}$ grows, $\mathsf{mse}_\text{uniform}$ reaches the upper bound $\mathsf{mse}_{\infty}$ and remains constant for $f_{\max} \geq 1$. This is because the queue length of uniform sampling is large at high sampling frequencies. 
In particular, when $f_{\max} \geq 1$, the queue length of uniform sampling is infinite.
On the other hand, $\mathsf{mse}_\text{age-opt}$ and $\mathsf{mse}_\text{opt}$ decrease with respect to $f_{\max}$. The reason behind this is that the set of feasible sampling policies satisfying the constraint in \eqref{eq_DPExpected} and \eqref{eq_age} becomes larger as $f_{\max}$ grows, and hence the optimal values of \eqref{eq_DPExpected} and \eqref{eq_age} are decreasing in $f_{\max}$. 
As we expected, $\mathsf{mse}_\text{zero-wait}$ is larger than $\mathsf{mse}_\text{opt}$ and $\mathsf{mse}_\text{age-opt}$. Moreover, all of them are between the lower bound $\mathsf{mse}_{Y_i}$ and upper bound $\mathsf{mse}_{\infty}$.

Figures \ref{figure2_lognormal_low_f} and \ref{figure2_lognormal_high_fmax} depict the MSE of \emph{i.i.d.} normalized log-normal service time for $f_{\max}=0.8$ and $f_{\max}=1.2$, respectively, where $Y_i={e^{{\alpha} {X_i}}}/{\mathbb{E}[e^{{\alpha} {X_i}}]}, {\alpha}>0$ is the scale parameter of log-normal distribution, and $(X_1, X_2, \dots)$ are \emph{i.i.d.} Gaussian random variables with zero mean and unit variance. Because $\mathbb{E}[Y_i]=1$, the maximum throughput of the queue is 1. In Fig. \ref{figure2_lognormal_low_f}, since $f_{\max}<1$, zero-wait sampling is not feasible and hence is not plotted. As the scale parameter $\alpha$ grows, the tail of the log-normal distribution becomes heavier.  

\begin{figure}
	\centering
	\includegraphics[width=0.5\textwidth]{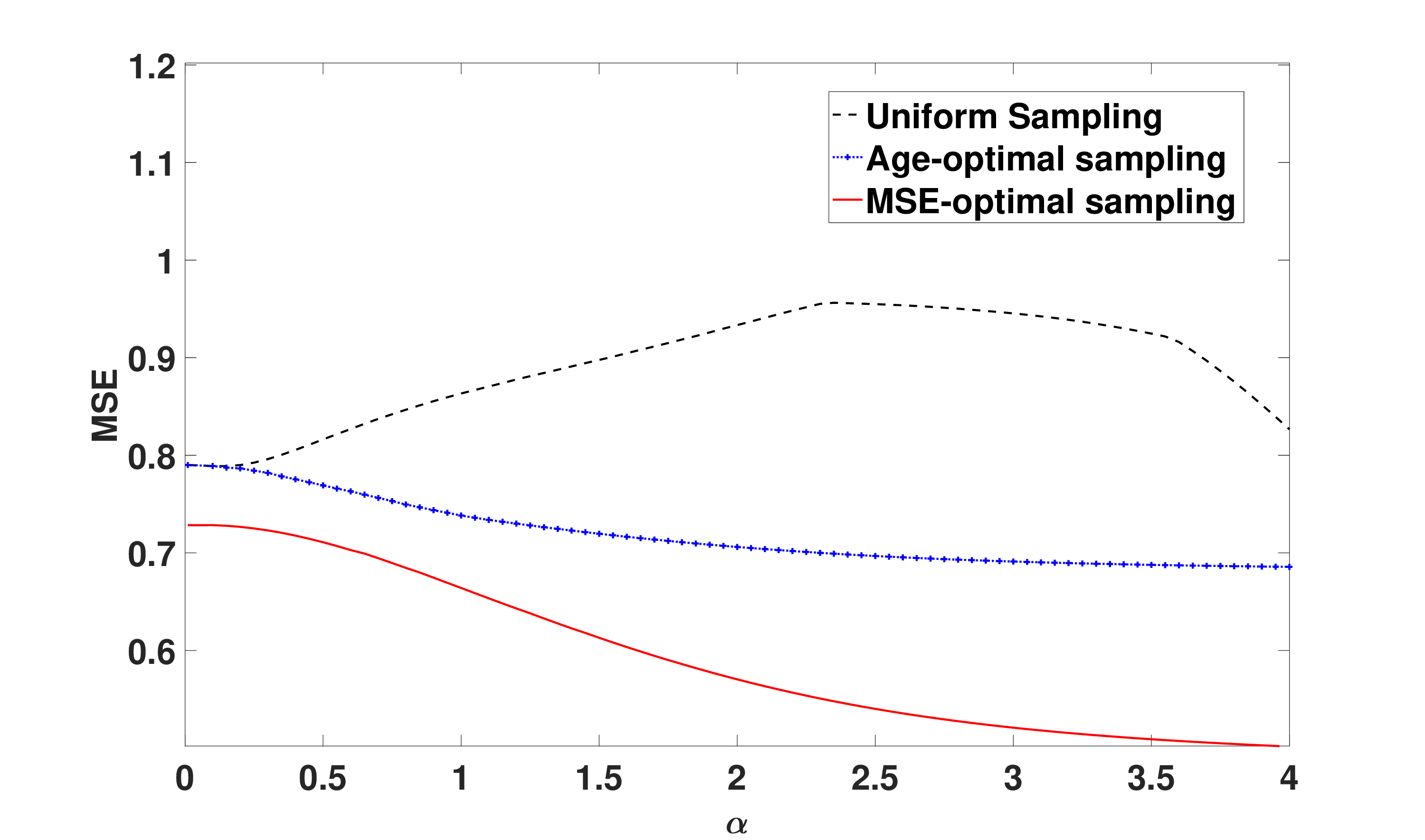}   
	\caption{MSE vs. the scale parameter ${\alpha}$ of \emph{i.i.d.} normalized log-normal service time distribution with $\mathbb{E}[Y_i]=1$ and $f_{\max}=0.8$, where the parameters of the OU process are $\sigma=1$ and $\theta=0.5$. Zero-wait sampling is not feasible here as $f_{\max} < {1}/{\mathbb{E}[Y_i]}$ and hence is not plotted.}
	\label{figure2_lognormal_low_f}
	\vspace{-3mm}
\end{figure}

In both figures, $\mathsf{mse}_{\text{age-opt}}$ and $\mathsf{mse}_{\text{opt}}$ drop with $\alpha$. This phenomenon may look surprising at first sight, because $\mathsf{mse}_{\text{age-opt}}$ and $\mathsf{mse}_{\text{opt}}$ grow quickly in $\alpha$ in the previous study \cite{2020Sun} on the Wiener process. To understand this phenomenon, let us consider the age penalty function $p(\Delta_t)$ in \eqref{eq_lem_estimation_error1} for the OU process. As the scale parameter $\alpha$ grows, the service time tends to become either shorter or much longer than the mean $\mathbb{E}[Y_i]$, rather than being close to $\mathbb{E}[Y_i]$. When $\Delta_t$ is small, $p(\Delta_t)$ reduces quickly in $\Delta_t$, and hence the service time smaller than $\mathbb{E} [Y_i]$ leads to a fast decrease in the average age penalty;  when $\Delta_t$ is quite large, $p(\Delta_t)$ cannot increase much because it is upper bounded by $\mathsf{mse}_{\infty}$, hence the service time much longer than $\mathbb{E} [Y_i]$ would not increase the average age penalty by much. By combining these two trends, the average age penalty $\mathsf{mse}_{\text{age-opt}}$ decreases in $\alpha$. The dropping of $\mathsf{mse}_{\text{opt}}$ in $\alpha$ can be understood in a similar fashion.
On the other hand, the age penalty function of the Wiener process is $p(\Delta_t) = \Delta_t$, which is quite different from the case considered here. We also observe that in both figures, the gap between $\mathsf{mse}_{\text{opt}}$ and $\mathsf{mse}_{\text{age-opt}}$ increases as $\alpha$ grows.

\begin{figure}
	\centering
	\includegraphics[width=0.5\textwidth]{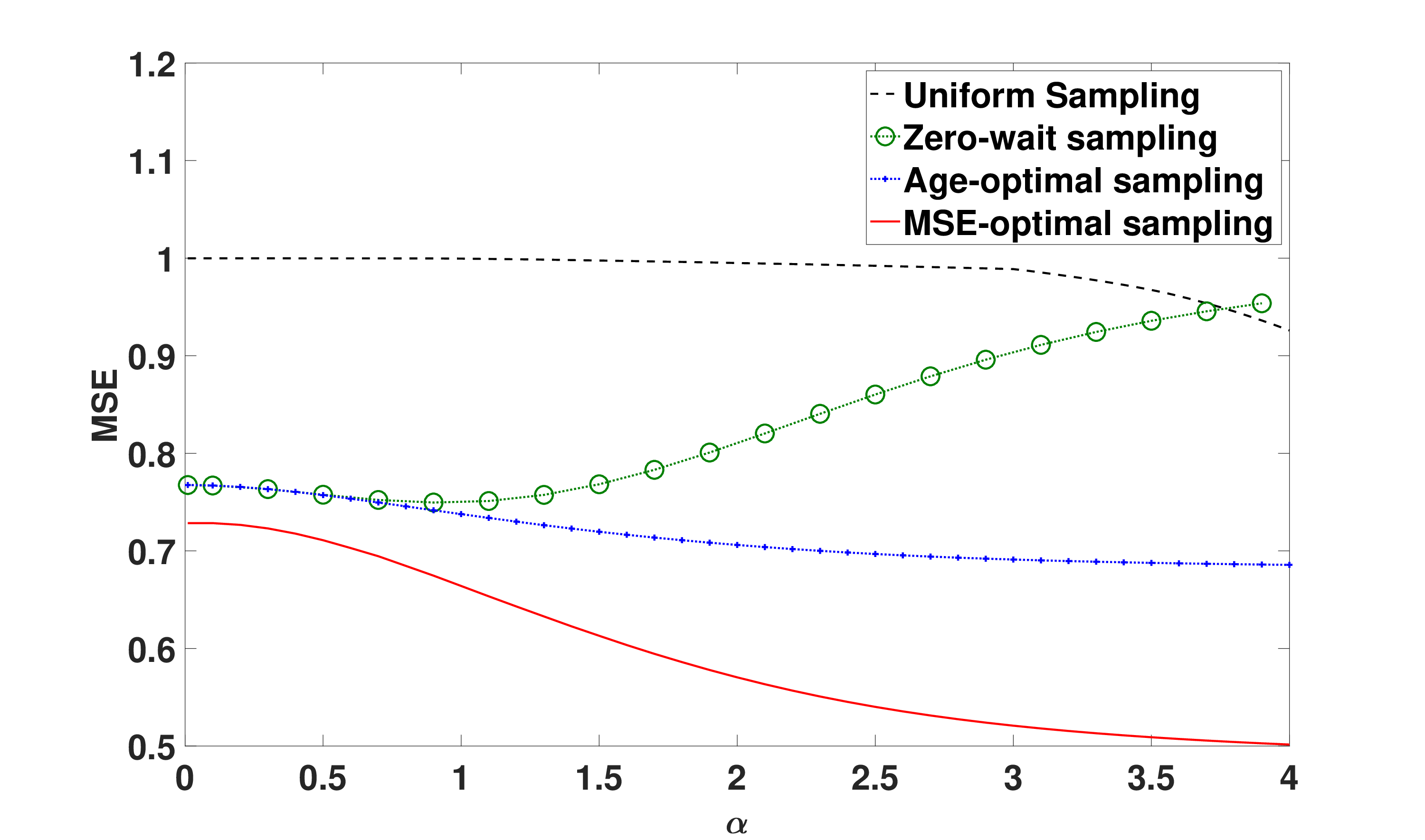}   
	\caption{MSE vs. the scale parameter ${\alpha}$ of \emph{i.i.d.} normalized log-normal service time distribution $\mathbb{E}[Y_i]=1$ and $f_{\max}=1.2$, where the parameters of the OU process are $\sigma=1$, $\theta=0.5$.} 
	\label{figure2_lognormal_high_fmax}
	\vspace{-3mm}
\end{figure}

\section{Conclusion} 
In this paper, we have studied the optimal sampler design for remote estimation of OU processes through queues. We have developed optimal causal sampling policies that minimize the estimation error of OU processes subject to a sampling rate constraint. These optimal sampling policies have nice structures and are easy to compute. A connection between remote estimation and nonlinear age metrics has been found. The structural properties of the optimal sampling policies shed lights on the possible structure of the optimal sampler designs for more general signal models, such as Feller processes, which is an important future research direction. 

\section*{Acknowledgement}
The authors are grateful to Thaddeus Roppel for a suggestion on this work. 


\bibliographystyle{IEEEtran}
\bibliography{ref,ref1,ref_2,sueh}

\begin{thebibliography}{10}
\providecommand{\url}[1]{#1}
\csname url@samestyle\endcsname
\providecommand{\newblock}{\relax}
\providecommand{\bibinfo}[2]{#2}
\providecommand{\BIBentrySTDinterwordspacing}{\spaceskip=0pt\relax}
\providecommand{\BIBentryALTinterwordstretchfactor}{4}
\providecommand{\BIBentryALTinterwordspacing}{\spaceskip=\fontdimen2\font plus
\BIBentryALTinterwordstretchfactor\fontdimen3\font minus
  \fontdimen4\font\relax}
\providecommand{\BIBforeignlanguage}[2]{{%
\expandafter\ifx\csname l@#1\endcsname\relax
\typeout{** WARNING: IEEEtran.bst: No hyphenation pattern has been}%
\typeout{** loaded for the language `#1'. Using the pattern for}%
\typeout{** the default language instead.}%
\else
\language=\csname l@#1\endcsname
\fi
#2}}
\providecommand{\BIBdecl}{\relax}
\BIBdecl

\bibitem{KaulYatesGruteser-Infocom2012}
S.~Kaul, R.~D. Yates, and M.~Gruteser, ``Real-time status: How often should one
  update?'' in \emph{IEEE INFOCOM}, 2012.

\bibitem{139341}
X.~Song and J.~W.~S. Liu, ``Performance of multiversion concurrency control
  algorithms in maintaining temporal consistency,'' in \emph{Proceedings.,
  Fourteenth Annual International Computer Software and Applications
  Conference}, Oct 1990, pp. 132--139.

\bibitem{3326507}
E.~Altman, R.~El-Azouzi, D.~S. Menasche, and Y.~Xu, ``Forever young: Aging
  control for hybrid networks,'' in \emph{ACM MobiHoc}, 2019.

\bibitem{2020Sun}
Y.~{Sun}, Y.~{Polyanskiy}, and E.~{Uysal}, ``Sampling of the {Wiener} process
  for remote estimation over a channel with random delay,'' \emph{IEEE
  Trans.~Inf. Theory}, vol.~66, no.~2, pp. 1118--1135, Feb 2020.

\bibitem{PhysRev.36.823}
G.~E. Uhlenbeck and L.~S. Ornstein, ``On the theory of the {Brownian} motion,''
  \emph{Phys. Rev.}, vol.~36, pp. 823--841, Sept. 1930.

\bibitem{Doob1942}
J.~L. Doob, ``The {Brownian} movement and stochastic equations,'' \emph{Annals
  of Mathematics}, vol.~43, no.~2, pp. 351--369, 1942.

\bibitem{Evans1994}
L.~Evans, S.~Keef, and J.~Okunev, ``Modelling real interest rates,''
  \emph{Journal of Banking and Finance}, vol.~18, no.~1, pp. 153 -- 165, 1994.

\bibitem{41c4ac91}
A.~Cika, M.~Badiu, and J.~Coon, ``Quantifying link stability in ad hoc wireless
  networks subject to {Ornstein-Uhlenbeck} mobility,'' in \emph{IEEE ICC},
  2019.

\bibitem{Kim2018}
H.~Kim, J.~Park, M.~Bennis, and S.~Kim, ``Massive {UAV}-to-ground communication
  and its stable movement control: A mean-field approach,'' in \emph{IEEE
  SPAWC}, June 2018, pp. 1--5.

\bibitem{Nuno2011}
G.~M. Lipsa and N.~C. Martins, ``Remote state estimation with communication
  costs for first-order {LTI} systems,'' \emph{IEEE Trans.~Auto.~Control},
  vol.~56, no.~9, pp. 2013--2025, Sept. 2011.

\bibitem{Vinogradov_2018}
E.~Vinogradov, H.~Sallouha, S.~D. Bast, M.~M. Azari, and S.~Pollin, ``Tutorial
  on uavs: A blue sky view onwireless communication,'' \emph{Journal of Mobile
  Multimedia}, vol.~14, no.~4, p. 395–468, 2018.

\bibitem{Peskir2006}
G.~Peskir and A.~N. Shiryaev, \emph{Optimal Stopping and Free-Boundary
  Problems}.\hskip 1em plus 0.5em minus 0.4em\relax Basel, Switzerland:
  Birkh\"{a}uswer Verlag, 2006.

\bibitem{Bernt2000}
B.~{\O}ksendal, \emph{Stochastic Differential Equations: An Introduction with
  Applications}, 5th~ed.\hskip 1em plus 0.5em minus 0.4em\relax Springer-Verlag
  Berlin Heidelberg, 2000.

\bibitem{SunNonlinear2019}
Y.~Sun and B.~Cyr, ``Sampling for data freshness optimization: Non-linear age
  functions,'' \emph{J. Commun. Netw.}, vol.~21, no.~3, pp. 204--219, 2019.

\bibitem{SunSPAWC2018}
------, ``Information aging through queues: A mutual information perspective,''
  in \emph{IEEE SPAWC Workshop}, 2018.

\bibitem{Xiao2018}
Y.~Xiao and Y.~Sun, ``A dynamic jamming game for real-time status updates,'' in
  \emph{IEEE INFOCOM AoI Workshop}, April 2018, pp. 354--360.

\bibitem{Ahmed2019}
A.~M. Bedewy, Y.~Sun, S.~Kompella, and N.~B. Shroff, ``Age-optimal sampling and
  transmission scheduling in multi-source systems,'' in \emph{ACM MobiHoc},
  2019.

\bibitem{KamTIT2016}
C.~Kam, S.~Kompella, G.~D. Nguyen, and A.~Ephremides, ``Effect of message
  transmission path diversity on status age,'' \emph{IEEE Trans.\ Inf.\
  Theory}, vol.~62, no.~3, pp. 1360--1374, Mar. 2016.

\bibitem{AgeOfInfo2016}
Y.~Sun, E.~Uysal-Biyikoglu, R.~D. Yates, C.~E. Koksal, and N.~B. Shroff,
  ``Update or wait: How to keep your data fresh,'' \emph{IEEE Trans.~Inf.
  Theory}, vol.~63, no.~11, pp. 7492--7508, Nov. 2017.

\bibitem{KamTIT2018}
C.~Kam, S.~Kompella, G.~D. Nguyen, J.~E. Wieselthier, and A.~Ephremides, ``On
  the age of information with packet deadlines,'' \emph{IEEE Trans.\ Inf.\
  Theory}, vol.~64, no.~9, pp. 6419--6428, Sept. 2018.

\bibitem{yates2019age}
R.~D. Yates and S.~K. Kaul, ``The age of information: Real-time status updating
  by multiple sources,'' \emph{IEEE Trans. Inf. Theory}, vol.~65, no.~3, pp.
  1807--1827, Mar. 2019.

\bibitem{HeTIT2018}
Q.~He, D.~Yuan, and A.~Ephremides, ``Optimal link scheduling for age
  minimization in wireless systems,'' \emph{IEEE Trans.\ Inf.\ Theory},
  vol.~64, no.~7, pp. 5381--5394, July 2018.

\bibitem{JooTON2018}
C.~Joo and A.~Eryilmaz, ``Wireless scheduling for information freshness and
  synchrony: Drift-based design and heavy-traffic analysis,'' \emph{IEEE/ACM
  Trans. Netw.}, vol.~26, no.~6, pp. 2556--2568, Dec 2018.

\bibitem{BedewyJournal2017}
A.~M. Bedewy, Y.~Sun, and N.~B. Shroff, ``Minimizing the age of the information
  through queues,'' \emph{IEEE Trans.~Inf. Theory}, vol.~65, no.~8, pp.
  5215--5232, Aug 2019.

\bibitem{BedewyJournal2017_2}
------, ``The age of information in multihop networks,'' \emph{IEEE/ACM Trans.
  Netw.}, vol.~27, no.~3, pp. 1248--1257, June 2019.

\bibitem{multiflow18}
Y.~Sun, E.~Uysal-Biyikoglu, and S.~Kompella, ``Age-optimal updates of multiple
  information flows,'' in \emph{IEEE INFOCOM AoI Workshop}, 2018.

\bibitem{KadotaINFOCOM2018}
I.~Kadota, A.~Sinha, and E.~Modiano, ``Optimizing age of information in
  wireless networks with throughput constraints,'' in \emph{IEEE INFOCOM},
  April 2018, pp. 1844--1852.

\bibitem{Talak:2018}
R.~Talak, S.~Karaman, and E.~Modiano, ``Optimizing information freshness in
  wireless networks under general interference constraints,'' in \emph{ACM
  MobiHoc}, 2018.

\bibitem{Lu:2018}
N.~Lu, B.~Ji, and B.~Li, ``Age-based scheduling: Improving data freshness for
  wireless real-time traffic,'' in \emph{ACM MobiHoc}, 2018.

\bibitem{maatouk2020status}
A.~Maatouk, Y.~Sun, A.~Ephremides, and M.~Assaad, ``Status updates with
  priorities: Lexicographic optimality,'' in \emph{IEEE/IFIP WiOpt}, 2020.

\bibitem{ZhouTCOM2019}
B.~{Zhou} and W.~{Saad}, ``Joint status sampling and updating for minimizing
  age of information in the {Internet} of things,'' \emph{IEEE Trans. Commun.},
  vol.~67, no.~11, pp. 7468--7482, Nov 2019.

\bibitem{ZhouTWC2020}
------, ``Minimum age of information in the {Internet} of things with
  non-uniform status packet sizes,'' \emph{IEEE Trans. Wireless Commun.},
  vol.~19, no.~3, pp. 1933--1947, 2020.

\bibitem{NET-060}
A.~Kosta, N.~Pappas, and V.~Angelakis, \emph{Age of Information: A New Concept,
  Metric, and Tool}.\hskip 1em plus 0.5em minus 0.4em\relax Now Publishers Inc,
  2018.

\bibitem{8940930_1}
Y.~{Sun}, I.~{Kadota}, R.~{Talak}, and E.~{Modiano}, \emph{Age of Information:
  A New Metric for Information Freshness}.\hskip 1em plus 0.5em minus
  0.4em\relax Morgan \& Claypool, 2019.

\bibitem{Ahmed2020ArXiv}
A.~M. Bedewy, Y.~Sun, S.~Kompella, and N.~B. Shroff, ``Optimal sampling and
  scheduling for timely status updates in multi-source networks,'' 2020,
  accepted by \emph{IEEE Trans. Inf. Theory}.

\bibitem{Hajek2008}
B.~Hajek, K.~Mitzel, and S.~Yang, ``Paging and registration in cellular
  networks: Jointly optimal policies and an iterative algorithm,'' \emph{IEEE
  Trans.~Inf. Theory}, vol.~54, no.~2, pp. 608--622, Feb 2008.

\bibitem{Rabi2012}
M.~Rabi, G.~V. Moustakides, and J.~S. Baras, ``Adaptive sampling for linear
  state estimation,'' \emph{SIAM Journal on Control and Optimization}, vol.~50,
  no.~2, pp. 672--702, 2012.

\bibitem{nayyar2013}
A.~Nayyar, T.~Ba\c{s}ar, D.~Teneketzis, and V.~V. Veeravalli, ``Optimal
  strategies for communication and remote estimation with an energy harvesting
  sensor,'' \emph{IEEE Trans.~Auto.~Control}, vol.~58, no.~9, pp. 2246--2260,
  Sept. 2013.

\bibitem{Basar2014}
K.~Nar and T.~Ba\c{s}ar, ``Sampling multidimensional {Wiener} processes,'' in
  \emph{IEEE CDC}, Dec. 2014, pp. 3426--3431.

\bibitem{GAO201857}
X.~Gao, E.~Akyol, and T.~Ba\c{s}ar, ``Optimal communication scheduling and
  remote estimation over an additive noise channel,'' \emph{Automatica},
  vol.~88, pp. 57 -- 69, 2018.

\bibitem{ChakravortyTAC2020}
J.~{Chakravorty} and A.~{Mahajan}, ``Remote estimation over a packet-drop
  channel with {Markovian} state,'' \emph{IEEE Trans. Auto. Control}, vol.~65,
  no.~5, pp. 2016--2031, 2020.

\bibitem{TsaiINFOCOM2020}
C.-H. Tsai and C.-C. Wang, ``Unifying {AoI} minimization and remote estimation:
  Optimal sensor/controller coordination with random two-way delay,'' in
  \emph{IEEE INFOCOM}, 2020.

\bibitem{GuoISIT2020}
N.~Guo and V.~Kostina, ``Optimal causal rate-constrained sampling for a class
  of continuous {Markov} processes,'' in \emph{IEEE ISIT}, 2020.

\bibitem{jog2019channels}
V.~Jog, R.~J. La, and N.~C. Martins, ``Channels, learning, queueing and remote
  estimation systems with a utilization-dependent component,'' 2019, coRR,
  abs/1905.04362.

\bibitem{Durrettbook10}
R.~Durrett, \emph{Probability: Theory and Examples}, 4th~ed.\hskip 1em plus
  0.5em minus 0.4em\relax Cambridge Univerisity Press, 2010.

\bibitem{Haas2002}
P.~J. Haas, \emph{Stochastic Petri Nets: Modelling, Stability,
  Simulation}.\hskip 1em plus 0.5em minus 0.4em\relax New York, NY: Springer
  New York, 2002.

\bibitem{Ross1970}
S.~M. Ross, \emph{Applied Probability Models with Optimization
  Applications}.\hskip 1em plus 0.5em minus 0.4em\relax San Francisco, CA:
  Holden-Day, 1970.

\bibitem{Mine1970}
H.~Mine and S.~Osaki, \emph{{Markovian} Decision Processes}.\hskip 1em plus
  0.5em minus 0.4em\relax New York: Elsevier, 1970.

\bibitem{Hayman1984}
D.~Hayman and M.~Sobel, \emph{Stochastic models in Operations Research, Volume
  II: Stochastic Optimizations}.\hskip 1em plus 0.5em minus 0.4em\relax New
  York: McGraw-Hill, 1984.

\bibitem{Feinberg1994}
E.~A. Feinberg, ``\BIBforeignlanguage{English}{Constrained semi-{Markov}
  decision processes with average rewards},''
  \emph{\BIBforeignlanguage{English}{Zeitschrift f\"ur Operations Research}},
  vol.~39, no.~3, pp. 257--288, 1994.

\bibitem{Bertsekas2005bookDPVol1}
D.~P. Bertsekas, \emph{Dynamic Programming and Optimal Control}, 3rd~ed.\hskip
  1em plus 0.5em minus 0.4em\relax Belmont, MA: Athena Scientific, 2005,
  vol.~1.

\bibitem{Maller2009}
R.~A. Maller, G.~M{\"u}ller, and A.~Szimayer, ``{Ornstein-Uhlenbeck} processes
  and extensions,'' in \emph{Handbook of Financial Time Series}, T.~Mikosch,
  J.-P. Krei{\ss}, R.~A. Davis, and T.~G. Andersen, Eds.\hskip 1em plus 0.5em
  minus 0.4em\relax Berlin, Heidelberg: Springer Berlin Heidelberg, 2009, pp.
  421--437.

\bibitem{SOLEYMANI20161}
T.~Soleymani, S.~Hirche, and J.~S. Baras, ``Optimal information control in
  cyber-physical systems,'' \emph{IFAC-PapersOnLine}, vol.~49, no.~22, pp. 1 --
  6, 2016.

\bibitem{2007247}
I.~Gradshteyn and I.~Ryzhik, \emph{Table of Integrals, Series, and Products},
  7th~ed.\hskip 1em plus 0.5em minus 0.4em\relax Academic Press, 2007.

\bibitem{200720017}
F.~W. Olver, D.~W. Lozier, R.~F. Boisvert, and C.~W. Clark, \emph{NIST Handbook
  of Mathematical Functions}.\hskip 1em plus 0.5em minus 0.4em\relax Cambridge
  University Press, 2010.

\bibitem{Mathews2004Numerical}
J.~H. Mathews and K.~K. Fink, \emph{Numerical Methods Using Matlab}.\hskip 1em
  plus 0.5em minus 0.4em\relax Simon \& Schuster, Inc., 1998.

\bibitem{CCWang2020}
C.-H. Tsai and C.-C. Wang, ``Age-of-information revisited: Two-way delay and
  distribution-oblivious online algorithm,'' in \emph{IEEE ISIT}, 2020.

\bibitem{spivak08}
M.~Spivak, \emph{Calculus}, 4th~ed.\hskip 1em plus 0.5em minus 0.4em\relax
  Publish or Perish, 2008.

\bibitem{Dinkelbach67}
W.~Dinkelbach, ``On nonlinear fractional programming,'' \emph{Management
  Science}, vol.~13, no.~7, pp. 492--498, 1967.

\bibitem{Bertsekas2003}
D.~P. Bertsekas, A.~Nedi\'c, and A.~E. Ozdaglar, \emph{Convex Analysis and
  Optimization}.\hskip 1em plus 0.5em minus 0.4em\relax Belmont, MA: Athena
  Scientific, 2003.

\bibitem{Borodin1996}
A.~N. Borodin and P.~Salminen, \emph{Handbook of {Brownian} Motion -- Facts and
  Formulae}.\hskip 1em plus 0.5em minus 0.4em\relax Basel, Switzerland:
  Birkh\"{a}uswer Verlag, 1996.

\bibitem{Ross1996}
S.~M. Ross, \emph{Stochastic Processes}, 2nd~ed.\hskip 1em plus 0.5em minus
  0.4em\relax John Wiley\& Sons, Inc., 1996.

\bibitem{Amann1990}
H.~Amann, \emph{Ordinary Differential Equations An Introduction to Nonlinear
  Analysis}.\hskip 1em plus 0.5em minus 0.4em\relax Berlin: Walter De Gruyter,
  1990.

\bibitem{strangcalculus}
G.~Strang, \emph{Calculus}.\hskip 1em plus 0.5em minus 0.4em\relax
  Wellesley-Cambridge Press, 1991.

\bibitem{Liggett2010}
T.~M. Liggett, \emph{Continuous Time Markov Processes: An Introduction}.\hskip
  1em plus 0.5em minus 0.4em\relax Providence, Rhode Island: American
  Mathematical Society, 2010.

\bibitem{Zhao2020}
C.~{Jia} and G.~{Zhao}, ``Moderate maximal inequalities for the
  {Ornstein-Uhlenbeck} process,'' \emph{Proc. Amer. Math. Soc.}, vol. 148, pp.
  3607--3615, 2020.

\bibitem{Poor:1994}
H.~V. Poor, \emph{An Introduction to Signal Detection and Estimation},
  2nd~ed.\hskip 1em plus 0.5em minus 0.4em\relax New York, NY, USA:
  Springer-Verlag New York, Inc., 1994.

\bibitem{JEFFREY1995}
A.~Jeffrey and H.-H. Dai, \emph{Handbook of Mathematical Formulas and
  Integrals}.\hskip 1em plus 0.5em minus 0.4em\relax Academic Press, 1995.

\bibitem{BMbook10}
P.~Morters and Y.~Peres, \emph{Brownian Motion}.\hskip 1em plus 0.5em minus
  0.4em\relax Cambridge Univerisity Press, 2010.

\end{thebibliography}

\appendices
%
%
%
%
%

\section{Proof of Lemma \ref{beta_unique}}\label{beta_unique_proof}
Part (i): According to \eqref{lagrange} and \eqref{eq_beta_value1}, the Lagrangian $L(\pi,\beta)$ is linear and strictly decreasing in $\beta$. Further, \eqref{eq_primal} tells us that $f(\beta)$ is the infimum of $L(\pi,\beta)$ among all policies $\pi \in {\Pi}_{1}$. Because the infimum of a linear and strictly decreasing function is concave and strictly decreasing, $f(\beta)$ is concave and strictly decreasing in $\beta$. Moreover, because $f(\beta)$ is concave, it is also continuous.

Part (ii): We first show that $f(\mathsf{mse}_{Y_i}) > 0$. 
According to \eqref{eq_threshold}, $v(\mathsf{mse}_{Y_i}) = 0$. This, together with \eqref{eq_R_1}, \eqref{eq_R_2}, and \eqref{two_root_1}, implies 
\begin{align}
f(\mathsf{mse}_{Y_i})=
& f_1(\mathsf{mse}_{Y_i}) - {\mathsf{mse}_{Y_i}} f_2(\mathsf{mse}_{Y_i}) \nonumber\\
=& \mathsf{mse}_{\infty} \{\mathbb{E} [Y_i] - {\gamma}\} + \mathbb{E} [O_{Y_i}^2] {\gamma} - \mathsf{mse}_{Y_i} \mathbb{E} [Y_i] \nonumber\\
=& \frac{{\sigma}^2}{2 {\theta}} \bigg\{ {\mathbb{E} [Y_i] - \frac{1}{2 {\theta}} \mathbb{E} {\left[1- e^{-2 {\theta} {Y_i}}\right]}} \nonumber\\
+ & {\frac{1}{2 {\theta}} \bigg\{\mathbb{E} {\left[1- e^{-2 {\theta} {Y_i}}\right]}\bigg\}^2} - {\mathbb{E} {\left[1- e^{-2 {\theta} {Y_i}}\right]} \mathbb{E} [Y_i]} \bigg\}.
\end{align}
Therefore, it suffices to prove that 
\begin{align} \label{beta_root_1}
& {\mathbb{E} [Y_i] - \frac{1}{2 {\theta}} \mathbb{E} [1- e^{-2 {\theta} {Y_i}}] + \frac{1}{2 {\theta}} \bigg\{\mathbb{E} [1- e^{-2 {\theta} {Y_i}}]\bigg\}^2} \nonumber\\
& {- \mathbb{E} [1- e^{-2 {\theta} {Y_i}}] \mathbb{E} [Y_i]} > 0,
\end{align}
which can be simplified as
\begin{align}
\left(\mathbb{E} [Y_i] - \frac{1}{2 {\theta}} \mathbb{E} [1 - e^{-2 {\theta} {Y_i}}] \right) \mathbb{E}[e^{-2 {\theta} {Y_i}}] > 0.
\end{align}
Because $x > 1 -e^{-x}$ for all $x >0$ and $\mathbb{E} [Y_i] >0$, we get
\begin{align} \label{pos}
\mathbb{E} [2 {\theta} {Y_i}] > \mathbb{E} [1 - e^{-2 {\theta} {Y_i}}].
\end{align}
By this, $f(\mathsf{mse}_{Y_i}) > 0$ is proven.


Finally, we prove that $f(\mathsf{mse}_{\infty}) < 0$. When $\beta \to \mathsf{mse}^{-}_{\infty}$, \eqref{eq_threshold} tells us that $v(\beta)$ grows to infinite. Further, according to \eqref{eq_R_1} and  \eqref{eq_R_2}, $R_1(v(\beta))$ and $R_2(v(\beta))$ are quite large compared to $R_1(O_{Y_i})$ and $R_2(O_{Y_i})$. Therefore, 
\begin{align} \label{neg}
\lim_{\!\!\!\!{{\beta} {\to} {\mathsf{mse}^{-}_{\infty}}}\!\!\!\!} {f(\beta)} = & - {\frac{1}{2 {\theta}}} \lim_{{\beta} {\to} {\mathsf{mse}^{-}_{\infty}}}  {v^2(\beta)} \mathbb{E} [e^{-2 {\theta} {Y_i}}] - {\mathsf{mse}_{\infty} {\gamma}} \nonumber\\
 = & -{\infty}.
\end{align}
This completes the proof.

\section{Proof of Lemma \ref{beta_unique_2}}\label{beta_unique_proof_2}


Part (i): From \eqref{eq_expectation1}, it is evident that the function $f_2(\beta)$ is continuous and hence, from \eqref{g}, $g(\beta)$ is also continuous. 
The derivatives of $R_1(v)$ in \eqref{eq_R_1}  and $v(\beta)$ in \eqref{eq_threshold} are given by
\begin{align} 
\label{R_1'}
R'_1 (v) =& {\frac{\sqrt{\pi}}{{\sigma} {\sqrt{\theta}}}} {\text{erf}}{\bigg( \frac{\sqrt{\theta}}{\sigma} v \bigg)} e^{\frac{\theta}{{\sigma}^2} v^2},\\
\label{v'}
v'(\beta) = &\frac{\sigma}{\sqrt{\theta}} \bigg\{G^{-1} \left(\frac{\mathsf{mse}_{\infty} - \mathsf{mse}_{Y_i}}{\mathsf{mse}_{\infty} - {\beta}}\right) \bigg\}'.
\end{align}
Denote 
\begin{align}
G^{-1} \bigg(\frac{\mathsf{mse}_{\infty} - \mathsf{mse}_{Y_i}}{\mathsf{mse}_{\infty} - {\beta}}\bigg) = y.
\end{align}
Then, by using the derivative of inverse function \cite{strangcalculus}, $v'(\beta)$ in \eqref{v'} becomes
\begin{align}\label{eq_v_derivative}
v'(\beta) = \frac{\sigma}{\sqrt{\theta}} \frac{1}{G'(y)} \frac{\mathsf{mse}_{\infty} - \mathsf{mse}_{Y_i}}{(\mathsf{mse}_{\infty} - {\beta})^2},
\end{align}
where 
\begin{align}
 G'(x) = \sqrt{\pi} e^{x^2} \text{erf} (x) - \frac{e^{x^2}}{x^2} \frac{\sqrt{\pi}}{2} \text{erf} (x) + \frac{1}{x}> 0
 \end{align}
 for all $x>0$. Hence, $v(\beta)$ is strictly increasing in $\beta$. From \eqref{R_1'}, we know $R_1'(v) > 0 $, i.e., $R_1(v)$ is strictly increasing in $v$.  Therefore, $R_1(v(\beta))$ is strictly increasing in $\beta$. This further implies that in \eqref{eq_expectation1}, $\max\{R_1 (v(\beta)) - R_1 (O_{Y_i}), 0\}$ is strictly increasing in $\beta$. Therefore, $\mathbb{E}[\max\{R_1 (v(\beta)) - R_1 (O_{Y_i}), 0\}]$ is also strictly increasing in $\beta$ and hence, $f_2(\beta)$ is strictly increasing in $\beta$. Then, by \eqref{g}, $g(\beta)$ is strictly decreasing in $\beta$.
 This completes the proof.

Part (ii): We first show that $g(\mathsf{mse}_{Y_i}) {\geq} 0$. 
If the root of \eqref{thm1_eq22} does not satisfy \eqref{thm1_eq24}, then, let $\beta^{*}$ is the root of \eqref{g}. Therefore, $g(\beta^{*}) = 0$. As $\mathsf{mse}_{Y_i} \leq \beta \leq \mathsf{mse}_{\infty}$ and from part (i), $g(\beta)$ is strictly decreasing in $\beta$, we get that
\begin{align}
g(\mathsf{mse}_{Y_i}) \geq g(\beta^{*}) =0.
\end{align} 
Hence, $g(\mathsf{mse}_{Y_i}) \geq 0$.

Finally, as ${\beta} {\to} {\mathsf{mse}^{-}_{\infty}}$, because $v(\beta)$ grows to infinite, $R_1(v(\beta))$ becomes quite large compared to $R_1(O_{Y_i})$. Hence, 
\begin{align}
\lim_{{\beta} {\to} {\mathsf{mse}^{-}_{\infty}}} {g(\beta)} = & \frac{1}{f_{\max}} - \lim_{{\beta} {\to} {\mathsf{mse}^{-}_{\infty}}} R_1(v(\beta)) \nonumber\\
= & -{\infty}.
\end{align}
This complete the proof.

\section{Proof of (\ref{eq_esti})}\label{app_MMSE}
The MMSE estimator $\hat X_t$ can be expressed as 
\begin{align}\label{eq_app_MMSE1}
\hat{X}_{t} & =  \mathbb{E}[{X}_t | M_t ] \nonumber\\
& =\mathbb{E}[{X}_t | \{(S_j, X_{S_j}, D_j): D_j \leq t\}]. 
\end{align}
Recall that $U_t= \max\{S_{i}: D_{i} \leq t\}$ is the generation time of the latest received sample at time $t$.
According to  the strong Markov property of $X_t$ \cite[Eq. (4.3.27)]{Peskir2006} and the fact that the $Y_i$'s are independent of $\{X_t, t\geq 0\}$, $\{U_t, X_{U_t}\}$ is a sufficient statistic for estimating $X_t$ based on $\{(S_j, X_{S_j}, D_j): D_j \leq t\}$. If $t\in[D_i,D_{i+1})$,  \eqref{eq_age2} suggests that $U_t=S_i$ and $X_{U_t}=X_{S_i}$. This and \eqref{eq_process} tell us that, if $t\in[D_i,D_{i+1})$, then
\begin{align}
\hat{X}_{t}  =  ~&\mathbb{E}[{X}_t | \{(S_i, X_{S_i}, D_i): D_i \leq t\}] \nonumber\\ 
=~&\mathbb{E}[{X}_t | S_i,X_{S_i}] \nonumber\\ 
=~&X_{S_i} e^{-\theta (t-S_i)}+ \mu\big[1-e^{-\theta (t-S_i)} \big].
\end{align}
This completes the proof.

\section{Proof of Lemma \ref{lem_estimation_error}}\label{app_lem_estimation_error}
In any signal-ignorant policy, because the sampling times $S_i$ and the service times $Y_i$ are both independent of $\{X_t, t\geq 0\}$, the delivery times $D_i$ are also independent of $\{X_t, t\geq 0\}$.
Hence, for any $t\in[ D_i,D_{i+1})$,
\begin{align}\label{eq_lem2_1}
&~\mathbb{E}\left[(X_t-\hat X_t)^2\big|S_i,D_i,D_{i+1}\right] \nonumber\\
\overset{(a)}{=}&~ \mathbb{E}\left[\frac{\sigma^2}{{2\theta}}e^{-2\theta (t-S_i)} W^2_{e^{2 \theta (t-S_i)}-1}\bigg| S_i,D_i,D_{i+1}\right]\nonumber\\
\overset{(b)}{=}&~ \frac{\sigma^2}{{2\theta}}\left[1-e^{-2\theta (t-S_i)}\right], 
\end{align}
where Step (a) is due to \eqref{eq_process}-\eqref{eq_esti} and Step (b) is due to $\mathbb{E}[W_t^2] = t$ for all constant $t\geq 0$. We note that in signal-aware sampling policies, 
\begin{align}
(X_t-\hat X_t)^2=\frac{\sigma^2}{{2\theta}}e^{-2\theta (t-S_i)} W^2_{e^{2 \theta (t-S_i)}-1}
\end{align} 
could be correlated with $(S_i,D_i,D_{i+1})$ and hence Step (b) of \eqref{eq_lem2_1} may not hold.
Substituting \eqref{eq_age2} into \eqref{eq_lem2_1}, yields that for all $t\geq 0$
\begin{align}
\mathbb{E}\left[(X_t-\hat X_t)^2\big|\pi, Y_1,Y_2, \ldots\right]= \frac{\sigma^2}{{2\theta}}\left(1-e^{-2\theta \Delta_t}\right),
\end{align}
which is strictly increasing in $\Delta_t$. This completes the proof. 

\section{Proof of (\ref{eq_Wiener})}\label{app_Wiener}

When $\sigma = 1$, \eqref{eq_threshold11} can be expressed as 
\begin{align}\label{Eq_v}
 \left({1} - {{2{\theta}{\beta}}{}}\right) G\left({\sqrt \theta}{}v\right) = \EE \left[e^{-2\theta Y_{i}}\right],
\end{align}
The error function  $\text{erf}(x)$
has a Maclaurin series representation, given by
\begin{align}
& \text{erf}(x) = {\frac{2}{\sqrt{\pi}}} \left[x - \frac{x^3}{3} + o(x^3)\right]. 
\end{align}
Hence, the Maclaurin series representation of $G(x)$ in \eqref{eq_g1} is
\begin{align}
& G(x) =  1 + \frac{2x^2}{3} + o(x^2). 
\end{align}
Let $x = \sqrt{\theta}v$, we get
\begin{align}
& G\left(\sqrt{\theta}v\right) = 1+ \frac{2}{3} {\theta} v^2 + o(\theta). 
\end{align}
In addition, 
\begin{align}
\EE \left[e^{-2\theta Y_{i}}\right]  = 1-2\theta\EE[Y_{i}]+o(\theta).
\end{align}
Hence, \eqref{Eq_v} can be expressed as
\begin{align}\label{Eq_taylor}
\left({1}{}-{2\beta\theta}\right)\! \left[1+ \frac{2}{3} {\theta} v^2 + o(\theta)\right] = 1-2\theta\EE[Y_{i}]+o(\theta).\!\!
\end{align}
Expanding \eqref{Eq_taylor}, yields
\begin{align}\label{Eq_taylor1}
2\theta\EE[Y_{i}] - 2{\beta} \theta+ \frac{2}{3} \theta{v^2} + o(\theta) =0.
\end{align}
Dividing by $\theta$ and letting $\theta \to 0$ on both sides of \eqref{Eq_taylor1},  yields
\begin{align}\label{Eq_taylor2}
v^2-3(\beta-\EE[Y_{i}])  =0. 
\end{align}
Equation \eqref{Eq_taylor2} has two roots $ -\sqrt{3(\beta-\EE[Y_{i}])}$, and $\sqrt{3(\beta-\EE[Y_{i}])}$. If  $v_*= -\sqrt{3(\beta-\EE[Y_{i}])}$, the free boundary problem in \eqref{eq_free1}-\eqref{eq_free3} are invalid. Hence, as $\theta \to 0$ and $\sigma = 1$, the root of \eqref{eq_threshold} is $v_*= \sqrt{3(\beta-\EE[Y_{i}])}$. This completes the proof.

\section{Proof of Lemma \ref{lem_stop}}\label{app_lem4}

We first prove \eqref{eq_stop}. It is known that the OU process $O_t$ is a Feller process \cite[Section 5.5]{Liggett2010}. By using a property of Feller process in   \cite[Theorem 3.32]{Liggett2010}, we get that
\begin{align}
&O_t^2 -  \int_0^t \mathcal{G} (O_s^2 ) ds\nonumber\\
=&O_t^2 - \int_0^t (- \theta O_s 2 O_s + \sigma^2) ds\nonumber\\
=& O_t^2 - \sigma^2 t +  2 \theta \int_0^t  O_s ^2 ds
\end{align}
is a martingale, where $\mathcal{G}$ is the infinitesimal generator of $O_t$ defined in \eqref{eq_generator}.
According to \cite{Durrettbook10}, the minimum of two stopping times is a stopping time and constant times are stopping times. Hence, $t{\wedge}{\tau}$ is a bounded stopping time for every $t\in[{0,\infty})$, where $x{\wedge}y=\min\{x, y\}$. Then, by Theorem 8.5.1 of \cite{Durrettbook10}, for every $t\in[{0,\infty})$
\begin{align}\label{eq_app_lem4_1}
&\mathbb{E}\left[\int_{0}^{t{\wedge}{\tau}}{O_s}^{2}ds\right] =\mathbb{E}\bigg[\frac{{\sigma}^2}{2\theta}(t{\wedge}{\tau})\bigg]-\mathbb{E}\bigg[\frac{1}{2\theta}O_{t{\wedge}{\tau}}^2\bigg].
\end{align}
Because $\mathbb{E}\left[\int_{0}^{t{\wedge}{\tau}}O_s^{2}ds\right]$ and $\mathbb{E}[t{\wedge}{\tau}]$ are positive and increasing with respect to $t$, by using the monotone convergence theorem \cite[Theorem 1.5.5]{Durrettbook10}, we get
\begin{align}\label{eq_app_lem4_4}
\lim_{t\to\infty} \mathbb{E}\left[\int_{0}^{t{\wedge}{\tau}}O_s^2ds\right]&=\mathbb{E}\left[\int_{0}^{\tau}O_s^2ds\right],\\
\lim_{t\to\infty}\mathbb{E}[(t{\wedge}{\tau})]&=\mathbb{E}[{\tau}] . \label{eq_app_lem4_3}
\end{align}
In addition, according to \cite[Theorem 2.2]{Zhao2020}, 
\begin{align}
\EE \left[\max_{0\leq s\leq \tau}O_{s}^2\right] \leq \frac{C}{\theta} \EE\left[\log\left(1+ \frac{\theta \tau}{\sigma}\right)\right]\leq \frac{C}{\sigma}\EE\left[ \tau\right]<\infty.
\end{align}
Because $O_{t\wedge \tau}^2 \leq \sup_{0\leq s\leq \tau} O_{s}^2$ for all $t$ and $\sup_{0\leq s\leq \tau} O_{s}^2$ is integratable, by invoking the dominated convergence theorem \cite[Theorem 1.5.6]{Durrettbook10}, we have
\begin{align}\label{eq_goal}
&\lim_{t\to\infty}\mathbb{E}\left[O_{t{\wedge}{\tau}}^2\right]=\mathbb{E}\left[{O_{\tau}^2}\right].
\end{align}
Combining \eqref{eq_app_lem4_4}-\eqref{eq_goal}, \eqref{eq_stop} is proven. 


We now prove \eqref{eq_stop11} and \eqref{eq_stop12}. By using the solution of the ODE in Appendix \ref{app_eq_W}, one can show that $R_1(v)$ in \eqref{eq_R_1} is the solution to the following ODE
\begin{align}
\frac{\sigma^2}{2} R_1''(v) - {\theta} v R_1'(v) =1,
\end{align}
and $R_2(v)$ in \eqref{eq_R_2} is the solution to the following ODE
\begin{align}
\frac{\sigma^2}{2} R_2''(v) - {\theta} v R_2'(v) =v^2.
\end{align}
In addition, $R_1(v)$ and $R_2(v)$ are twice continuously differentiable.
According to Dynkin's formula in \cite[Theorem 7.4.1 and the remark afterwards]{Bernt2000}, because the initial value of $O_t$ is $O_0=0$, if $\tau$ is the first exit time of a bounded set, then
\begin{align}
\!\!\!\EE_0 [R_1(O_\tau)] &= R_1(0) + \EE_0 \left[\int_0^\tau 1 ds\right] = R_1(0) + \EE_0 [\tau],\!\!\!\\
\!\!\!\EE_0 [R_2(O_\tau)] &= R_2(0) + \EE_0 \left[\int_0^\tau O_s^2 ds\right].\!\!\!
\end{align}
Because $R_1(0)=R_2(0)  = 0$, \eqref{eq_stop11} and \eqref{eq_stop12} follow. 
%
This completes the proof.


\section{Proof of Lemma \ref{lem_zeroqueue}}\label{app_lem1}

Suppose that in the sampling policy $\pi$, sample $i$ is generated when the server is busy sending another sample, and hence sample $i$  
needs to wait for some time before being submitted to the server, i.e., $S_i<G_i$. 
Let us consider a \emph{virtual} sampling policy $\pi' = \{S_0,\ldots, S_{i-1},G_i,$ $ S_{i+1},\ldots\}$ such that the generation time of sample $i$ is postponed from $S_i$ to $G_i$. We call policy $\pi'$ a virtual policy because it may happen that  $G_i > S_{i+1}$. 
However, this will not affect our  proof below. We will show that the MSE of the sampling policy $\pi'$ is smaller than that of the sampling policy $\pi =\{S_0,\ldots, S_{i-1},S_i, S_{i+1},\ldots\}$. 

Note that $\{X_t: t\in[0, \infty)\}$ does not change according to the sampling policy, and the sample delivery times $\{D_0,D_1, D_2,\ldots\}$ remain the same in policy $\pi$ and policy $\pi'$. 
Hence, the only difference between policies $\pi$ and $\pi'$ is that \emph{the generation time of sample $i$ is postponed from $S_i$ to $G_i$}. 
The MMSE estimator under policy $\pi$ is given by \eqref{eq_esti} and the  MMSE estimator under policy $\pi'$ is given by
\begin{align}\label{eq_esti_pi1}
\hat X_t' = &\mathbb{E}[X_{t}|(S_j, X_{S_j}, D_j)_{j\leq i-1}, (G_i, X_{G_i}, D_i)]\nonumber\\
 = &\left\{\begin{array}{l l} 
\mathbb{E}[X_{t}|X_{G_i},G_i],& t\in[D_i,D_{i+1}); \\
\mathbb{E}[X_{t}|X_{S_j},S_j],& t\in[D_j,D_{j+1}),~j\neq i. \\
\end{array}\right.
\end{align}

Next, we consider a third virtual sampling policy $\pi''$ in which the samples $(X_{G_i}, G_i)$ and  $(X_{S_i}, S_i)$ are both delivered to the estimator at the same time $D_i$. Clearly, the estimator under policy $\pi''$ has more information than those under policies $\pi$ and $\pi'$. By following the arguments in Appendix \ref{app_MMSE}, one can show that the MMSE estimator under policy $\pi''$ is
\begin{align}\label{eq_esti_pi2}
\hat X_t'' =& \mathbb{E}[X_{t}|(S_j, X_{S_j}, D_j)_{j\leq i}, (G_i, X_{G_i}, D_i)] \nonumber\\
=&\left\{\begin{array}{l l} 
\mathbb{E}[X_{t}|X_{G_i},G_i],& t\in[D_i,D_{i+1}); \\
\mathbb{E}[X_{t}|X_{S_j},S_j],& t\in[D_j,D_{j+1}),~j\neq i. \\
\end{array}\right.
\end{align}
Notice that, because of the strong Markov property of OU process, the estimator under policy $\pi''$ uses the fresher sample $(X_{G_i},G_i)$, instead of the stale sample $(X_{S_i},S_i)$, to construct $\hat X_t''$ during $[D_i,D_{i+1})$. Because the estimator under policy $\pi''$ has more information than that under policy $\pi$, one can imagine that policy $\pi''$ has a smaller estimation error than policy $\pi$, i.e., 
\begin{align}\label{eq_small_error}
\mathbb{E}\left\{\int_{D_i} ^{D_{i+1}} (X_t-\hat X_{t})^2dt\right\}\geq \mathbb{E}\left\{\int_{D_i} ^{D_{i+1}} (X_t-\hat X_{t}'')^2  dt\right\}.
\end{align}
To prove \eqref{eq_small_error}, we invoke the orthogonality principle of the MMSE estimator \cite[Prop. V.C.2]{Poor:1994} under policy $\pi''$ and obtain 
\begin{align}\label{eq_orthogonality}
\mathbb{E}\left\{\int_{D_i} ^{D_{i+1}}2 (X_t-\hat X_t'') (\hat X_t''- \hat X_t) dt\right\}=0,
\end{align}
where we have used the fact that  $(X_{G_i}, G_i)$ and  $(X_{S_i}, S_i)$ are available by the MMSE estimator under policy $\pi''$.
Next, from \eqref{eq_orthogonality}, we can get
\begin{align}
&\mathbb{E}\left\{\int_{D_i} ^{D_{i+1}} (X_t-\hat X_{t})^2dt\right\} \nonumber\\
= & \mathbb{E}\left\{\int_{D_i} ^{D_{i+1}} (X_t-\hat X_t'')^2 +(\hat X_t''- \hat X_{t})^2 dt\right\} \nonumber\\
& + \mathbb{E}\left\{\int_{D_i} ^{D_{i+1}}2 (X_t-\hat X_t'') (\hat X_t''- \hat X_{t}) dt\right\} \nonumber 
\end{align}
\begin{align}
= & \mathbb{E}\left\{\int_{D_i} ^{D_{i+1}} (X_t-\hat X_t'')^2 +(\hat X_t''- \hat X_{t})^2 dt\right\} \nonumber\\
\geq &\mathbb{E}\left\{\int_{D_i} ^{D_{i+1}} (X_t-\hat X_t'')^2  dt\right\}. 
\end{align}
In other words, the estimation error of policy $\pi''$ is no greater than that of policy $\pi$. Furthermore, by comparing \eqref{eq_esti_pi1} and \eqref{eq_esti_pi2}, we can see  that the MMSE estimators under policies $\pi''$ and $\pi'$ are exact the same. Therefore, the estimation error of policy $\pi'$ is no greater than that of policy $\pi$.

By repeating the above arguments for all samples $i$ satisfying $S_i<G_i$, one can show that the sampling policy $\{S_0,G_1,\ldots, G_{i-1},G_i,$ $ G_{i+1},\ldots\}$ is better than the sampling policy $\pi =\{S_0,S_1,\ldots, S_{i-1},S_i,$ $ S_{i+1},\ldots\}$.
This completes the proof.

\section{Proof of (\ref{eq_simplification_1})}\label{app_eq_simplification_1}

According to Lemma \ref{lem_stop},
\begin{align}\label{eq_simplification}
&\mathbb{E}\left[\int_{Y_{i}+Z_i}^{Y_i+Z_i+Y_{i+1}} O_{s}^2ds\right] \nonumber\\
=& \frac{\sigma^2}{2\theta}\EE [Y_{i+1}] - \frac{1}{2\theta} \EE \left[ O_{Y_i+Z_i+Y_{i+1}}^2 - O_{Y_i+Z_i}^2\right].
\end{align}
The second term in \eqref{eq_simplification} can be expressed as 
\begin{align}
&\EE \left[O_{Y_i+Z_i+Y_{i+1}}^2 - O_{Y_i+Z_i}^2\right]\nonumber\\
=&\EE \left[ \left(O_{Y_i+Z_i} e^{-\theta Y_{i+1} } + \frac{\sigma}{\sqrt{2\theta}} e^{-{\theta}Y_{i+1}} W_{e^{2{\theta}Y_{i+1}}-1}\right)^2\right.\nonumber\\
&~~~~~~~~~~~~~~~~~~~~~~~~~~~~~~~~~~~~~~~~~~~~~~~~~-O_{Y_i+Z_i}^2\bigg]\!\!\nonumber\\
=& \EE \left[O_{Y_i+Z_i}^2 (e^{-2{\theta}Y_{i+1}} - 1) + \frac{\sigma^2}{{2\theta}} e^{-2{\theta}Y_{i+1}} W_{e^{2{\theta}Y_{i+1}}-1}^2 \right] \nonumber\\
& +\EE \left[2O_{Y_i+Z_i}e^{-{\theta}Y_{i+1}} \frac{\sigma}{\sqrt{2\theta}} e^{-{\theta}Y_{i+1}} W_{e^{2{\theta}Y_{i+1}}-1} \right].
\end{align}
Because $Y_{i+1}$ is independent of $O_{Y_i+Z_i}$ and $W_t$, we have 
\begin{align}
 &\EE \left[O_{Y_i+Z_i}^2 (e^{-2{\theta}Y_{i+1}} - 1)\right]  = \EE \left[O_{Y_i+Z_i}^2\right]\EE \left[e^{-2{\theta}Y_{i+1}} - 1\right], 
\end{align}
 and
\begin{align}
& \EE \left[2O_{Y_i+Z_i}e^{-{\theta}Y_{i+1}} \frac{\sigma}{\sqrt{2\theta}} e^{-{\theta}Y_{i+1}} W_{e^{2{\theta}Y_{i+1}}-1} \right] \nonumber\\
=& \EE \left[2O_{Y_i+Z_i}\right]\EE\left[e^{-{\theta}Y_{i+1}} \frac{\sigma}{\sqrt{2\theta}} e^{-{\theta}Y_{i+1}} W_{e^{2{\theta}Y_{i+1}}-1} \right]\nonumber\\
\overset{(a)}{=}& \EE \left[2O_{Y_i+Z_i}\right]\!\EE\!\left[\EE \!\left[e^{-{\theta}Y_{i+1}} \frac{\sigma}{\sqrt{2\theta}} e^{-{\theta}Y_{i+1}} W_{e^{2{\theta}Y_{i+1}}-1} \bigg| Y_{i+1}\right]\right]\!\!.
\end{align}
where Step (a) is  due to the law of iterated expectations.
Because $\EE[W_t]=0 $ for all constant $t\geq 0$, it holds for all realizations of $Y_{i+1}$ that
\begin{align}
\EE \left[e^{-{\theta}Y_{i+1}} \frac{\sigma}{\sqrt{2\theta}} e^{-{\theta}Y_{i+1}} W_{e^{2{\theta}Y_{i+1}}-1} \bigg| Y_{i+1}\right] =0.
\end{align}
Hence, 
\begin{align}
\EE \left[2O_{Y_i+Z_i}e^{-{\theta}Y_{i+1}} \frac{\sigma}{\sqrt{2\theta}} e^{-{\theta}Y_{i+1}} W_{e^{2{\theta}Y_{i+1}}-1} \right]  = 0.
\end{align}
In addition, 
\begin{align}
 &\EE \left[\frac{\sigma^2}{{2\theta}} e^{-2{\theta}Y_{i+1}} W_{e^{2{\theta}Y_{i+1}}-1}^2 \right] \nonumber\\
\overset{(a)}{=} &\frac{\sigma^2}{{2\theta}} \EE \left[\EE \left[e^{-2{\theta}Y_{i+1}} W_{e^{2{\theta}Y_{i+1}}-1}^2 \bigg| Y_{i+1}\right]\right] \nonumber\\
\overset{(b)}{=}&  \frac{\sigma^2}{{2\theta}}\EE \left[1-e^{-2{\theta}Y_{i+1}} \right],
\end{align}
where Step (a) is due to the law of iterated expectations and Step (b) is due to $\EE [W_t^2] = t$ for all constant $t\geq 0$. 
Hence, 
\begin{align}\label{eq_simplification11}
\!&\mathbb{E}\left[\int_{Y_{i}+Z_i}^{Y_i+Z_i+Y_{i+1}} O_{s}^2ds\right] \nonumber\\
\!=&\frac{\sigma^2}{2\theta}\EE [Y_{i+1}] + {\gamma }\EE \left[O_{Y_i+Z_i}^2\right] - \frac{\sigma^2}{4\theta^2} \EE \left[1-e^{-2{\theta}Y_{i+1}} \right]\nonumber\\
\!=&\frac{\sigma^2}{2\theta}[\EE (Y_{i+1}) - \gamma] + \EE \left[O_{Y_i+Z_i}^2\right]\gamma , 
\end{align}
where $\gamma $ is defined in \eqref{eq_gamma1}.
Using this, \eqref{eq_simplification_1} can be shown readily.

%
%
%
%
%
%
%

\section{Proof of Lemma \ref{eq_expectation}}\label{eq_expectation_1}
According to \eqref{eq_process} and \eqref{eq_esti}, the estimation error $(X_t-\hat X_t)$ is of the same distribution with $O_{t-S_i(\beta)}$ for $t\in[D_i(\beta),D_{i+1}(\beta))$. We will use $(X_t-\hat X_t)$ and $O_{t-S_i(\beta)}$ interchangeably  for $t\in[D_i(\beta),D_{i+1}(\beta))$. 
In order to prove Lemma \ref{eq_expectation}, we need to consider the following two cases:  

\emph{Case 1:} If $|X_{D_i(\beta)} - \hat X_{D_i(\beta)} |= |O_{Y_i} | \geq v(\beta)$, then \eqref{eq_opt_solution} tells us $S_{i+1}(\beta) = D_i(\beta)$. Hence, 
\begin{align}\label{eq_expectation_12}
D_{i} (\beta) &= S_{i}(\beta) + Y_{i}, \\
D_{i+1} (\beta)& = S_{i+1}(\beta) + Y_{i+1} = D_i (\beta) + Y_{i+1}.
\end{align} 
Using these and the fact that the $Y_i$'s are independent of the OU process, we can obtain
\begin{align}\label{eq_expectation_6}
\mathbb{E}\left[D_{i+1} (\beta) - D_i (\beta) \Big|O_{Y_i}, |O_{Y_i} |\geq v(\beta)\right] 
= \mathbb{E}[Y_{i+1}],
\end{align}
and 
\begin{align}\label{eq_expectation_9}
&\mathbb{E}\left[\int_{D_i(\beta)}^{D_{i+1}(\beta)}\! (X_t-\hat X_t)^2dt \bigg| O_{Y_i}, |O_{Y_i} | \geq v(\beta)\right] \nonumber\\
=& \mathbb{E} \bigg[\int_{Y_{i}}^{Y_i + Y_{i+1}} O_s^2 ds \bigg| O_{Y_i}, |O_{Y_i} | \geq v(\beta)\bigg] \nonumber\\
\overset{(a)}{=}& \frac{\sigma^2}{2\theta}\EE [Y_{i+1}] + {\gamma }O_{Y_i}^2 - \frac{\sigma^2}{4\theta^2} \EE \left[1-e^{-2{\theta}Y_{i+1}} \right] \nonumber\\
=& \mathsf{mse}_{\infty} [\mathbb{E} (Y_{i+1})- \gamma] +  O_{Y_i}^2\gamma,
\end{align}
where Step (a) follows from the proof of \eqref{eq_simplification11}.

\emph{Case 2:} If $|X_{D_i(\beta)} - \hat X_{D_i(\beta)} |= |O_{Y_i} | < v(\beta)$, then \eqref{eq_opt_solution} tells us that, almost surely,
\begin{align}
|X_{S_{i+1}(\beta)} - \hat X_{S_{i+1}(\beta)} | = v(\beta).
\end{align}
Let us consider the following equation: 
\begin{align}\label{eq_expectation_7}
&\mathbb{E}\left[D_{i+1} (\beta) - D_i (\beta) \Big|O_{Y_i}, |O_{Y_i} |< v(\beta)\right] \nonumber\\
=& \mathbb{E}\Big[(D_{i+1} (\beta) - S_{i+1} (\beta))+(S_{i+1} (\beta) - S_i (\beta))\nonumber\\
&~~~~~~~~~~~~~~~~~~- (D_i (\beta) - S_i (\beta))\left.\Big|O_{Y_i}, |O_{Y_i} |< v(\beta)\right].
\end{align}
Because $D_{i+1} (\beta) = S_{i+1}(\beta) + Y_{i+1}$, the remaining task is to find $\mathbb{E}\left[S_{i+1} (\beta) - S_i (\beta) \Big|O_{Y_i}, |O_{Y_i} |< v(\beta)\right]$, and $\mathbb{E}\left[D_{i} (\beta) - S_i (\beta) \Big| O_{Y_i}, |O_{Y_i} |< v(\beta)\right]$ to compute \eqref{eq_expectation_7}. By invoking Lemma \ref{lem_stop}, we can obtain
\begin{align}\label{eq_expectation_2}
&\mathbb{E}\left[S_{i+1} (\beta) - S_i (\beta) \Big|O_{Y_i}, |O_{Y_i} |< v(\beta)\right] = R_1 (v(\beta)), \\
&\mathbb{E}\left[D_{i} (\beta) - S_i (\beta) \Big| O_{Y_i}, |O_{Y_i} |< v(\beta)\right] = R_1 (|O_{Y_i} |),\label{eq_expectation_3}
\end{align}
Substituting \eqref{eq_expectation_2}, \eqref{eq_expectation_3}, and $D_{i+1} (\beta) = S_{i+1}(\beta) + Y_{i+1}$ in \eqref{eq_expectation_7}, we get that
\begin{align}\label{eq_expectation_15}
&\mathbb{E}\left[D_{i+1} (\beta) - D_i (\beta) \Big|O_{Y_i}, |O_{Y_i} |< v(\beta)\right] \nonumber\\
=& \mathbb{E}[Y_{i+1}] + R_1 (v(\beta)) - R_1 (|O_{Y_i} |).
\end{align}
In addition, let us consider the following equation:
\begin{align}\label{eq_expectation_101}
&\mathbb{E}\left[\int_{D_i(\beta)}^{D_{i+1}(\beta)}\! (X_t-\hat X_t)^2dt \bigg| O_{Y_i}, |O_{Y_i} | < v(\beta)\right] \nonumber\\
=& \mathbb{E}\left[\int_{S_{i+1}(\beta)}^{D_{i+1}(\beta)}\! (X_t-\hat X_t)^2dt+\int_{S_i(\beta)}^{S_{i+1}(\beta)}\! (X_t-\hat X_t)^2dt\right.\nonumber\\
&~~~~~~~~~~\left.- \int_{S_i(\beta)}^{D_{i}(\beta)}\! (X_t-\hat X_t)^2dt\bigg|O_{Y_i}, |O_{Y_i} |< v(\beta)\right]
\end{align}
Next, we need to compute the expectations in \eqref{eq_expectation_101}. By invoking Lemma \ref{lem_stop} again, we can obtain
\begin{align}\label{eq_expectation_4}
&\mathbb{E}\left[\int_{S_{i}(\beta)}^{S_{i+1}(\beta)} (X_t-\hat X_t)^2dt \bigg|O_{Y_i}, |O_{Y_i} |< v(\beta)\right] \nonumber\\
= &\mathbb{E}\left[\int_{0}^{Y_i + Z_i} O_s^2ds \bigg|O_{Y_i}, |O_{Y_i} |< v(\beta)\right] = R_2 (v(\beta)),\\
&\mathbb{E}\left[\int_{S_{i}(\beta)}^{D_{i}(\beta)} (X_t-\hat X_t)^2dt \bigg| O_{Y_i}, |O_{Y_i} |< v(\beta)\right] \nonumber\\
= &\mathbb{E}\left[\int_0^{Y_i} O_s^2ds \bigg| O_{Y_i}, |O_{Y_i} |< v(\beta)\right]= R_2 (|O_{Y_i} |).\label{eq_expectation_5}
\end{align}
By substituting \eqref{eq_expectation_4}, \eqref{eq_expectation_5}, and \eqref{eq_simplification11} in \eqref{eq_expectation_101}, we have
\begin{align}\label{eq_expectation_10}
&\mathbb{E}\left[\int_{D_i(\beta)}^{D_{i+1}(\beta)}\! (X_t-\hat X_t)^2dt \bigg| O_{Y_i}, |O_{Y_i} | < v(\beta)\right] \nonumber\\
=& \mathsf{mse}_{\infty} [\mathbb{E} (Y_{i+1}) - \gamma] +  v^2(\beta) \gamma+ R_2 (v(\beta)) - R_2 (|O_{Y_i} |).
\end{align}

By combining \eqref{eq_expectation_6} and \eqref{eq_expectation_15} of the two cases, yields
\begin{align}\label{eq_expectation_8}
&\mathbb{E}\left[D_{i+1} (\beta) - D_i (\beta) \Big|O_{Y_i}\right] \nonumber\\
= &\mathbb{E}[Y_{i+1}] + \max\{R_1 (v(\beta)) - R_1 (|O_{Y_i} |),0\}.
\end{align}
Similarly, by combining \eqref{eq_expectation_9} and \eqref{eq_expectation_10}  of the two cases, yields
\begin{align}\label{eq_expectation_11}
&\mathbb{E}\left[\int_{D_i(\beta)}^{D_{i+1}(\beta)}\! (X_t-\hat X_t)^2dt \bigg|O_{Y_i}\right] \nonumber\\
= &\mathsf{mse}_{\infty} [\mathbb{E} (Y_{i+1}) - \gamma] +  \max\{v^2(\beta),O_{Y_i}^2\} \gamma\nonumber\\
&+ \max\{R_2 (v(\beta)) - R_2 (|O_{Y_i} |),0\}.
\end{align}
Finally, by taking the expectation over $O_{Y_i}$ in \eqref{eq_expectation_8} and \eqref{eq_expectation_11} and using the fact that $R_1(\cdot)$ and $R_2(\cdot)$ are even functions, Lemma \ref{eq_expectation} is proven.

\section{Proof of Lemma \ref{thm_solution_form}}\label{app_thm_solution_form}
Because the $Y_i$'s are \emph{i.i.d.}, \eqref{eq_simplification_1} is determined by the control decision $Z_i$ and the information $(O_{Y_i},Y_i)$. Hence, $(O_{Y_i},Y_i)$ is a \emph{sufficient statistic} for determining $Z_i$ in \eqref{eq_primal}. Therefore, there exists an optimal policy $(Z_0,Z_1,\ldots)$ to \eqref{eq_primal}, in which $Z_i$ is determined based on only  $(O_{Y_i},Y_i)$. 
By this, \eqref{eq_primal} is decomposed into a sequence of per-sample MDPs, given by \eqref{eq_opt_stopping}. 
This completes the proof. 


\section{Proof of (\ref{eq_W})} \label{app_eq_W}
Define $S(v)= H'(v)$. Then, \eqref{eq_free1} becomes
\begin{align}\label{eq_app_eq_W1}
& S'(v) - \frac{2{\theta}}{\sigma^2} v S(v) = \frac{2}{\sigma^2}(v^2 -{\beta}).
\end{align}
Equation \eqref{eq_app_eq_W1} can be solved by using the integrating factor method \cite[Sec. I.5]{Amann1990}, which applies to any ODE of the form 
\begin{align}
S'(v)+ a(v) S(v) = b(v). 
\end{align}
In the case of \eqref{eq_app_eq_W1}, 
\begin{align}
a(v) = - \frac{2{\theta}}{\sigma^2} v,~b(v) = \frac{2}{\sigma^2}(v^2 -{\beta}).
\end{align}
The integrating factor of \eqref{eq_app_eq_W1} is 
\begin{align}
M(v) = e^{\int a(v) dv} =e^{-{\frac{\theta}{\sigma^2} v^2}}. 
\end{align}
Multiplying $e^{-{\frac{\theta}{\sigma^2} v^2}}$ on both sides of \eqref{eq_app_eq_W1} and transforming the left-hand side  into a total derivative, yields  
\begin{align}\label{eq_app_eq_W7}
\left[S(v)e^{-{\frac{\theta}{\sigma^2} v^2}}\right]' = b(v)e^{-{\frac{\theta}{\sigma^2} v^2}}.
\end{align}
Taking the integration on both sides of \eqref{eq_app_eq_W7}, yields
\begin{align}\label{eq_app_eq_W2}
\!\! { S(v)e^{-{\frac{\theta}{\sigma^2} v^2}}} = &\int\frac{2}{\sigma^2}  (v^2 - \beta) e^{-{\frac{\theta}{\sigma^2} v^2}} dv\nonumber\\
\!\!=& \int \frac{2}{\sigma^2} {e^{-{\frac{\theta}{\sigma^2} v^2}}} v^2 dv- \int\frac{2}{\sigma^2}{\beta} {e^{-{\frac{\theta}{\sigma^2}}v^2} } dv.
\end{align}
The indefinite integrals in \eqref{eq_app_eq_W2} are given by \cite[Sec. 15.3.1, (Eq. 36)] {JEFFREY1995}
\begin{align}
 \int\frac{2}{\sigma^2} {e^{-{\frac{\theta}{\sigma^2}} v^2} v^2} dv &= \frac{{\sqrt{\pi}}{\sigma}}{2{\theta}^{\frac{3}{2}}}\text{erf}\left(\frac{\sqrt{\theta}}{\sigma} v\right) - \frac{v}{\theta} e^{-{\frac{\theta} {{\sigma}^2}} v^2} + C_1,\\
 \int \frac{2}{\sigma^2} {\beta} {e^{-{\frac{\theta}{\sigma^2}} v^2}} dv&= \frac{{\sqrt{\pi}}{\beta}}{{\sigma}{\sqrt{\theta}}} \text{erf}\left(\frac{\sqrt{\theta}}{\sigma} v\right) +C_2,\label{eq_app_eq_W3}
\end{align}
where $\text{erf}(\cdot)$ is the error function defined in \eqref{eq_erf}. 
Combining \eqref{eq_app_eq_W2}-\eqref{eq_app_eq_W3}, results in
\begin{align}\label{eq_W_diff}
& \!\!S(v)\!= \!\bigg(\frac{{\sqrt{\pi}}{\sigma}}{2{\theta}^{\frac{3}{2}}}\! - \!\frac{{\sqrt{\pi}}{\beta}}{{\sigma}{\sqrt{\theta}}}\bigg) \text{erf}{\bigg(\frac{\sqrt{\theta}}{\sigma} v \bigg)} e^{\frac{\theta}{{\sigma}^2}v^2}\!
 -\!\frac{v}{\theta}\! + \!C_3e^{\frac{\theta}{{\sigma}^2}v^2}\!\! , 
\end{align}
where $C_3 = C_1 + C_2$.
We need to  integrate $S(v)$ in \eqref{eq_W_diff} again to get $H(v)$
\begin{align} \label{eq_integral6}
H(v) = & \int S(v)dv \nonumber\\
=& \int{\!\bigg(\frac{{\sqrt{\pi}}{\sigma}}{2{\theta}^{\frac{3}{2}}}\! - \!\frac{{\sqrt{\pi}}{\beta}}{{\sigma}{\sqrt{\theta}}}\bigg) \text{erf}{\bigg(\frac{\sqrt{\theta}}{\sigma} v \bigg)} e^{\frac{\theta}{{\sigma}^2}v^2}\! dv} - \int{\!\frac{v}{\theta}\!} dv \nonumber\\
& + \int{\!C_3e^{\frac{\theta}{{\sigma}^2}v^2} dv, \!\!}
\end{align} 
which requires the following integral \cite[Sec. 8.250 (Eq. 1,4)] {2007247}:
\begin{align} \label{eq_integral1}
& {\int \text{erf}{\bigg(\frac{\sqrt{\theta}}{\sigma} v \bigg)} e^{\frac{\theta}{{\sigma}^2}v^2}  dv} \nonumber\\
=& \frac{\sigma}{\sqrt{\theta} \sqrt{\pi}} \frac{\theta}{{\sigma}^2} {v^2} {}_2F_2\left(1,1;\frac{3}{2},2;\frac{\theta}{\sigma^2}v^2\right) +C. 
\end{align}
By using \eqref{eq_integral1}, we can compute the first integral of \eqref{eq_integral6}
\begin{align} \label{eq_integral4}
& \int \bigg(\frac{{\sqrt{\pi}}{\sigma}}{2{\theta}^{\frac{3}{2}}} - \frac{{\sqrt{\pi}}{\beta}}{{\sigma}{\sqrt{\theta}}}\bigg) \text{erf}{\bigg(\frac{\sqrt{\theta}}{\sigma} v \bigg)} e^{\frac{\theta}{{\sigma}^2}v^2}  dv \nonumber\\
= & \left({\frac{1}{2{\theta}}} - {\frac{\beta}{{\sigma}^2}}\right) 
{v^2}~ {}_2F_2\left(1,1;\frac{3}{2},2;\frac{\theta}{\sigma^2}v^2\right) +C_4. 
\end{align}
The remaining integrals in \eqref{eq_integral6} are as follows \cite[Sec. 3.478 (Eq. 3)]{2007247}
\begin{align}
&\int C_3 e^{\frac{\theta}{{\sigma}^2}v^2} dv =  C_5\text{erfi}\left(\frac{\sqrt{\theta}}{\sigma}v\right) +C_6, \label{eq_integral2} \\ 
&\int \frac{v}{\theta} dv = -\frac{v^2}{2\theta} + C_7, \label{eq_integral3}
\end{align}
where $\text{erfi}(\cdot)$ is the imaginary error function defined in \eqref{eq_erfi}. 
Hence, by substituting \eqref{eq_integral4}, \eqref{eq_integral2}, and \eqref{eq_integral3} in \eqref{eq_integral6}, $H(v)$  in \eqref{eq_W} follows. This completes the proof of \eqref{eq_W}.

\section{Proof of Lemma \ref{lem_stop1}}\label{lem_stop_11}

The proof of Lemma \ref{lem_stop1} consists of the following two cases:

\emph{Case 1:} If $|v| \geq v_*$, \eqref{eq_optimal_stopping123} implies $\tau_* = 0.$ Hence,
\begin{align}\label{lem9_1}
& \mathbb{E}_v \left[\tau_* \big||v| \geq v_* \right] = \mathbb{E}_v \bigg[\int_{0}^{\tau_*} 1 ds \bigg||v| \geq v_* \bigg] = 0,
\end{align}
and
\begin{align}
& \mathbb{E}_v \bigg[\int_{0}^{\tau_*} {V_s}^2 ds\bigg| |v| \geq v_* \bigg] = 0.
\end{align}
Because $V_0 = v$, we have
\begin{align}\label{lem9_2}
\mathbb{E}_v [V_{\tau_*}^2] = \mathbb{E}_v [V_{0}^2]  = {v}^2.
\end{align}
By combining \eqref{lem9_1}-\eqref{lem9_2}, we get
\begin{align}\label{lem9_8}
& \mathbb{E}_v \left[ -\gamma V_{\tau_*}^2 - \int_{0}^{\tau_*} ({V_s}^2 - \beta) ds\bigg||v| \geq v_*\right] = -\gamma v^2 .
\end{align}

\emph{Case 2:} If $|v| < v_*$, \eqref{eq_optimal_stopping123} tells us that, almost surely,
\begin{align}
V_{\tau_*} = v_*.
\end{align}
Similar to the proof of Lemma \ref{eq_expectation}, we can use Lemma \ref{lem_stop} to obtain
\begin{align}\label{lem9_5}
&\mathbb{E}_v \left[\tau_*\big||v| < v_*\right] \nonumber\\
= & \mathbb{E}_v \bigg[\int_{0}^{\tau_*} 1 ds \bigg||v| < v_*\bigg] \nonumber\\
= &R_1(v_*) - R_1(v) \nonumber\\
 =& \frac{v_*^2}{\sigma^2}  {}_2F_2\left(1,1;\frac{3}{2},2;\frac{\theta}{\sigma^2}v_*^2\right) - \frac{{v}^2}{\sigma^2}  {}_2F_2\left(1,1;\frac{3}{2},2;\frac{\theta}{\sigma^2}{v}^2\right)\!,\!\!
\end{align}
\begin{align}\label{lem9_6}
&\mathbb{E}_v \bigg[\int_{0}^{\tau_*} {V_s}^2 ds \bigg||v| < v_*\bigg] \nonumber\\
=& R_2(v_*) - R_2(v) \nonumber\\
=& -\frac{v_*^2}{2\theta} + \frac{v_*^2}{2\theta} ~ {}_2F_2\left(1,1;\frac{3}{2},2;\frac{\theta}{\sigma^2}v_*^2\right) \nonumber\\ 
&+ \frac{v^2}{2\theta} - \frac{v^2}{2\theta} ~ {}_2F_2\left(1,1;\frac{3}{2},2;\frac{\theta}{\sigma^2}{v}^2\right),
\end{align}
and
\begin{align}\label{lem9_7}
& \mathbb{E}_v \left[V_{\tau_*}^2\big||v| < v_*\right] = v_*^2.
\end{align}
Combining \eqref{lem9_5}-\eqref{lem9_7}, yields
\begin{align}\label{lem9_9}
& \mathbb{E}_v \left[ -\gamma V_{\tau_*}^2 - \int_{0}^{\tau_*} ({V_s}^2 - \beta) ds\bigg||v| < v_*\right] \nonumber\\
=& -\frac{v^2}{2\theta} + \left({\frac{1}{2{\theta}}} - {\frac{\beta}{{\sigma}^2}}\right) 
{}_2F_2\big(1,1;\frac{3}{2},2;\frac{\theta}{\sigma^2}v^2\big) {v^2} \nonumber\\
& + \!\! \frac{1}{2\theta}\EE \left(e^{-2\theta Y_{i}} \right)  v_*^2\!  -\! \left({\frac{1}{2{\theta}}}\! - \!{\frac{\beta}{{\sigma^2}}}\right)\!\! {}_2F_2\big(1,1;\frac{3}{2},2;\frac{\theta}{\sigma^2}v_*^2\big){v_*^2} .
\end{align}
By combining \eqref{lem9_8} and \eqref{lem9_9}, Lemma \ref{lem_stop1} is proven.

\section{Proof of Lemma \ref{lem_stop2}}\label{lem_stop_22}

The proof of Lemma \ref{lem_stop2} consists of the following two cases:

\emph{Case 1:} If $|v| \geq v_*$, 
\eqref{eq_W1} tells us that
\begin{align} \label{lem10_2}
H(v) = -\gamma v^2.
\end{align}
Hence, Lemma \ref{lem_stop2} holds in \emph{Case 1}. 

\emph{Case 2:}  $|v| < v_*$. Because $H(v)$ is an even function and $H(v) = -\gamma{v^2}$ holds at $v=\pm v_*$, to prove $H(v) \geq {-\gamma}v^2$ for $|v| < v_*$, it is sufficient to show that for all $v\in[0, v_*)$
\begin{align}\label{lem10_4}
& H'(v) < [-\gamma v^2]' = -2 \gamma v.
\end{align}
Hence, the remaining task is to prove that \eqref{lem10_4} holds for $v\in[0, v_*)$. 

After some manipulations, we can obtain from \eqref{eq_threshold11} that
\begin{align}
\left(\mathsf{mse}_{\infty}-{\beta}\right) G\left(\frac{\sqrt{\theta}}{\sigma} v_*\right) =\mathsf{mse}_{\infty} \mathbb{E} (e^{-2{\theta}{Y_i}}).
\end{align}
%
Because $G(\cdot)>0$ is an increasing function, it holds for all $v\in[0, v_*)$ that
\begin{align}
 \left(\mathsf{mse}_{\infty}-{\beta}\right) G\left(\frac{\sqrt{\theta}}{\sigma} v\right) <&\left(\mathsf{mse}_{\infty}-{\beta}\right) G\left(\frac{\sqrt{\theta}}{\sigma} v_*\right) \nonumber\\
=&\mathsf{mse}_{\infty} \mathbb{E} (e^{-2{\theta}{Y_i}}).\label{lem10_1}
 \end{align}
One can obtain \eqref{lem10_4} from \eqref{H_gradient} and \eqref{lem10_1}.
%
%
Hence, Lemma \ref{lem_stop2} also holds in \emph{Case 2}. This completes the proof. 
 
\section{Proof of Lemma \ref{lem_stop3}}\label{lem_stop_33}
We need the following lemma in the proof of Lemma \ref{lem_stop3}:
\begin{lemma}\label{lem_G_x}
$(1-2x^2) G(x) \leq 1$ for all $x\geq 0$.
\end{lemma}
\begin{proof}
Because $G(0) =1$, it suffices to show that for all $x>0$
\begin{align}\label{eq_lem_G_x}
[(1-2x^2) G(x)]' \leq 0.
\end{align}
We have
\begin{align}\label{eq_lem_G_x1}
\!\!\!\!&[(1-2x^2) G(x)]' \nonumber\\ 
\!\!\!\!= &-\frac{1}{x^2} e^{x^2} \int_{0}^{x} e^{-{t^2}} dt + \frac{1}{x} - 4{x^2}{e^{x^2}} \int_{0}^{x} e^{-{t^2}} dt - 2x.\!\!
\end{align}
Because $e^{-{t^2}}$ is decreasing on $t\in[0,\infty)$, for all $x>0$
\begin{align}
\int_{0}^{x} e^{-{t^2}} dt \geq \int_{0}^{x} e^{-{x^2}} dt = x e^{-{x^2}} . 
\end{align}
Hence,
\begin{align}\label{eq_lem_G_x2}
-\frac{1}{x^2} e^{x^2} \int_{0}^{x} e^{-{t^2}} dt + \frac{1}{x}\leq 0.
\end{align}
Substituting \eqref{eq_lem_G_x2} into \eqref{eq_lem_G_x1}, \eqref{eq_lem_G_x} follows. This completes the proof. 
\end{proof}
Now we are ready to prove Lemma \ref{lem_stop3}. 
\begin{proof}[Proof of Lemma \ref{lem_stop3}]
The function $H(v)$ is continuously differentiable on $\mathbb R$. In addition, $H''(v)$ is continuous everywhere but at $v = \pm v_*$. Since the Lebesgue measure of those time $t$ for which $V_t = \pm v_*$ is zero, the values $H''(\pm v_*)$ can be chosen in the sequel arbitrarily. 
By using  It\^{o}'s formula \cite[Theorem 7.13]{BMbook10}, we obtain that almost surely
\begin{align}\label{eq_ito}
&H(V_t) - H(v) \nonumber\\
=& \int_0^t \frac{\sigma^2}{2}\left[ H''(V_r) -\theta V_r H'(V_r) - (V_r^2 - \beta)\right] d r \nonumber\\
& +\int_0^t \sigma H'(V_r) dW_r.
\end{align}
For all $t\geq 0$ and all $v\in \mathbb{R}$, we can show that
\begin{align}
&\mathbb{E}_v\left\{\int_0^t \left[\sigma H'(V_r) \right]^2 dr \right\} <\infty.\nonumber
\end{align}
This and \cite[Theorem 7.11]{BMbook10} imply that $\int_0^t \sigma H'(V_r) dW_r$ is a martingale and 
\begin{align}\label{eq_ito_integral}
\mathbb{E}_v\left[\int_0^t \sigma H'(V_r) dW_r\right] =0,~ \forall~ t\geq0.
 \end{align}   
Hence,
\begin{align}\label{eq_ito}
\!\!&\mathbb{E}_v \left[H(V_t) - H(v)\right] \nonumber\\
\!\!= &\mathbb{E}_v \left[ \int_0^t \frac{\sigma^2}{2}\left[ H''(V_r) -\theta V_r H'(V_r) - (V_r^2 - \beta)\right] d r \right].\!\!
\end{align}

Next, we show that for all $v\in \mathbb{R}$
\begin{align}\label{eq_ito2}
\frac{\sigma^2}{2} H''(v) -\theta v H'(v)-  (v^2 - \beta)\leq 0.
\end{align}
Let us consider the following two cases: 

\emph{Case 1:} If $|v|<v_*$, then \eqref{eq_free1} implies
\begin{align}\label{eq_ito1}
\frac{\sigma^2}{2} H''(v) -\theta v H'(v)-  (v^2 - \beta) = 0.
\end{align}

\emph{Case 2:}  $|v|> v_*$. In this case, $H(v) = -\gamma v^2$. Hence,
\begin{align}\label{eq_ito3}
&\frac{\sigma^2}{2} H''(v) -\theta v H'(v) \nonumber\\
=& \frac{\sigma^2}{2} (-2\gamma) -\theta v (-2\gamma v)\nonumber\\
=& -{\sigma}^2 {\gamma} + 2{\theta}{\gamma} v^2 \nonumber\\
=&  - \mathsf{mse}_{Y_i}+\mathbb{E} [1-e^{-2\theta {Y_i}}]v^2.
\end{align}
Substituting \eqref{eq_ito3} into \eqref{eq_ito2}, yields 
\begin{align}\label{eq_ito4}
\mathbb{E} [e^{-2\theta {Y_i}}] v^2\geq     \beta - \mathsf{mse}_{Y_i}.
\end{align}
To prove \eqref{eq_ito4}, since $|v|> v_*$, it suffices to show that 
\begin{align}\label{eq_ito5}
\mathbb{E} [e^{-2\theta {Y_i}}] v_*^2\geq     \beta -   \mathsf{mse}_{Y_i},
\end{align}
which is equivalent to
\begin{align}\label{eq_ito6}
& (\mathsf{mse}_{\infty} - \mathsf{mse}_{Y_i}) \frac{v_*^2}{\mathsf{mse}_{\infty}} \geq (\mathsf{mse}_{\infty} - \mathsf{mse}_{Y_i}) - (\mathsf{mse}_{\infty} - \beta). 
\end{align}

We now prove \eqref{eq_ito6}. By Lemma \ref{lem_G_x}, we get
\begin{align}
 \left(1 - \frac{v_*^2 2\theta}{\sigma^2}\right) G\left(\frac{\sqrt{\theta}}{\sigma} v_*\right) \leq 1. \label{excessive_1}
\end{align}
Hence,
\begin{align}\label{eq_ito7}
\left(1 - \frac{v_*^2 }{\mathsf{mse}_{\infty}}\right) G\left(\frac{\sqrt{\theta}}{\sigma} v_*\right)  \leq 1.
\end{align}
By substituting \eqref{eq_threshold11} into \eqref{eq_ito7}, we obtain
\begin{align}\label{eq_ito7}
&({\mathsf{mse}_{\infty} - \mathsf{mse}_{Y_i}})\left(1 - \frac{v_*^2 }{\mathsf{mse}_{\infty}}\right) G\left(\frac{\sqrt{\theta}}{\sigma} v_*\right)  \nonumber\\
\leq & \left({\mathsf{mse}_{\infty} - {\beta}}\right) G{\bigg(\frac{\sqrt{\theta}}{\sigma} v_*\bigg)}.
\end{align}
Because $G(x)>0$ for all $x>0$,
\begin{align}
({\mathsf{mse}_{\infty} - \mathsf{mse}_{Y_i}})\left(1 - \frac{v_*^2 }{\mathsf{mse}_{\infty}}\right) \leq {\mathsf{mse}_{\infty} - {\beta}},
\end{align}
which implies \eqref{eq_ito6}. Hence, \eqref{eq_ito2} holds in both cases. Thus, $\mathbb{E}_v \left[H(V_t) - H(v)\right]\leq 0$ holds for all $t\geq 0$ and $v\in \mathbb{R}$. This completes the proof. 
\end{proof}

 \section{Proof of Theorem \ref{thm6_strong_duality}}\label{app_thm_strong_duality}

According to \cite[Prop. 6.2.5]{Bertsekas2003}, if we can find $\pi^{\star} = (Z_1,$ $Z_2,\ldots)$ and $\lambda^{\star}$ that satisfy
 the following conditions:
\begin{align}
&\pi^{\star}\in\Pi_1, \lim_{n\rightarrow \infty} \frac{1}{n} \sum_{i=0}^{n-1} \mathbb{E}\left[Y_i+Z_i\right] - \frac{1}{f_{\max}} \geq  0,\label{eq_mix_or_not0}\\
&\lambda^{\star}\geq 0,\label{eq_mix_or_not1}\\
&L(\pi^{\star};\lambda^{\star}) = \inf_{\pi\in\Pi_1}  L(\pi;\lambda^{\star}),\label{eq_mix_or_not}\\
&\lambda^\star \left\{\lim_{n\rightarrow \infty} \frac{1}{n} \sum_{i=0}^{n-1} \mathbb{E}\left[Y_i+Z_i\right] - \frac{1}{f_{\max}}\right\} = 0,\label{eq_KKT_last}
\end{align}
then $\pi^{\star}$ is an optimal solution to \eqref{eq_SD} and $\lambda^\star$ is a geometric multiplier \cite{Bertsekas2003} for  \eqref{eq_SD}. 
Further, if we can find such $\pi^{\star}$ and $\lambda^{\star}$, then the duality gap between \eqref{eq_SD} and \eqref{eq_dual} must be zero, because otherwise there is no geometric multiplier \cite[Prop. 6.2.3(b)]{Bertsekas2003}. 
The remaining task is to find $\pi^{\star}$ and $\lambda^{\star}$ that satisfy \eqref{eq_mix_or_not0}-\eqref{eq_KKT_last}.

According to Theorem \ref{thm_optimal_stopping} and Corollary \ref{coro_stop}, a solution $\pi^{\star}=(Z_0(\beta), Z_1(\beta), \ldots)$ to
\eqref{eq_mix_or_not} is given by \eqref{eq_optimal_stopping1234}, where $\beta = \mathsf{mse}_{\text{opt}}+\lambda^{\star}$. 
In addition, because the $Y_i$'s are \emph{i.i.d.}, the $Z_i(\beta)$'s in policy $\pi^{\star}$ are  \emph{i.i.d.}
Using \eqref{eq_mix_or_not0}, \eqref{eq_mix_or_not1}, and \eqref{eq_KKT_last}, the value of $\lambda^\star$ can be obtained  by considering two cases: 
If $\lambda^\star >0$, because the $Z_i(\beta)$'s are \emph{i.i.d.}, we have from \eqref{eq_KKT_last} that
\begin{align}\label{eq_KKT_2}
\lim_{n\rightarrow \infty} \frac{1}{n} 
\sum_{i=0}^{n-1} \mathbb{E}\left[Y_i+Z_i(\beta)\right] = \mathbb{E}\left[Y_i+Z_i(\beta)\right]= \frac{1}{f_{\max}}.
\end{align}
If $\lambda^\star =0$, then \eqref{eq_mix_or_not0} implies
\begin{align}\label{eq_KKT_3}
\lim_{n\rightarrow \infty} \frac{1}{n}\sum_{i=0}^{n-1} \mathbb{E}\left[Y_i+Z_i(\beta)\right] = \mathbb{E}\left[Y_i+Z_i(\beta)\right] \geq \frac{1}{f_{\max}}.
\end{align}

Next, we compute ${\mathsf{mse}}_{\text{opt}}$ and $\beta = \mathsf{mse}_{\text{opt}}+\lambda^{\star}$.  To compute ${\mathsf{mse}}_{\text{opt}}$, we substitute policy $\pi^{\star}$ into \eqref{eq_Simple}, which yields 
\begin{align}\label{eq_KKT_1}
\mathsf{mse}_{\text{opt}}=&\lim_{n\rightarrow \infty}\!\!\frac{\sum_{i=0}^{n-1}\mathbb{E}\left[\int_{Y_i}^{Y_i+Z_i(\beta)+Y_{i+1}}\! O_s^2ds\right]}{\sum_{i=0}^{n-1} \mathbb{E}\left[Y_i\!+\!Z_i(\beta)\right]} \nonumber\\
=& \frac{\mathbb{E}\left[\int_{Y_i}^{Y_i+Z_i(\beta)+Y_{i+1}}\! O_s^2ds\right]}{\mathbb{E}[Y_i+Z_i(\beta)]},
\end{align}
where in the last equation we have used that the $Z_i(\beta)$'s are \emph{i.i.d.} 
Hence, the value of $\beta = {\mathsf{mse}}_{\text{opt}} + \lambda^\star$ can be obtained by considering the following two cases:

\emph{Case 1}: If $\lambda^\star =0$, then \eqref{eq_KKT_3} and \eqref{eq_KKT_1} imply that
\begin{align}
&\mathbb{E}\left[Y_i+Z_i(\beta)\right]\geq \frac{1}{f_{\max}},\label{beta_unique_1} \\
&\beta = {\mathsf{mse}}_{\text{opt}} = \frac{\mathbb{E}\left[\int_{Y_i}^{Y_i+Z_i(\beta)+Y_{i+1}}\! O_s^2ds\right]}{ \mathbb{E}\left[Y_i\!+\!Z_i(\beta)\right]}.\label{eq_KKT_5}
\end{align}
Notice that \eqref{eq_KKT_5} can rewritten as \eqref{newton_4}, which is a fixed-point equation on $\beta$. According to Lemma \ref{beta_unique}, one root of \eqref{newton_4} is in the set $(\mathsf{mse}_{Y_i}, \mathsf{mse}_\infty)$, which is also the unique root of \eqref{thm2_eq22}; we denote this root as $\beta_1$. We choose $\pi^\star = (Z_0(\beta_1), Z_1(\beta_1)...)$, where $Z_i(\cdot)$ is given by \eqref{eq_Z_1}. In addition, $\lambda^\star$ must be  $0$  in  \emph{Case 1}. Because $\lambda^\star = \beta_1 - \mathsf{mse}_{\text{opt}}$, we get $\mathsf{mse}_{\text{opt}} = \beta_1$, which is required in \eqref{eq_KKT_5}. \emph{Case 1} occurs if the root $\beta_1$ of \eqref{eq_KKT_5} satisfies \eqref{beta_unique_1}. We note that $\beta = \mathsf{mse}_{\infty}$ is another root of \eqref{eq_KKT_5}, but we do not pick policy $\pi^\star$ based on this root. 

\emph{Case 2}: If $\lambda^\star >0$, then \eqref{eq_KKT_2} and \eqref{eq_KKT_1} imply that
\begin{align}
&\mathbb{E}\left[Y_i+Z_i(\beta)\right]= \frac{1}{f_{\max}}, \label{eq_KKT_4}\\
&\beta > {\mathsf{mse}}_{\text{opt}}  = \frac{\mathbb{E}\left[\int_{Y_i}^{Y_i+Z_i(\beta)+Y_{i+1}}\! O_s^2ds\right]}{ \mathbb{E}\left[Y_i\!+\!Z_i(\beta)\right]}. \label{eq_KKT_8}
\end{align}
When the root $\beta_1$ of \eqref{eq_KKT_5} does not satisfy \eqref{beta_unique_1}, Lemma \ref{beta_unique_2} tells us that \eqref{eq_KKT_4} has a unique root in the set $[\mathsf{mse}_{Y_i}, \mathsf{mse}_{\infty})$, which is denoted by $\beta_2$. We choose $\pi^\star = (Z_0(\beta_2), Z_1(\beta_2)...)$, where $Z_i(\cdot)$ is given by \eqref{eq_Z_1}. Further, we choose $\lambda^\star = \beta_2 - \mathsf{mse}_{\text{opt}}$.



Theorem \ref{thm_optimal_stopping}, together with the fact that $\beta_1, \beta_2 {\in} [\mathsf{mse}_{Y_i}, \mathsf{mse}_\infty)$ and the arguments above, implies that the selected ${\pi}^{\star}$ and ${\lambda}^\star$ satisfy \eqref{eq_mix_or_not0}-\eqref{eq_KKT_last}.
By \cite[Prop. 6.2.3(b)]{Bertsekas2003},  the duality gap between \eqref{eq_SD} and \eqref{eq_dual} is zero. A solution to \eqref{eq_SD} and \eqref{eq_dual} is $\pi^{\star}$. This completes the proof. 

\ifCLASSOPTIONcaptionsoff
  \newpage
\fi

\end{document}